%% file: 0_main.tex
\documentclass[final,3p,times,twocolumn,authoryear]{elsarticle}
\usepackage{graphicx}      % include this line if your document contains figures
\makeatletter
\let\old@ssect\@ssect % Store how ifacconf defines \@ssect
\makeatother

\usepackage{amssymb}
\usepackage{amsfonts}
\usepackage{mathrsfs}
\usepackage{amsmath}
\usepackage{textcomp}
\usepackage{float}
\usepackage[dvipsnames]{xcolor}
\usepackage{multirow}
\usepackage{csquotes}
\usepackage{epstopdf} 
\usepackage{gensymb}
\usepackage{breqn}
\usepackage{siunitx}
\usepackage{hyperref} 
\hypersetup{colorlinks=true,colorlinks,linkcolor={blue},citecolor={blue},urlcolor={red}} 
\makeatletter
\def\endfigure{\end@float}
\def\endtable{\end@float}
\makeatother
\usepackage[dvipsnames]{xcolor}

\usepackage{subcaption}

\usepackage{booktabs, caption, nicematrix}
\usepackage[thicklines]{cancel}
\usepackage{amsthm}
\newtheoremstyle{boldtheorem}
  {\topsep}   % Space above
  {\topsep}   % Space below
  {\normalfont}%{\bfseries} % Body font (bold)
  {}          % Indent amount (empty means no indent)
  {\bfseries} % Theorem head font (bold)
  {.}         % Punctuation after theorem head
  {.5em}      % Space after theorem head
  {}          % Theorem head specification (can be left empty, meaning `normal`)

\theoremstyle{boldtheorem}

\newtheorem{thm}{Theorem}%  meant for continuous numbers
\newtheorem{rem}{Remark}%
\newtheorem{cor}{Corollary}%
\newtheorem{lem}{Lemma}%
\newtheorem{assum}{Assumption}%

\usepackage[flushleft]{threeparttable}
\usepackage{xcolor}
\usepackage[thinlines]{easytable}

\begin{document}

\begin{frontmatter}



\cortext[cor1]{Corresponding author.}

\title{Frequency Response Analysis for Reset Control Systems: Application to Predict Precision of Motion Systems}

\author[inst]{Xinxin Zhang}\ead{X.Zhang-15@tudelft.nl} 
\author[inst]{Marcin B Kaczmarek}\ead{MBKaczmarek@tudelft.nl}
\author[inst]{S. Hassan HosseinNia\corref{cor1}}\ead{S.H.HosseinNiaKani@tudelft.nl}
\affiliation[inst]{organization={Department of Precision and Microsystems Engineering (PME), Delft University of Technology},%Department and Organization
            addressline={Mekelweg 2}, 
            city={Delft},
            postcode={2628CD}, 
            % state={State One},
            country={The Netherlands}}
% \affiliation[inst2]{organization={Department Two},%Department and Organization
%             addressline={Address Two}, 
%             city={City Two},
%             postcode={22222}, 
%             state={State Two},
%             country={Country Two}}

\begin{abstract}
The frequency response analysis describes the steady-state responses of a system to sinusoidal inputs at different frequencies, providing control engineers with an effective tool for designing control systems in the frequency domain. However, conducting this analysis for closed-loop reset systems is challenging due to system nonlinearity. This paper addresses this challenge through two key contributions. First, it introduces novel analysis methods for both open-loop and closed-loop reset control systems at steady states. These methods decompose the frequency responses of reset systems into base-linear and nonlinear components. Second, building upon this analysis, the paper develops closed-loop higher-order sinusoidal-input describing functions for reset control systems at steady states. These functions facilitate the analysis of frequency-domain properties, establish a connection between open-loop and closed-loop analysis, and enable the prediction of time-domain performance for reset systems. The accuracy and effectiveness of the proposed methods are successfully validated through simulations and experiments conducted on a reset Proportional-Integral-Derivative (PID) controlled precision motion system.
% Additionally, a new reset control structure within the PID framework is introduced, serving as a case study to demonstrate the contributions made in this work.
\end{abstract}

% \begin{figure}[h]
% 	\centering
% 	%	\missingfigure{CLCI}
% 	\includegraphics[width=\columnwidth]{figs/Graphic Abstract.pdf}
% 	\captionsetup{labelformat=empty}
% 	\caption*{Figure 0: New model for the RCS in open-loop and closed-loop.}
% 	\label{fig:graphicalabstract}
% \end{figure}

% \begin{highlights}
% \item A pulse-based model for analyzing general reset control systems (RCSs) in frequency domain is proposed.
% \item This study for the first time providing the perspective that the RCS output can be segmented into linear and non-linear elements both in open-loop and closed-loop, and further enabling the loop-shaping technique applied in the RCS analysis accurately.
% \item This study proposes an analytical method to calculate all the harmonics of RCSs  in closed-loop that can distinguish the first-order harmonic of RCSs from the DF.
% \item Sensitivity functions for reset control systems are developed.
% \end{highlights}


\end{frontmatter}

\input{1_Introduction}
\input{2_Preliminaries}
\input{3_Methodology}

\input{4_Results}
\input{5_Conclusion}
\input{6_Appendix}
% Ex: \cite{Blondeletal2008,FabricioLiang2013} 
% \appendix

\section*{CRediT authorship contribution statement}
\textbf{Xinxin Zhang}: Conceptualization, Modelling $\&$ Simulation $\&$ Conducting Experiments $\&$ Analysis, Writing $\&$ editing. \textbf{Marcin B Kaczmarek}: Discussion $\&$ Review. \textbf{S. Hassan HosseinNia}: Conceptualization, Supervision, Review $\&$ editing. 

\section*{Declaration of competing interest}
The authors declare that they have no known competing financial interests or personal relationships that could have appeared to influence the work reported in this paper.
\section*{Declaration of Generative AI and AI-assisted technologies in the writing process}
During the preparation of this work the authors used [ChatGPT] in order to improve readability and language. After using this tool, the authors reviewed and edited the content as needed and take full responsibility for the content of the publication.

\bibliographystyle{elsarticle-num-names} 
\bibliography{References}

% \bibliographystyle{elsarticle-harv} 
% \bibliography{References}

\end{document}

%% file: 1_Introduction.tex
\section{Introduction}
\label{sec:introduction}
This paper aims to develop a method for analyzing the frequency response of reset control systems. The development of the reset element starts from the Clegg integrator (CI), introduced in 1958 \cite{clegg1958nonlinear}. The CI is a linear integrator encapsulating with a reset mechanism, which enables the output of the CI to be reset to zero whenever its input crosses zero. Through the Describing Function (DF) analysis, the CI demonstrates the same gain-frequency characteristics as a linear integrator but exhibits a significant phase lead of 51.9\textdegree. This phase-frequency characteristic highlights the ability of the reset element to overcome the Bode gain-phase restriction in linear controllers \cite{chen2018beyond}. Numerous other applications of reset elements showcase their ability to enhance the system performance. These include the First-order Reset Element (FORE) \cite{krishnan1974synthesis, horowitz1975non}, the Second-order Reset Element (SORE) \cite{hazeleger2016second}, partial reset techniques \cite{beker2004fundamental}, Proportional-Integral (PI) + CI \cite{banos2007definition}, reset control systems with reset bands \cite{banos2011limit}, Fractional-order Reset Element (FrORE) \cite{saikumar2017generalized, weise2019fractional, weise2020extended}, and Constant in gain Lead in phase (CgLp) \cite{saikumar2019constant}.  

Frequency response analysis is a method used to assess the magnitude and phase properties of a control system by analyzing its steady-state responses to sinusoidal inputs across various frequencies \cite{tian2007frequency}. Engineers can shape and tune the performance of closed-loop systems based on their open-loop analysis, a technique referred to as loop shaping \cite{van2017frequency, ogata2010modern}. The frequency-domain-based loop shaping approach has proven effective for designing linear control systems, including PID controllers, in industries \cite{saikumar2019constant, deenen2017hybrid, saikumar2021loop}. Reset controllers are seamlessly integrated into the classical PID framework, thus attracting interest for their potential applications across various industries \cite{beerens2021reset}. However, the lack of effective frequency-domain analysis tools tailored for reset control systems has hindered their widespread adoption in industries.
% This analysis provides engineers with the tools to anticipate system behavior across different frequencies, facilitating the optimization of system performance. \cite{van2017frequency, ogata2010modern} In the realm of industrial applications,

The frequency response analysis includes both open-loop and closed-loop analysis. For open-loop reset controllers, the DF \cite{guo2009frequency} was first employed to analyze their frequency response, but it falls short in capturing the complete dynamics of reset control systems as it only analyzes the first harmonic of the outputs. The Higher-Order Sinusoidal Input Describing Function (HOSIDF) \cite{saikumar2021loop}, which accounts for higher-order harmonics, proves to be an effective method for analyzing open-loop reset control systems. However, in closed-loop reset control systems, the existence of high-order harmonics in the output signals leads to the generation of higher-order sub-harmonics through the feedback loop, presenting challenges for frequency response analysis. Existing tools for analyzing closed-loop reset control systems, such as pseudo-sensitivity functions in \cite{dastjerdi2022closed}, primarily rely on time-domain approaches and lack a direct connection between open-loop and closed-loop analysis of reset control systems. The current frequency-domain-based analysis method for reset control systems in \cite{saikumar2021loop} lacks precision as it overlooks certain higher-order harmonics in the closed-loop outputs. 
 
The lack of precise frequency response analysis methods for closed-loop reset systems and the disconnect between open-loop and closed-loop analysis in reset systems motivates this research. The objective of this research is to develop new frequency response analysis methods for both open-loop and closed-loop reset control systems with sinusoidal inputs. These methods aim to (1) predict the steady-state performance of the closed-loop reset control system by rectifying inaccuracies present in previous methods, and (2) establish a reliable connection between the frequency-domain analysis of the open-loop and closed-loop reset control systems.
% The ease of frequency-domain analysis for linear controllers contributes to their widespread use in industry. 
% Developing accurate frequency response analysis techniques for reset control systems is crucial to facilitate their widespread industrial use. 
% However, such techniques for reset control systems are currently lacking. 
% Therefore, the analysis and design of RC cannot be accurately incorporated into the loop-shaping technique using the DF. As a result, there still remains a bottleneck for applying the loop-shaping technique to connect the open-loop and closed-loop characteristics of reset control systems, which is one of the essential reasons that prevent reset control systems from being widely implemented in the industry.
% Thus, developing a precise analysis method for reset control systems that can effectively integrate them into the loop-shaping framework is required. 
% This study presents a pulse-based model in the frequency domain for analysing reset control systems, for the first time providing a perspective that a Single Input Single Output (SISO) reset control system can be analytically segmented into its base linear system (BLS) and a filtered-pulse signal. This model validates both in open-loop and in closed-loop. % Moreover, the new model enables the loop-shaping technique applied in the analysis of reset control systems. 
% Additionally, we developed sensitivity functions for the reset control system. The effectiveness of this method is evaluated in two reset control system examples. 

The structure of the study is organized as follows. Section \ref{sec:preliminaries} provides definitions, discusses existing frequency response analysis methods for reset control systems, and states the research problems. The main contributions are detailed in Section \ref{subsec: Open-loop Model} (for the open-loop reset system) and Section \ref{subsec: Closed-loop Model}  (for the closed-loop reset system), including:
\begin{enumerate}
    \item Theorem \ref{thm: vnl_npi} presents a new pulse-based approach for analyzing open-loop reset systems. This method decomposes the steady-state outputs of reset systems into base-linear outputs and pulse-based nonlinear signals. Building on Theorem \ref{thm: vnl_npi}, Theorem \ref{thm: Open loop model for RC} proposes an open-loop HOSIDF for the frequency response analysis of reset controllers.
    \item Building upon Theorem \ref{thm: Open loop model for RC}, Theorem \ref{thm: Pulse-based model for RCS in closed loop} introduces an analysis method for a closed-loop Single-Sinusoid-Input Single-Output (SSISO) reset control system featuring two reset instants per steady-state period (referred to as a two-reset system). This model decomposes the system's steady-state output into its base-linear and nonlinear components.
    \item Based on Theorem \ref{thm: Pulse-based model for RCS in closed loop}, a closed-loop HOSIDF is developed for the frequency response analysis of reset control systems, as detailed in Theorem \ref{thm: Method C, new HOSIDF}. This analysis connects the open-loop and closed-loop responses of reset control systems, enabling the application of loop-shaping techniques in nonlinear systems for the first time.
    % (\textcolor{red}{you need to develop the open-loop and closed-loop connection in a corollary}).
\end{enumerate}
% Section \ref{sec:spider} presents the precision positioning stage.
Section \ref{example 1} assesses the accuracy and highlights limitations of the proposed methods on a reset control system within the PID framework, demonstrated through simulations and experiments. In Section \ref{example 2}, we introduce a novel reset control structure aimed at addressing the limitations identified in Section \ref{example 1}. Simulated and experimental results on this new control system validate the effectiveness of the proposed analysis methods. Additionally, this section discusses the applicability of the methods in reset control system design. Finally, the study concludes in Section \ref{sec:conclusion}.
% In \ref{sec:results},

%% file: 2_Preliminaries.tex
\section{Background and Problem Statement}
\label{sec:preliminaries}
This section begins by offering background information on reset control systems. It then discusses existing frequency response analysis methods for reset systems, highlighting their limitations. Finally, the research problems are introduced.
\subsection{The Definition of the Reset Control System}
Figure \ref{fig1:RC system} illustrates the block diagram of the Single-Input Single-Output (SISO) reset control system. This system consists of a reset controller $\mathcal{C}$, a linear controller $\mathcal{C}_\alpha$, a plant $\mathcal{P}$, and a feedback loop. The blue lines depict the resetting mechanism activated by the reset trigger signal (which is $e(t)$ in this diagram). The signals $r(t)$, $e(t)$, $v(t)$, $u(t)$, and $y(t)$ correspond to the reference input signal, the error signal, the reset output signal, the control input signal, and the output signal, respectively. Note that in the reset system under a sinusoidal input signal $r(t) = |R|\sin(\omega t)\ (\omega \in\mathbb{R}^+)$, these signals are functions of both time $t$ and input frequency $\omega$. In the time domain, they are expressed as functions of $t$ for a specific $\omega$ for simplicity. Under the condition of the existence of steady states, in the Fourier domain, $R(\omega)$, $E(\omega)$, $V(\omega)$, $U(\omega)$, and $Y(\omega)$ represent their respective Fourier transforms.
\begin{figure}[h]
	%	\missingfigure{RCsystem}
	\centerline{\includegraphics[width=\columnwidth]{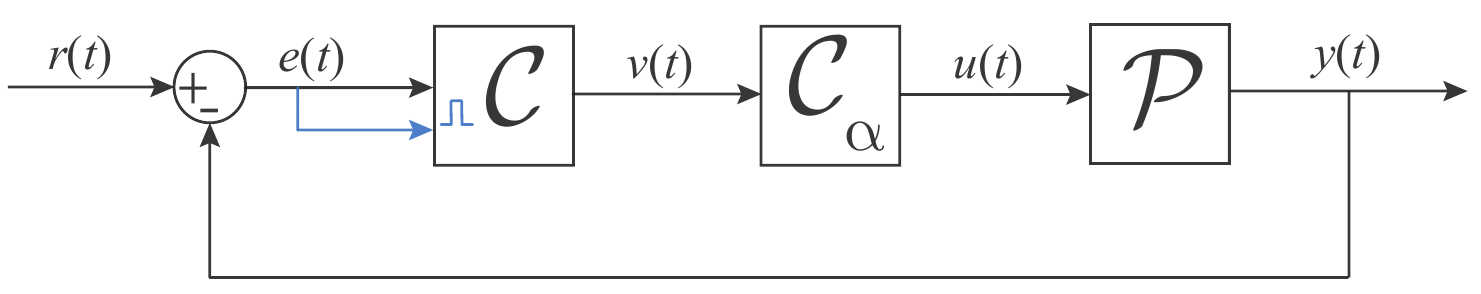}}
	\caption{The block diagram of the reset control system, where $r(t)$, $e(t)$, $v(t)$, $u(t)$, and $y(t)$ denote the reference input signal, the error signal, the reset output signal, the control input signal, and the output signal, respectively. The blue lines indicate the reset-triggered actions, with $e(t)$ serving as the reset-triggered signal in this system.}
	\label{fig1:RC system}
\end{figure}

The state-space equations of the reset controller $\mathcal{C}$ in Fig. \ref{fig1:RC system} are given by:
\begin{equation} 
	\label{eq: State-Space eq of RC} 
	\mathcal{C} = \begin{cases}
        \dot{x}_c(t) = A_Rx_c(t) + B_Re(t), &  t \notin J, \\
		x_c(t^+) = A_\rho x_c(t), & t \in J, \\
		v(t) = C_Rx_c(t) + D_Re(t).\end{cases}
\end{equation} 
In \eqref{eq: State-Space eq of RC}, $x_c(t) \in \mathbb{R}^{n_c\times 1}$ represents the state of the reset controller $\mathcal{C}$, where $n_c$ denotes the number of states. The matrices $A_R\in \mathbb{R}^{n_c\times n_c}$, $B_R\in \mathbb{R}^{n_c\times 1}$, $C_R\in \mathbb{R}^{1\times n_c}$, and $D_R\in \mathbb{R}^{1\times 1}$ describe the continuous dynamics of the base-linear controller (BLC), denoted as $\mathcal{C}_{bl}(\omega)$, given by
\begin{equation}
\label{eq: Cbl}
  \mathcal{C}_{bl}(\omega) = C_R(j\omega I-A_R)^{-1}B_R+D_R,\ (j =\sqrt{-1}).
\end{equation}
The base-linear system (BLS) of the reset control system in Fig. \ref{fig1:RC system} is the system characterized by substituting the reset controller $\mathcal{C}$ with its base-linear counterpart $\mathcal{C}_{bl}$. 

The reset controller employs the \enquote{zero-crossing law} as the reset mechanism, which enables the state $x_c(t)$ of $\mathcal{C}$ resets to a predetermined value whenever the reset-triggered signal crosses zero \cite{banos2012reset, guo2019stability}. In Fig. \ref{fig1:RC system}, the reset triggered signal is $e(t)$. The second equation in \eqref{eq: State-Space eq of RC} describes the reset action, which is an instantaneous or impulsive change of the state $(x_c(t)\to x_c(t^+))$ applied whenever $e(t_i)=0$ \cite{barreiro2014reset}. The reset instant, denoted as $t_i$, is defined as the time at which the reset condition is satisfied, i.e., $e(t_i)=0$. The set of reset instants for $\mathcal{C}$ is defined as $J := \{t_i| e(t_i)=0,\ i \in \mathbb{Z}^+\}$. The symbol $A_\rho$ represents the reset matrix, given by
\begin{equation}
\label{eq: A_rho}
%	\resizebox{.9\hsize}{!}{$
		A_\rho=\begin{bmatrix}
			\gamma & \\
			& I_{n_l}
		\end{bmatrix}\in\mathbb{R}^{n_c\times n_c},
%		$}
\end{equation}
where $\gamma = \textrm{diag}(\gamma_1, \gamma_2,\cdots, \gamma_o, \cdots, \gamma_{n_r}), \ o \in \mathbb{Z}^+$. $\gamma_o \in (-1,1)$ denotes the ratio of the after-reset state value (at $t_i^+$) to the before-reset state value (at $t_i$). The subscripts $n_r$ and $n_l$ denote the numbers of reset states and non-reset (linear) states, respectively, where $n_c = n_r + n_l$. When $\gamma = I_{n_r}$, the reset controller $\mathcal{C}$ is referred to as its BLC $\mathcal{C}_{bl}$. In this study, we specifically focus on the reset controller \(\mathcal{C}\) that employs the ``Zero-crossing Law" and involves a single reset state, where \(n_r=1\) in \eqref{eq: A_rho}. The reset controllers with \(n_r=1\) encompass common elements such as the CI, the FORE, and the higher-order reset elements like the ``Second-Order Single State Reset Element (SOSRE)" \cite{karbasizadeh2021fractional}. 

In Fig. \ref{fig1:RC system}, the linear controller $\mathcal{C}_\alpha$ combined with the plant $\mathcal{P}$ is defined as $\mathcal{P}_\alpha = \mathcal{C}_\alpha\mathcal{P}$. The state-space representation of $\mathcal{P}_\alpha$ is defined as:
\begin{equation} 
	\label{eq:P} 
	\mathcal{P}_\alpha = \begin{cases}
		\dot{x}_{\alpha}(t) = A_\alpha x_{\alpha}(t) + B_\alpha v(t),\\
		y_{\alpha}(t) = C_\alpha x_{\alpha}(t),
	\end{cases}
\end{equation} 
where $A_{\alpha}\in \mathbb{R}^{n_\alpha\times n_\alpha}$, $B_{\alpha}\in \mathbb{R}^{n_\alpha\times1}$, and $C_{\alpha}\in \mathbb{R}^{1\times n_\alpha}$ are the state-space matrices for $\mathcal{P}_\alpha$. $x_{\alpha} \in \mathbb{R}^{n_{\alpha}\times 1}$ represents the state of $\mathcal{P}_\alpha$ and $n_\alpha \in \mathbb{N}$ is the number of the state.

Combining \eqref{eq: State-Space eq of RC} and \eqref{eq:P}, the state-space representative of the reset control system without inputs is given by
\begin{equation}
	\label{eq:cl}
	\mathcal{H} = \begin{cases}
		\dot{x}(t) = A_{cl}x(t), & x \notin J_H, \\
		{x}(t^+) = A_{\rho cl}x(t), & x \in J_H, \\
		y(t) = C_{cl}x(t),
            % e(t) = r(t) - C_{cl}x(t),
	\end{cases}
\end{equation}
where $x^T = [{x_c}^T \  {x_{\alpha}}^T] \in \mathbb{R}^{n_s\times 1}$ is the state of the reset control system $\mathcal{H}$, with the number of $n_s = n_c + n_{\alpha}$. $J_H := \lbrace x \in \mathbb{R}^{n_s\times 1} \lvert C_{cl}x=0 \rbrace $ is defined to be the set of reset instants satisfying $e(t) =0$. The matrices in \eqref{eq:cl} are given by
\begin{equation}
%	\resizebox{.9\hsize}{!}{$
\begin{aligned}
A_{cl} &=
		% \begin{bmatrix}
		% 	A_R & -B_{R}C_\alpha\\
		% 	 B_\alpha C_{R} &  (A_{\alpha}-B_\alpha D_R C_\alpha)\\
		% \end{bmatrix}\in\mathbb{R}^{n_s\times n_s},\ 
            \begin{bmatrix}
			A_R & -B_{R}C_\alpha\\
			 B_\alpha C_{R} &  A_{\alpha}\\
		\end{bmatrix}\in\mathbb{R}^{n_s\times n_s},\\
  %           B_{cl} =
		% \begin{bmatrix}
		% 	B_{R}\\
  %               B_{\alpha}D_R\\
		% \end{bmatrix}\in\mathbb{R}^{n_s\times 1},\ 
		C_{cl} &= 
		\begin{bmatrix}
			 0^{1\times n_c} & C_{\alpha}              
		\end{bmatrix}\in\mathbb{R}^{1\times n_s},\\
A_{\rho cl} &= 
		\begin{bmatrix}
			A_\rho & 0 \\
			0 & I_{n_\alpha}
		\end{bmatrix}\in\mathbb{R}^{n_s\times n_s}.
\end{aligned}
%	$}
\end{equation}
\subsection{Current Frequency Response Analysis of the Open-loop Reset System}
The stability and convergence of systems are crucial for achieving a steady-state solution and enabling frequency response to sinusoidal inputs \cite{pavlov2006uniform, pavlov2007frequency}. An open-loop reset system \eqref{eq: State-Space eq of RC} with input signal of $e(t) = |E|\sin(\omega t + \angle E)$ has a globally asymptotically stable $2\pi/\omega$-periodic solution if and only if (\cite{guo2009frequency})
\begin{equation}
\label{eq:open-loop stability}
    |\lambda (D_Re^{A_R}\delta)|<1,\ \forall \delta \in \mathbb{R}^+.
\end{equation} 
Thus, we have the following assumption:
\begin{assum}
The reset system \eqref{eq: State-Space eq of RC} with input $e(t) = |E|\sin(\omega t + \angle E)$ is assumed to meet the condition in \eqref{eq:open-loop stability}.
\label{open-loop stability}   
\end{assum}
Regarding the Zeno-ness problem in the hybrid system, the outputs for the reset system are Zeno-free if the reset time interval $\sigma_i = t_{i+1}-t_i,\ i\in \mathbb{Z}^+$, between any two consecutive reset instants $(t_i,\ t_{i+1})$, is lower bounded.
\begin{equation}
\label{eq: zeno-free}
 \sigma_i>\sigma_{\text{min}},   
\end{equation}
at least in some working domain $\Omega$ \cite{barreiro2014reset}. 

For an open-loop reset controller $\mathcal{C}$ \eqref{eq: State-Space eq of RC}, with the input signal and the reset triggered signal of $e(t) = |E|\sin(\omega t)$ and satisfying Assumption \ref{open-loop stability}, there exist $n\in\mathbb{N}$ harmonics in the reset output signal $v(t)$. Utilizing the \enquote{Virtual Harmonic Generator} \cite{heinen2018frequency}, the input signal $e(t)$ generates $n$ harmonics $e_{1n}(t) = |E|\sin(n\omega t)$. The function $H_n(\omega)$, where $n$ denotes the number of harmonics involved, is defined to represent the transfer function from $e_{1n}(t)$ to the $n$-th harmonic in $v(t)$ at steady states. The expression for $H_n(\omega)$ is provided by \cite{heinen2018frequency, saikumar2021loop}:
% This function, grounded in the \enquote{Virtual Harmonics Generator} concept, \cite{nuij2006higher, nuij2008measuring, ucun2014hosidf, heinen2018frequency}
\begin{equation}
	\label{eq: HOSIDF}
	\resizebox{1\hsize}{!}{$
\begin{aligned}
    H_n(\omega)=
    \begin{cases}
			C_R(j\omega I-A_R)^{-1}(I+j\Theta _D(\omega))B_R+D_R, & \text{for}\ n=1,\\
			C_R(jn\omega I-A_R)^{-1}j\Theta _D(\omega)B_R, & \text{for odd} \ n > 1,\\
			0, & \text{for even } n \geqslant 2,
   \end{cases}
\end{aligned}
		$}
\end{equation}
with\\
\begin{equation}
	\begin{aligned}
		\Lambda(\omega) &= \omega ^2I+{A_R}^2,\\
		\Delta(\omega) &= I+e^{(\frac{\pi}{\omega}A_R)},\\
		\Delta _r(\omega) &= I+A_{\rho}e^{(\frac{\pi}{\omega}A_R)},\\
		\Gamma _r(\omega) &= \Delta _r^{-1}(\omega)A_{\rho}\Delta(\omega)\Lambda ^ {-1}(\omega),\\
		\Theta _D(\omega) &= \frac{-2\omega ^2}{\pi}\Delta(\omega)[\Gamma _r(\omega)-\Lambda ^ {-1}(\omega)].
	\end{aligned} 	
\end{equation}
Note that the expression for the first-order harmonic $H_1(\omega)$ aligns with the classical DF representation for the reset controller in \cite{guo2009frequency}. 
% \subsection{The Frequency Response Analysis of the Closed-loop Reset System}
\subsection{The Stability and Convergence Conditions of the Closed-loop Reset System}
For the frequency response analysis of closed-loop reset control systems, stability and convergence conditions are also necessary. The closed-loop reset control system \eqref{eq:cl} is quadratically stable if and only if it satisfies the well-known $H_{\beta}$ condition \cite{beker2004fundamental, carrasco2008stability}, i.e., there exists a $\beta \in \mathbb{R}^{n_r\times 1}$ and a positive definite matrix $P_{n_r} \in \mathbb{R}^{{n_r} \times {n_r}}$ such that the transfer function
\begin{equation}
\label{Hbeta1}
H_{\beta}(s) \overset{\Delta}{=} 
\begin{bmatrix}
P_{n_r}  &	0_{n_r \times n_l} & \beta C_{\alpha}
\end{bmatrix}(sI-A_{cl})^{-1}
\begin{bmatrix}
I_{n_r}\\
0_{n_l \times n_r}\\
0_{n_\alpha \times n_r}
\end{bmatrix}
\end{equation}\\
is strictly positive real and additionally a non-zero reset matrix $A_{\rho r}$ satisfies the condition	
\begin{equation}
\label{Hbeta2}
A^T_{\rho r}P_{n_r}A_{\rho r} - P_{n_r} \leqslant 0,
\end{equation}
where $I_{n_r}$ is an identity matrix of size $n_r \times n_r$.

% If the $H_\beta$ condition holds, the reset control system \eqref{eq:cl} has the uniform bounded-input bounded-state (UBIBS) property and the reset instants have the well-posedness property. \cite{dastjerdi2023frequency}. 
Literature \cite{dastjerdi2022closed} proposed that the closed-loop reset control system \eqref{eq:cl} is uniformly exponentially convergent if the initial condition of the reset controller is zero, there are infinitely many reset instants $t_i$ with $\lim_{t_i \to \infty} = \infty$, the input signal $r(t)$ is a Bohl function \cite{barabanov2001bohl}, and the $H_\beta$ condition is satisfied \cite{dastjerdi2022closed, hollot2001establishing}. Additionally, when these convergent conditions are met, the closed-loop reset control system  \eqref{eq:cl} under sinusoidal input signal $r(t) = |R|\sin(\omega t)$ exhibits a periodic steady-state solution. This solution can be represented as $x(t) = \mathcal{S}(\sin(\omega t), \cos(\omega t),\omega)$ for some function $\mathcal{S}: \mathbb{R}^3\to \mathbb{R}^{n_c+n_\alpha}$ \cite{dastjerdi2022closed}.

Considering the stability, convergence, and the existence of steady-state is needed for the frequency response analysis for the closed-loop system, we have the following Assumption \ref{assum: stable}.
\begin{assum}
\label{assum: stable}
The closed-loop reset control system \eqref{eq:cl} is assumed to satisfy the following conditions: the initial condition of the reset controller $\mathcal{C}$ is zero, there are infinitely many reset instants $t_i$ with $\lim_{t_i \to \infty} = \infty$, the input signal is a Bohl function, \cite{barabanov2001bohl}, there is no Zenoness behaviour, and the $H_\beta$ condition (in \eqref{Hbeta1} and \eqref{Hbeta2}) is satisfied.
\end{assum}
Note that Assumption \ref{assum: stable} can be met through appropriate design considerations, see \cite{saikumar2021loop, banos2012reset, samad2019ifac}.

\subsection{Problem Statement}
\label{subsec: PS}
% Under Assumption \ref{assum: stable},  In the first steady-state period, the time instant marked by $r(t_0) = 0 \ \& \ \dot r(t_0) > 0$ is designated as the initiation time $t_0=0$.
Under Assumption \ref{assum: stable}, in a SISO closed-loop reset control system with a sinusoidal reference input signal $r(t) = |R|\sin(\omega t)$ (in Fig. \ref{fig1:RC system}), the steady-state signals $e(t)$, $v(t)$, $u(t)$, and $y(t)$ are periodic and nonlinear. The Fourier transformations of these signals are define as $E(\omega)$, $V(\omega)$, $U(\omega)$, and $Y(\omega)$, respectively. These signals encompass infinite harmonics and share the fundamental frequency of $r(t)$ \cite{pavlov2006uniform}. They can be expressed as follows:
\begin{equation}
\label{eq: e,y,u}
    \begin{aligned}
        e(t) &= \sum\nolimits_{n=1}^{\infty}e_n(t) = \sum\nolimits_{n=1}^{\infty} |E_n|\sin(n\omega t+\angle E_n),\\
        v(t) &= \sum\nolimits_{n=1}^{\infty}v_n(t) = \sum\nolimits_{n=1}^{\infty} |V_n|\sin(n\omega t+\angle V_n),\\
        u(t) &= \sum\nolimits_{n=1}^{\infty}u_n(t) = \sum\nolimits_{n=1}^{\infty} |U_n|\sin(n\omega t+\angle U_n),\\
        y(t) &= \sum\nolimits_{n=1}^{\infty}y_n(t) = \sum\nolimits_{n=1}^{\infty} |Y_n|\sin(n\omega t+\angle Y_n), 
    \end{aligned}
\end{equation}
where $\angle E_n$, $\angle V_n$, $\angle U_n$, $\angle Y_n \in (-\pi, \pi]$. The signals $e_n(t)$, $v_n(t)$, $u_n(t)$ and $y_n(t)$ are the n$^{\text{th}}$ harmonics of $e(t)$, $v(t)$, $u(t)$ and $y(t)$, respectively. The Fourier transformations of these signals are define as $E_{n}(\omega)$, $V_{n}(\omega))$, $U_{n}(\omega)$, and $Y_{n}(\omega)$, respectively.

% Conducting a frequency response analysis, which involves analyzing the gain-phase responses of closed-loop reset control systems to inputs across different frequencies, is an effective tool for understanding the system's dynamics in practice \cite{saikumar2021loop}. Additionally, the frequency response analysis can be used to predict the output signals in \eqref{eq: e,y,u} of the closed-loop RCS, which is crucial for evaluating the system's performance. 
Conducting a frequency response analysis serves as an effective tool not only for understanding the frequency-domain dynamics of the system \cite{saikumar2021loop}, but also for predicting the performance of the closed-loop reset control system. For instance, the steady-state error $e(t)$ in \eqref{eq: e,y,u} serves as an important metric in assessing the tracking precision of the system. Currently, two frequency response analysis methods for closed-loop reset control systems are available, denoted as Method A and Method B.
% However, conducting the frequency response analysis for closed-loop reset control systems is challenging due to the presence of infinite harmonics in the system outputs.

Method A (\cite{guo2009frequency}): For a SISO reset system with a reference input signal $r(t) = \left|R\right|\sin(\omega t)$ under Assumption \ref{assum: stable}, as shown in Fig. \ref{fig1:RC system}, the sensitivity function $\mathcal{S}_{DF}(\omega)$, complementary sensitivity function $\mathcal{T}_{DF}(\omega)$, and control sensitivity function $\mathcal{CS}_{DF}(\omega)$ based on the steady-state DF analysis are defined as
    \begin{equation}
    \begin{aligned}
    \mathcal{S}_{DF}(\omega) &= \frac{E(\omega)}{R(\omega)} = \frac{1}{1+H_1(\omega)\mathcal{C}_\alpha(\omega)\mathcal{P}(\omega)},\\
    \mathcal{T}_{DF}(\omega) &= \frac{Y(\omega)}{R(\omega)} = \frac{H_1(\omega)\mathcal{P}(\omega)}{1+H_1(\omega)\mathcal{C}_\alpha(\omega)\mathcal{P}(\omega)},\\
    \mathcal{CS}_{DF}(\omega) &= \frac{U(\omega)}{R(\omega)}= \frac{H_1(\omega)\mathcal{C}_\alpha(\omega)}{1+H_1(\omega)\mathcal{C}_\alpha(\omega)\mathcal{P}(\omega)},
    \end{aligned}
    \end{equation}
where $H_1(\omega)$ represents the first-order harmonic transfer function of $\mathcal{C}$, as defined in \eqref{eq: HOSIDF}. However, Method A is inaccurate for predicting the performance of closed-loop reset control systems since it only considers the first-order harmonic of the reset control system, thus being valid only when $e_n(t) = 0$ for $n > 1$ in \eqref{eq: e,y,u}. In contrast, the following Method B incorporates the higher-order harmonics, which is more accurate.

Method B (\cite{saikumar2021loop}): For a SISO reset control system  in Fig. \ref{fig1:RC system} with a reference input signal $r(t) = |R|\sin(\omega t)$, under three assumptions: 1) Assumption \ref{assum: stable}, 2) the reset triggered signal is $e_1(t)$ which results in two reset instants occurring $\pi/\omega$ apart per cycle, and 3) the error $e_n(t)$ for $n>1$ does not undergo reset actions, the $n^{\text{th}}\ (n\in \mathbb{N})$ steady-state sensitivity function, complementary sensitivity function, and control sensitivity function denoted as $\mathcal{S}_n(\omega)$, $\mathcal{T}_n(\omega)$, and $\mathcal{CS}_n(\omega)$ are given by
    \begin{equation}
    \resizebox{1\columnwidth}{!}{$
    \begin{aligned}
        \mathcal{S}_n(\omega) &= \frac{E_n(n\omega)}{R_n(\omega)}\\
        &=\begin{cases}
            \mathcal{S}_{l1}(\omega), & \text{for}\ n=1,\\
            -\mathcal{L}_{n}(\omega)\mathcal{S}_{bl}(n\omega)(\left|\mathcal{S}_{l1}(\omega)\right|\angle(n\angle\mathcal{S}_{l1}(\omega))), & \text{for odd} \ n > 2,\\
            0, & \text{for even} \ n \geqslant 2,\\
        \end{cases}\\
        \mathcal{T}_n(\omega) &= \frac{Y_n(n\omega)}{R_n(\omega)}\\
        &=\begin{cases}
            1-\mathcal{S}_{l1}(\omega), & \text{for}\ n=1,\\
            \mathcal{L}_{n}(\omega)\mathcal{S}_{bl}(n\omega)(\left|\mathcal{S}_{l1}(\omega)\right|\angle(n\angle\mathcal{S}_{l1}(\omega))), & \text{for odd} \ n > 2,\\
            0, & \text{for even} \ n \geqslant 2,\\
        \end{cases}\\
        \mathcal{CS}_n(\omega) &= \frac{U_n(n\omega)}{R_n(\omega)}= \mathcal{T}_n(\omega)/\mathcal{P}(n\omega),
    \end{aligned}
    $}
    \end{equation}
    where
    \begin{equation}
        \begin{aligned}
            \mathcal{S}_{ln}(\omega) &= {1}/{(1+\mathcal{L}_{n}(\omega))},\\
            R_n(\omega) &= |R|\mathscr{F}[\sin(n\omega t)],\\
            \mathcal{S}_{bl}(n\omega) &= {1}/{(1+\mathcal{L}_{bl}(n\omega))}, \\ 
            \mathcal{L}_{n}(\omega) &= H_n(\omega)\mathcal{C}_\alpha(n\omega)\mathcal{P}(n\omega),\\
            \mathcal{L}_{bl}(n\omega) &= \mathcal{C}_{bl}(n\omega)\mathcal{C}_\alpha(n\omega)\mathcal{P}(n\omega).
        \end{aligned}
    \end{equation}
% Due to the nonlinearity, it is difficult to analyse reset control systems accurately in the closed loop, especially in the frequency domain.

Note that the first-order harmonic in Method B is identical to Method A. 
% In closed-loop reset control systems, the error signal \(e(t)\) incorporates numerous harmonics \(e_n(t)\), as illustrated in \eqref{eq: e,y,u}. Each harmonic \(e_n(t)\) contributes to the generation of reset signals. 
Although Method B takes the higher-order harmonics (for $n > 1$) into consideration, the Assumption 3) overlooks the higher-order harmonics generated by \(e_n(t) (n>1)\) in \eqref{eq: e,y,u}. This oversight will result in analysis inaccuracies, motivating the contributions made in this paper.

The main idea of this paper is to introduce a new method for analyzing the frequency response of closed-loop SISO reset control systems. It begins with open-loop analysis of reset controllers. The open-loop analysis decomposes the steady-state outputs of reset controllers into base-linear and pulse-based nonlinear elements. This decomposition approach facilitates the development of frequency responses for closed-loop reset control systems, correcting inaccurate assumptions Assumption 3) in Method B \cite{saikumar2021loop}. The proposed methods serve as tools for conducting frequency-domain analysis and predicting time-domain performance for reset systems.

%% file: 3_Methodology.tex
\section{The Frequency Response Analysis for the Open-loop Reset System}%
\label{subsec: Open-loop Model}
%\subsection{Model for Clegg Integrator}%
%\textcolor{red}{Before all the theorems, I need to bring the convergence proof in Ali's paper.}\\
This section introduces a pulse-based analysis model for the open-loop reset controller at syeady states. The model separates the output of a SSISO reset controller into its base-linear sinusoidal output and a filtered pulse signal. The term \enquote{filtered pulse signal} refers to a signal obtained by filtering a normolized pulse signal through a finite-dimensional transfer function.

The reset controller $\mathcal{C}$ \eqref{eq: State-Space eq of RC} is built on the primitive CI \cite{banos2012reset}. To establish the analysis model for the general reset controller, we begin with analyzing the CI. The Generalized CI (GCI) is defined as the reset controller $\mathcal{C}$ \eqref{eq: State-Space eq of RC} with $A_R=0$, $B_R = 1$, $C_R = 1$, $D_R = 0$, and $A_\rho = \gamma \in(-1,1)$, under Assumption \ref{open-loop stability} and the Zeno-free condition in \eqref{eq: zeno-free}. Lemma \ref{lem: CI MODEL} illustrates that the GCI's output comprises the summation of its base-linear output and a square wave component.
% Note that uppercase letters are used to indicate the frequency-domain components, while lowercases denote time-domain ones as per convention.
\begin{lem}(The pulse-based model for the open-loop GCI)
	\label{lem: CI MODEL} 
For a GCI subjected to a sinusoidal input signal $e(t) = |E_1|\sin (\omega t)$, its steady-state output signal denoted by $u_{ci}(t)$ consists of two components: one is its base-linear output $u_{i}(t)$ and another is a square wave represented as $q_i(t)$, expressed by:
	\begin{equation}
		\label{eq: ci to square wave}
		u_{ci}(t) = u_{i}(t) + q_i(t),
	\end{equation}
    where $u_i(t) = -|E_1/\omega|[\cos(\omega t)-1]$, and $q_i(t)$ is a $2\pi/\omega$-periodical square wave given by
	\begin{equation}
	\label{eq: q(t)0}
	q_i(t) = \begin{cases}
		-{2|E_1|\gamma (\gamma + 1)^{-1}}/{\omega}, & \text{for}\ t \in [2k, 2k+1)\cdot\pi/\omega,\\
		-{2|E_1| (\gamma + 1)^{-1}}/{\omega}, &  \text{for}\ t \in [2k+1, 2k+2)\cdot\pi/\omega,
	\end{cases}
	\end{equation} 		
 where $k\in\mathbb{N}$.
\end{lem}
\begin{proof} 
The proof can be found in \ref{appen: proof for lemma1}.
\end{proof}

Figure \ref{fig: ci_qi_example} displays the simulation results of an open-loop CI with the input signal $e(t)=\sin(\omega t)$ ($\omega = \pi$ rad/s and $\gamma = 0$). In this case, $q_i(t)$ is a square wave with a period of \SI{2}{\second} and amplitudes of 0 and 0.64, as calculated by \eqref{eq: q(t)0}.
\begin{figure}[htp]
	\centering
	% \missingfigure{Block diagram of FORE}
	\centerline{\includegraphics[width=0.92\columnwidth]{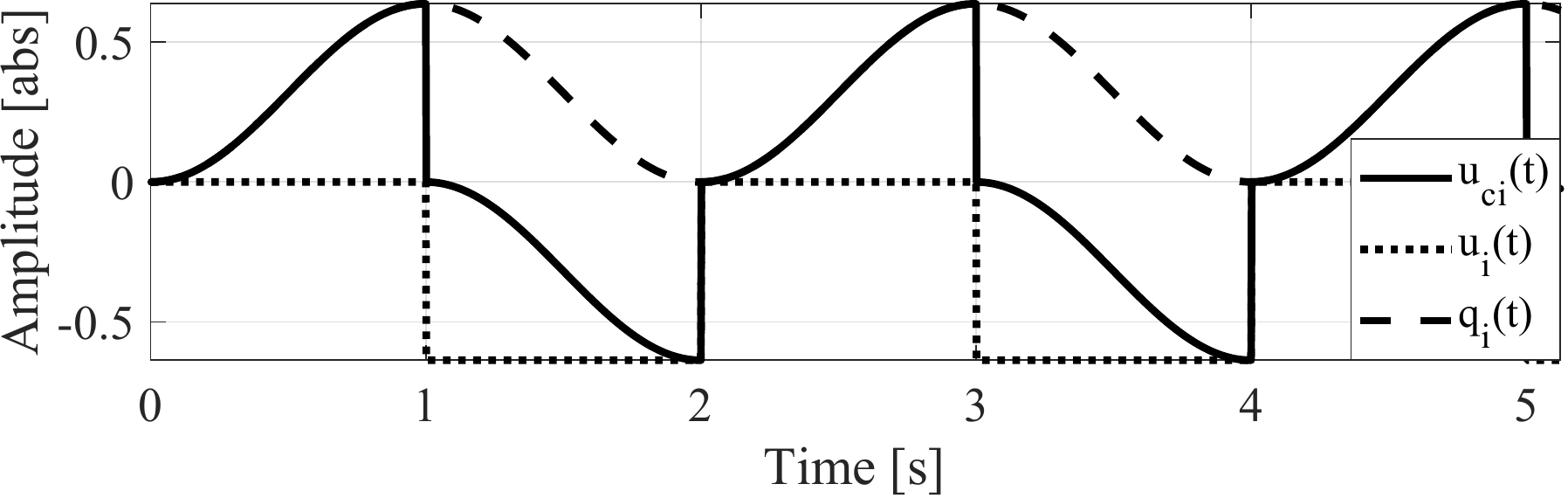}}
	\caption{$u_{ci}(t)$ (solid line), $u_i(t)$ (dotted line), and $q_i(t)$ (dashed line) of open-loop CI.}
	\label{fig: ci_qi_example}
\end{figure}

Theorem \ref{thm: vnl_npi} extends the pulse-based analysis model from the GCI to a open-loop reset controller $\mathcal{C}$ \eqref{eq: State-Space eq of RC} under different input and reset-triggered signals.
\begin{thm} (The pulse-based analysis model for the open-loop reset controller)
\label{thm: vnl_npi}
Consider a reset controller $\mathcal{C}$ described by \eqref{eq: State-Space eq of RC}, where $n_r=1$, subject to a sinusoidal input signal $e(t) = |E_n|\sin(n\omega t + \angle E_n)$ with $\angle E_n \in (-\pi,\pi]$ and $n = 2k+1$ for $k \in \mathbb{N}$, along with a $2\pi/\omega$-periodic reset triggered signal denoted by $e_s(t) = |E_s|\sin(\omega t + \angle E_s)$, where $\angle E_s \in (-\pi,\pi]$, under Assumption \ref{open-loop stability} and adhere to the Zeno-free condition outlined in \eqref{eq: zeno-free}. The input signal $e(t) = |E_n|\sin(n\omega t + \angle E_n)$ varies as a function of time $t$, while the parameters $\omega$ and $n$ remain constant. The steady-state reset output $v(t)$ is expressed as:
\begin{equation}
    v(t) = v_{bl}(t) + v_{nl}(t),
\end{equation}
where $v_{bl}(t)$ is the steady-state base-linear output given by 
\begin{equation}
\label{eq:cor:vbln}
    v_{bl}(t) = |E_n\mathcal{C}_{bl}(n\omega)| \sin(n\omega t + \angle E_n + \angle \mathcal{C}_{bl}(n\omega)).
\end{equation}
The nonlinear signal $v_{nl}(t)$ shares the same phase and period as the reset-triggered signal $e_s(t)$, as obtained by
\begin{equation}
\label{eq:Vnl(t)npi1}
\begin{aligned}
v_{nl}(t) &= \sum\nolimits_{\mu=1}^{\infty}\mathscr{F}^{-1}[\Delta_x(\mu\omega)Q^\mu(\omega)],\ \mu=2k+1,\ k\in\mathbb{N},
\end{aligned}
\end{equation}
where 
\begin{equation}
\label{eq:Vnl(t)npi2}
\begin{aligned}
\Delta_l(n\omega) &= (jn\omega I-A_R)^{-1}B_R,\\
\Delta_x(\mu\omega) &= C_R(j\mu\omega I-A_R)^{-1}j\mu\omega I,\\
\Delta_c^n(\omega) &= |\Delta_l(n\omega)|\sin(\angle\Delta_l(n\omega)+ \angle E_n-n\angle E_s),\\   
Q^\mu(\omega) &= 2|E_n|\Delta_q^n(\omega) \mathscr{F}[\sin(\mu\omega t + \mu\angle E_s)]/(\mu\pi),\\
\Delta_q^n(\omega) &= (I+e^{A_R\pi/\omega})(A_\rho e^{A_R\pi/\omega}-I)^{-1}(I-A_\rho)\Delta_c^n(\omega).
\end{aligned}
\end{equation}
% The amplitude of the nonlinear element $v_{nl}(t)$ is proportional to $|E_n|\Delta_c^n(\omega)$.
\end{thm}
\begin{proof}
The proof is provided in \ref{appen: proof for thm1}.
\end{proof}

The reset controller $\mathcal{C}$ in Fig. \ref{fig1:RC system} with the same input signal and reset-triggered signal $e(t) = |E_1|\sin(\omega t+\angle E_1)$, corresponds to the reset controller discussed in Theorem \ref{thm: vnl_npi} when $e_s(t)=e(t)$. Based on Theorem \ref{thm: vnl_npi}, Theorem \ref{thm: Open loop model for RC} presents the Higher-order Sinusoidal Input Describing Function (HOSIDF) for the open-loop reset controller $\mathcal{C}$ in Fig. \ref{fig1:RC system}. The block diagram for $\mathcal{C}$ based on the HOSIDF analysis is depicted in Fig. \ref{fig: State Space of Reset Controller}.
\begin{thm}(The HOSIDF for the open-loop reset controller) 
	\label{thm: Open loop model for RC}
Consider a reset controller $\mathcal{C}$ \eqref{eq: State-Space eq of RC} with one reset state where $n_r = 1$ in response to the input signal and reset-triggered signal $e(t) = |E_1|\sin(\omega t+\angle E_1),\ (\angle E_1\in(-\pi,\pi])$, operating under Assumption \ref{open-loop stability} and satisfying the Zeno-free condition in \eqref{eq: zeno-free}. The steady-state output signal $v(t)$ comprises $n\in\mathbb{N}$ harmonics, expressed as $v(t) = \sum\nolimits_{n=1}^{\infty} v_n(t)$, with the Fourier transform defined as $V(\omega) = \sum\nolimits_{n=1}^{\infty} V_n(\omega)$. In Fig. \ref{fig: State Space of Reset Controller}, the application of the \enquote{Virtual Harmonic Generator} \cite{saikumar2021loop, nuij2006higher} introduces $e_{1n}(t) = |E_1|\sin(n\omega t+n\angle E_1)$, with Fourier transform denoted as $E_{1n}(\omega)$. The $n$-th steady-state transfer function of $\mathcal{C}$, denoted as $\mathcal{C}_n(\omega)$, represents the ratio of $V_n(\omega)$ to $E_{1n}(\omega)$ and is given by:
\begin{equation}
\label{eq: u_ol(t)}
    \mathcal{C}_n(\omega) = \frac{V_n(\omega)}{E_{1n}(\omega)}=
    \begin{cases}
    \mathcal{C}_{bl}(\omega) + \mathcal{C}_{nl}(\omega),\ &\text{for }n=1,\\
    \mathcal{C}_{nl}(n\omega),\ &\text{for odd } n>1,\\    
    0,\ &\text{for even } n\geq2.
    \end{cases}
\end{equation}
where $\mathcal{C}_{bl}(\omega)$ is the base-linear transfer function given in \eqref{eq: Cbl} and $\mathcal{C}_{nl}(n\omega)$ is derived by
\begin{equation}
\label{Cnl_final}
\begin{aligned}
 \Delta_l(\omega) &=  (j\omega I-A_R)^{-1}B_R,\\
 \Delta_q(\omega) &=  (I+e^{A_R \pi/\omega})\Delta_v(\omega),\\
 \Delta_c(\omega)&=\left|\Delta_l(\omega)\right| \sin (\angle \Delta_l(\omega)),\\
 \mathcal{C}_{nl}(n\omega)&= {2}\Delta_x (n\omega)\Delta_q(\omega) /{(n\pi)},\\
 \Delta_x (n\omega) &= C_R(jn\omega I-A_R)^{-1}jn\omega I,\\ 
 \Delta_v(\omega) &= (A_\rho e^{A_R\pi/\omega}-I)^{-1}(I-A_\rho)\Delta_c(\omega).
 % , (T_R = [1, 0,...,0]^T \in \mathbb{N}^{n_c\times 1}).
\end{aligned}	
\end{equation}
\end{thm}
\begin{figure}[htp]
	\centering
	%	\missingfigure{Block diagram of FORE}
	\centerline{\includegraphics[width=0.92\columnwidth]{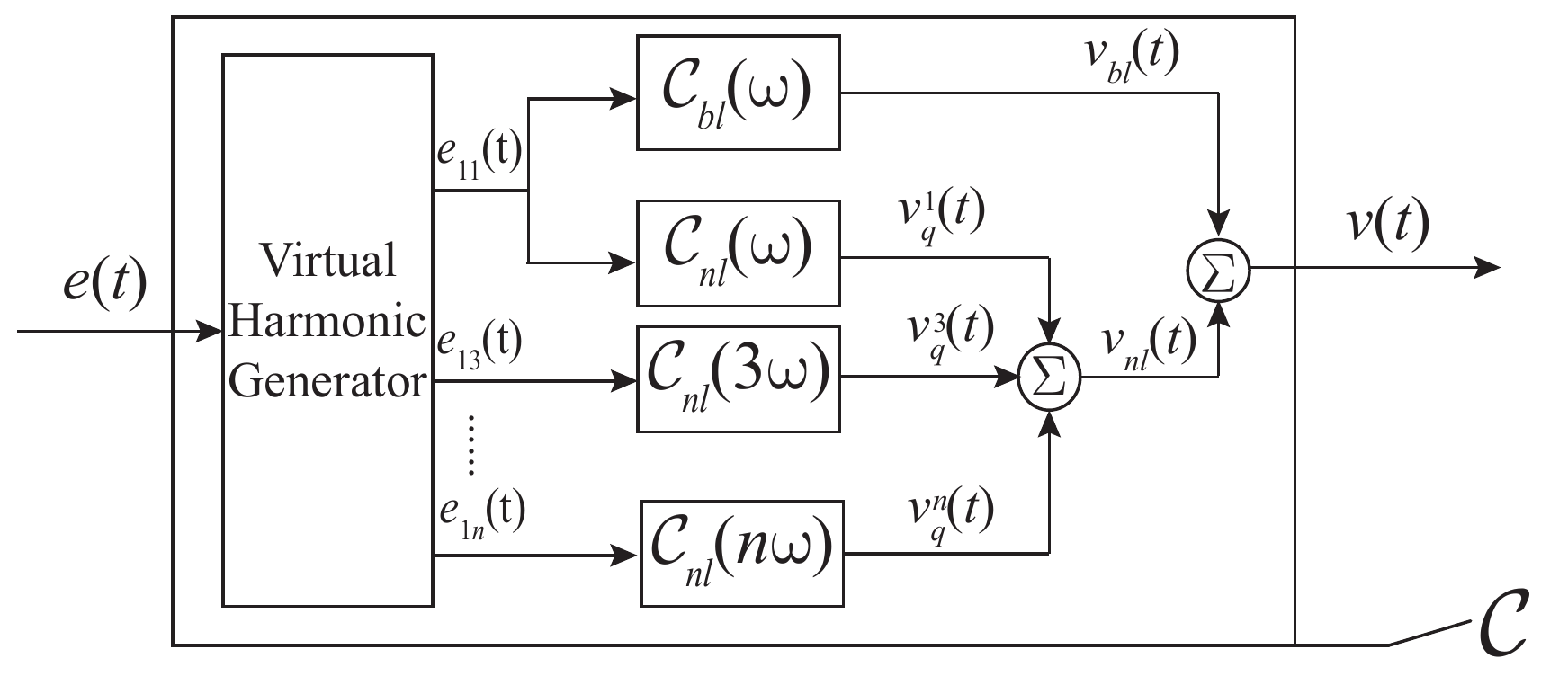}}
	\caption{The new block diagram of an open-loop reset controller $\mathcal{C}$.}
	\label{fig: State Space of Reset Controller}
 \end{figure}	
\begin{proof}
The proof can be found in \ref{appen: proof for thm2}.
\end{proof}
\begin{cor}
\label{cor: v,vnl}
Consider a reset controller $\mathcal{C}$ \eqref{eq: State-Space eq of RC} with one reset state where $n_r = 1$ in response to the input signal and reset-triggered signal $e(t) = |E_1|\sin(\omega t+\angle E_1),\ (\angle E_1\in (-\pi,\pi])$. The reset controller $\mathcal{C}$ \eqref{eq: State-Space eq of RC} operates under Assumption \ref{open-loop stability} and satisfies the Zeno-free condition in \eqref{eq: zeno-free}. The steady-state reset output signal $v(t)$ is expressed as:
\begin{equation}
\label{eq:vt=vbl+vnl}
    v(t) = v_{bl}(t) +v_{nl}(t),
\end{equation}
where $v_{bl}(t)$ is the steady-state base-linear output given by
\begin{equation}
\label{eq:vbl-11}
    v_{bl}(t) = |E_1\mathcal{C}_{bl}(\omega)|\sin(\omega t + \angle \mathcal{C}_{bl}(\omega)).
\end{equation}
The nonlinear signal $v_{nl}(t)$ is given by
\begin{equation}
\label{vnlt_rem}
    v_{nl}(t) = \sum\nolimits_{n=1}^{\infty}\mathscr{F}^{-1}[E_{1n}(\omega)\mathcal{C}_{nl}(n\omega)].
\end{equation}
\end{cor}
\begin{proof}    
The proof is provided in \ref{appen: proof for cor1}.
\end{proof}
Note that our earlier work \cite{kaczmarek2022steady} conceptualized the pulse-based model for the open-loop reset controller but did not conclude with a frequency-domain analysis for this controller. Theorem \ref{thm: Open loop model for RC} completes this research and presents a new HOSIDF for analyzing the open-loop reset controller. 

Applying Theorem \ref{thm: Open loop model for RC}, Remark \ref{rem: open-loop RCS} provides the analysis for the open-loop reset control system in Fig. \ref{fig1:RC system}.
\begin{rem} 
\label{rem: open-loop RCS}
For an open-loop reset control system in Fig. \ref{fig1:RC system} with the sinusoidal input signal $e(t) = |E_1|\sin(\omega t + \angle E_1)$ and under Assumption \ref{open-loop stability}, the transfer function $\mathcal{L}_n(\omega), \ n\in\mathbb{N}$ from the input $e(t)$ to the steady-state output $y(t)$ is composed of a linear transfer function $\mathcal{L}_{bl}(\omega)$ and nonlinear transfer functions $\mathcal{L}_{nl}(n\omega)$, given by
\begin{equation}
		\label{eq: OL MODEL FOR RCS}
		\begin{aligned}
		\mathcal{L}_n(\omega) = \begin{cases}
			\mathcal{L}_{bl}(\omega) + \mathcal{L}_{nl}(\omega), & \text{for } n=1\\
			\mathcal{L}_{nl}(n\omega), & \text{for odd } n >1\\
			0, & \text{for even } n \geqslant 2 
		\end{cases}
		\end{aligned}
	\end{equation}
	with
	\begin{equation}
		\label{eq: L, LNL}
		\begin{aligned}
			\mathcal{L}_{bl}(n\omega) &= \mathcal{C}_{bl}(n\omega)\mathcal{C}_{\alpha}(n\omega)\mathcal{P}(n\omega),\\
			\mathcal{L}_{nl}(n\omega) &= \mathcal{C}_{nl}(n\omega)\mathcal{C}_{\alpha}(n\omega)\mathcal{P}(n\omega),
		\end{aligned} 	
	\end{equation} 
 where $\mathcal{C}_{bl}(n\omega)$ and $\mathcal{C}_{nl}(n\omega)$ are given in \eqref{eq: Cbl} and \eqref{Cnl_final}, respectively.
 
The steady-state output signal $y(t)$ of the open-loop reset system is given by
\begin{equation}
\label{eq: ol-y(t)}
\begin{aligned}
    y(t) &= \mathscr{F}^{-1}[E_{1n}(\omega)\mathcal{L}_n(\omega)],\\
    E_{1n}(\omega) &= |E_1|\sin(n\omega t + n\angle E_1).
\end{aligned}
\end{equation}    
\end{rem}
To validate the accuracy of Theorem \ref{thm: Open loop model for RC} and Remark \ref{rem: open-loop RCS}, Fig. \ref{Open_loop_validate} compares the simulated and equation \eqref{eq: ol-y(t)}-predicted output signals $y(t)$ in an reset control system with input signal of $e(t) = \sin (400\pi t)$. The parameters of the system are given as: $\mathcal{C}_{bl}(s) = 1/(s/(300\pi)+1)$, $\mathcal{C}_{\alpha}(s) = (s/(75\pi)+1)/(s/(1200\pi)+1)$, $\mathcal{P}(s) = 1$, and $\gamma =0$. 
\begin{figure}[h]
	%	\missingfigure{RCsystem}
	\centerline{\includegraphics[width=0.92\columnwidth]{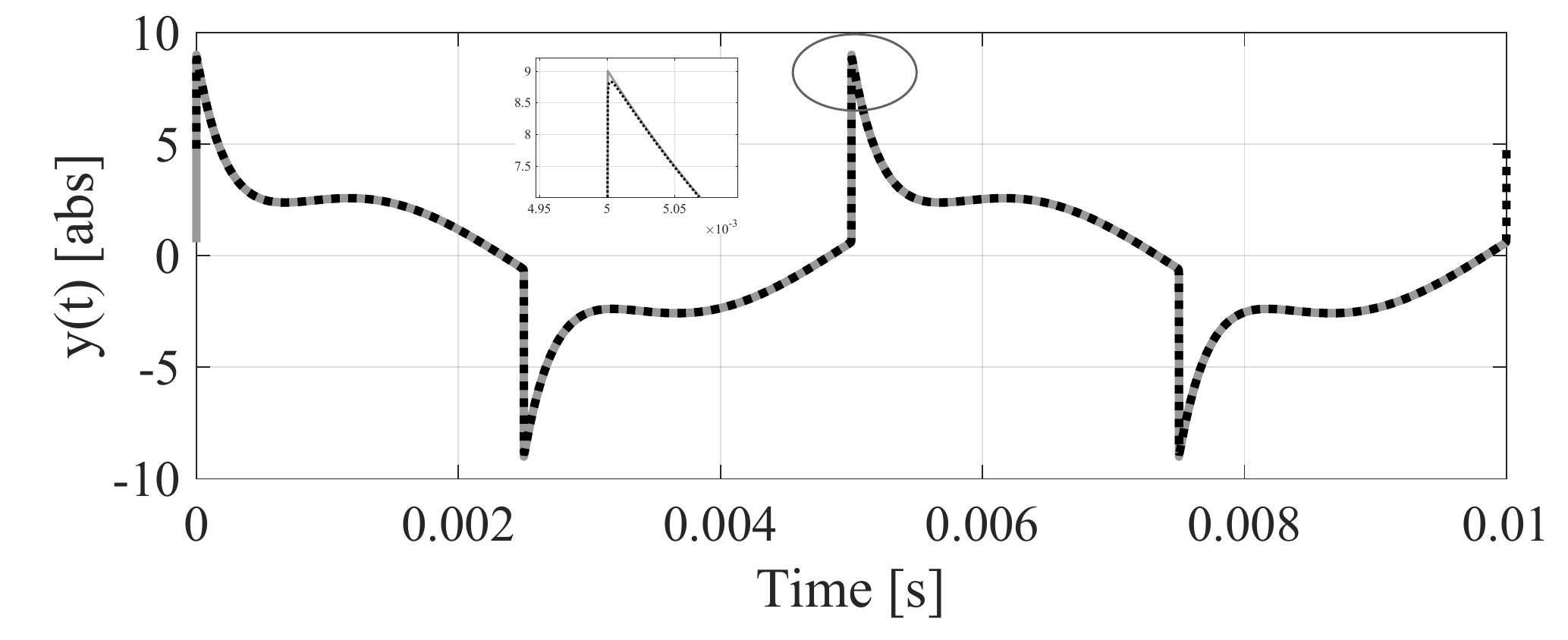}}
	\caption{The comparison between the simulated and the Theorem \ref{thm: Open loop model for RC}-predicted output signals in a reset control system.}
	\label{Open_loop_validate}
\end{figure}

The results demonstrate the accuracy of Theorem \ref{thm: Open loop model for RC} and Remark \ref{rem: open-loop RCS}. Note that that in the zoom-in plot, the slight discrepancy at the signal edges between the two plots arises from the consideration of 5000 harmonics during the calculation process, whereas the actual output comprises an infinite number of harmonics. The choice of the number of harmonics considered in the calculation allows readers to balance computation time and prediction precision.
%----------------------------------------------------

%----------------------------------------------------
% \begin{rem}
The newly introduced HOSIDF $\mathcal{C}_n(\omega)$ \eqref{eq: u_ol(t)} is mathematically equivalent to the function $H_n(\omega)$ \eqref{eq: HOSIDF} when the input signal and reset-triggered signal for the reset controller are the same, denoted as $e(t) = |E_1|\sin(\omega t +\angle E_1)$. However, the new HOSIDF offers several new insights: (1) It facilitates the separation of the output $v(t)$ of the reset controller into linear $v_{bl}(t)$ \eqref{eq:vbl-11} and nonlinear $v_{nl}(t)$ \eqref{vnlt_rem} components. Specifically, $v_{bl}(t)$ is derived from the base-linear transfer function $\mathcal{C}_{bl}(\omega)$ \eqref{eq: Cbl}, while $v_{nl}(t)$ is obtained from $\mathcal{C}_{nl}(n\omega)$ \eqref{Cnl_final}. (2) The nonlinear signal $v_{nl}(t)$ is a filtered pulse signal that shares the same phase and period as the reset-triggered signal $e_s(t)$. (3) The magnitude of the signal $v_{nl}(t)$ is determined by $\mathcal{C}_{nl}(n\omega)$. These insights enable the connection between open-loop analysis and closed-loop frequency analysis, as discussed in Section \ref{subsec: Closed-loop Model}.
\section{The Frequency Response Analysis for the Closed-loop Reset System}
\label{subsec: Closed-loop Model}
Based on the open-loop analysis in Theorem \ref{thm: Open loop model for RC}, this section develops the closed-loop frequency response analysis for the reset control depicted in Fig. \ref{fig1:RC system}.

In a closed-loop reset control system under a sinusoidal input signal \(r(t) = |R|\sin (\omega t)\) at steady states, a ``two-reset system" is defined as a reset system having two reset instants per steady-state cycle, while a ``multiple-reset system" involves more than two reset instants per steady-state cycle. In the two-reset system, the dominated component of the error signal $e(t)$ is the first-order harmonic $e_1(t)$ in \eqref{eq: e,y,u}. Multiple-reset actions, such as those implemented in the PI+CI control system, introduce excessive higher-order harmonics compared to two-reset actions \cite{banos2007definition}. This issue can be mitigated through careful design considerations \cite{saikumar2019constant, karbasizadeh2020band}. Conditions for achieving periodic output in a multiple-reset system, where the interval between successive resets is not constant, are discussed in \cite{beker2001analysis}. Classical DF also assumes the $e_1(t)$ results in reset actions. Moreover, based on the authors' best knowledge, most practical reset control systems in the literature are designed to take advantage of two-reset systems \cite{banos2012reset}. Given these insights, it is essential to explore frequency-domain analysis methods tailored to two-reset systems.  Hence, we introduce the following assumption. 
\begin{assum}
\label{assum:2reset}
     There are two reset instants in a SISO closed-loop reset control system with a sinusoidal reference input signal $r(t) = |R|\sin(\omega t)$ at steady states, where the reset-triggered signal is $e_1(t)$. 
\end{assum}
Note that achieving Assumption \ref{assum:2reset} is feasible through practical reset control design. For instance, the CgLp reset element introduced in \cite{saikumar2019constant} enables the realization of wide-band two-reset systems. The analysis of multiple-reset systems extends beyond the current scope and will be the focus of our future research.

Under Assumptions \ref{assum: stable} and \ref{assum:2reset}, the reset actions in the closed-loop SISO reset control system occur when $e_{1}(t) = |E_1|\sin(\omega t + \angle E_1)= 0$, where $\angle E_1\in(-\pi,\pi]$. The set of reset instants for this closed-loop reset system is denoted as $J_m := \{t_m = (m\pi - \angle E_1)/\omega | m \in \mathbb{Z}^+\}$. Since the reset interval $\sigma_m = t_{m+1}-t_m = \pi/\omega> \delta_{\text{min}}$ \cite{barreiro2014reset}, the trajectories for the reset system are Zeno-free.
% \end{defn}
\begin{figure}[htp]
	\centering
	%	\missingfigure{CLCI}
	\includegraphics[width=\columnwidth]{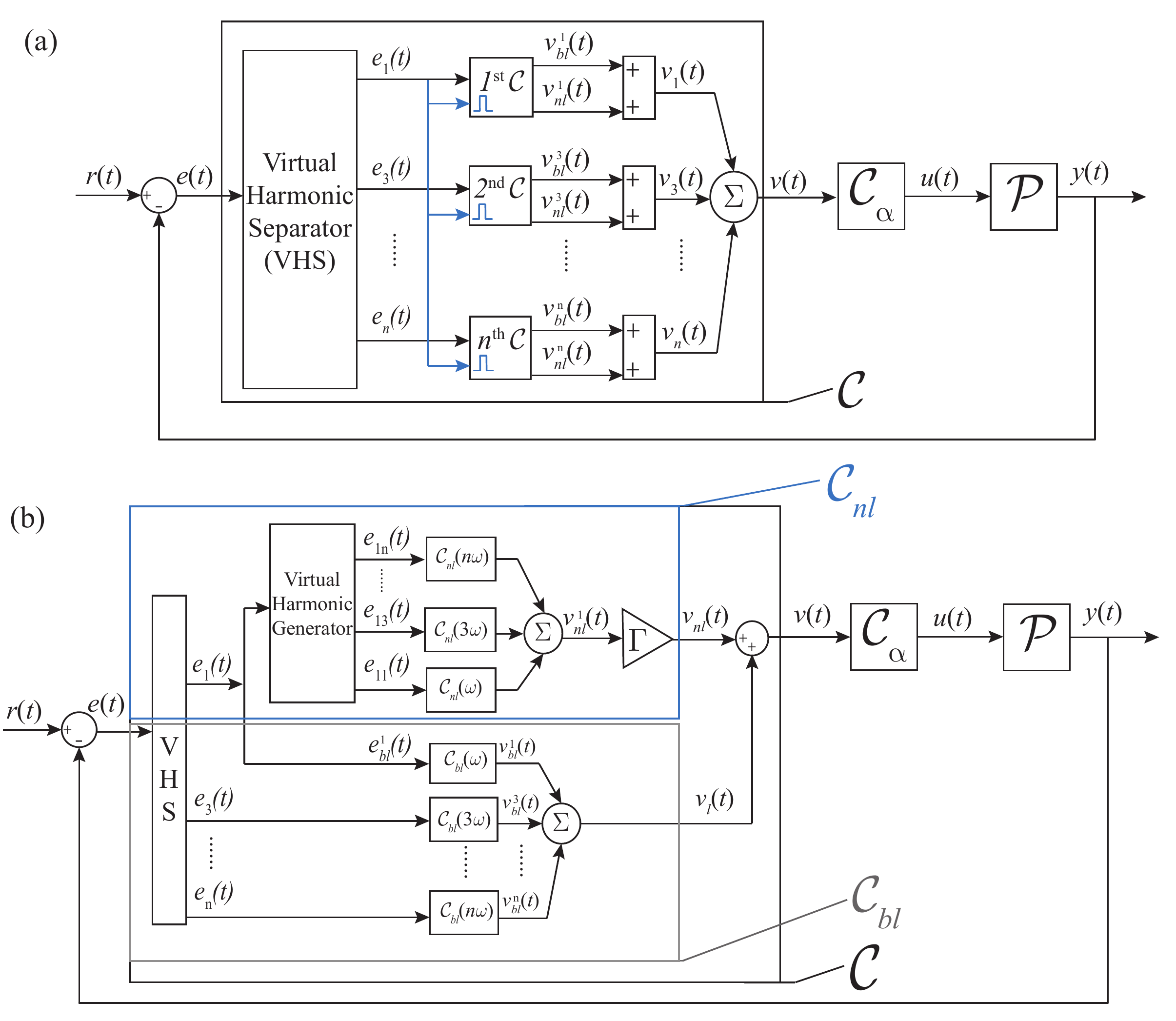}
	\caption{Block diagrams for the closed-loop RCS, wherein (a) the resetting actions are indicated by the blue lines. In (b), the reset controller is decomposed into two components: the linear part $\mathcal{C}_{bl}$ within the grey box and the nonlinear part $\mathcal{C}_{nl}$ contained within the blue box.}
	\label{fig: Model for Closed-loop RCS}
\end{figure}

Figure \ref{fig: Model for Closed-loop RCS}(a) constructs the first block diagram for the closed-loop reset control system under Assumptions \ref{assum: stable} and \ref{assum:2reset}. In this block diagram, first, utilizing the \enquote{Virtual Harmonic Separator} \cite{nuij2006higher}, the error signal $e(t)$ is decomposed into its harmonics, $e_n(t)$, as defined in \eqref{eq: e,y,u}. Next, each $e_n(t)$ is filtered by the reset controller $\mathcal{C}$, resulting in a response denoted as $v_n(t)$. By summing up $v_n(t)$, the reset output signal $v(t)$ is obtained. In this context, we refer to $\mathcal{C}$ with the input signal $e_n(t)$ as the $n$-th reset controller. Theorem \ref{thm: vnl_npi} illustrates that the steady-state output of a reset controller $\mathcal{C}$ under a sinusoidal input signal is composed of base-linear $v_{bl}(t)$ and nonlinear $v_{nl}(t)$ components. Let $v_{bl}^{n}(t)$ and $v_{nl}^{n}(t)$ represent the steady-state base-linear and nonlinear output signals for the $n$-th $\mathcal{C}$. Then, the steady-state reset output signal $v(t)$ in the closed-loop reset system is given by:
\begin{equation}
\label{eq: ut1}
\begin{aligned}
    v(t) &= \sum\nolimits_{n=1}^{\infty}v_{n}(t),\\
    v_n(t) &= v_{bl}^{n}(t) + v_{nl}^{n}(t).    
\end{aligned}
\end{equation}
From \eqref{eq: ut1}, by defining
\begin{equation}
\label{vl,vnl}
    \begin{aligned}
       v_{l}(t) &=  \sum\nolimits_{n=1}^{\infty}v_{bl}^{n}(t),\\
       v_{nl}(t) &=  \sum\nolimits_{n=1}^{\infty}v_{nl}^{n}(t),
    \end{aligned}
\end{equation}
$v(t)$ can be written as
\begin{equation}
\label{vt,vl,vnl}
    v(t) = v_{l}(t) + v_{nl}(t).
\end{equation}
In the Fourier domain, equation \eqref{eq: ut1} is expressed as
\begin{equation}
\label{eq: uw2}
\begin{aligned}
    V(\omega) &= \sum\nolimits_{n=1}^{\infty} V_n(\omega),\\
    V_n(\omega) &= V_{bl}^{n}(\omega) + V_{nl}^{n}(\omega).
\end{aligned}
\end{equation}
Derived from \eqref{eq:cor:vbln}, $ V_{bl}^{n}(\omega)$ is given by
\begin{equation}
\label{eq: vl}
    V_{bl}^{n}(\omega)=E_n(\omega)\mathcal{C}_{bl}(n\omega).
\end{equation}

Based on the block diagram for the closed-loop reset system shown in Fig. \ref{fig: Model for Closed-loop RCS}(a), Theorem \ref{thm: Pulse-based model for RCS in closed loop} concludes the development of the pulse-based model for the closed-loop reset control system, visually represented in Fig. \ref{fig: Model for Closed-loop RCS}(b).
\begin{thm} (The pulse-based analysis model for the closed-loop reset system)
\label{thm: Pulse-based model for RCS in closed loop}
In a closed-loop reset control system (with reset controller $\mathcal{C}$ \eqref{eq: State-Space eq of RC} where $n_r = 1$), as depicted in Fig. \ref{fig1:RC system} with a sinusoidal reference input signal $r(t) = |R|\sin(\omega t)$ and under Assumptions \ref{assum: stable} and \ref{assum:2reset}, the steady-state reset output signal $v(t)$ is expressed as:
\begin{equation}
    \label{eq:u_sum}
    \begin{aligned}
        v(t) &= v_{l}(t) + v_{nl}(t),\\
        v_{l}(t) &= \sum\nolimits_{n=1}^{\infty}v_{bl}^{n}(t),\\
        v_{nl}(t) &= \Gamma(\omega)\sum\nolimits_{n=1}^{\infty}v_{nl}^{1}(t),\\
        v_{bl}^{n}(t) &= \mathscr{F}^{-1}[E_n(\omega)\mathcal{C}_{bl}(n\omega)],\\
        v_{nl}^{1}(t) &= \sum\nolimits_{n=1}^{\infty}\mathscr{F}^{-1}[E_{1n}(\omega)\mathcal{C}_{nl}(n\omega)],
    \end{aligned}
\end{equation}
    where 
    \begin{equation}
    \label{eq: Gamma_thm}
        \begin{aligned}
            E_{1n}(\omega) &= |E_1|\mathscr{F}^{-1}[\sin(n\omega t+n\angle E_1)],\\
            \Delta_c^1(\omega) &=|\Delta_l(\omega)| \sin(\angle \Delta_l(\omega)),\\
            \Gamma(\omega) &= 1/(1-{\sum\nolimits_{n=3}^{\infty}\Psi_n(\omega)\Delta_c^n(\omega)}{/\Delta_c^1(\omega)}),\\
            \Psi_n(\omega) &= {|\mathcal{L}_{nl}(n\omega)|}/{|1+\mathcal{L}_{bl}(n\omega)|},\ n=2k+1,\ k\in\mathbb{N},\\
        \Delta_c^n(\omega) &= -|\Delta_l(n\omega)| \sin(\angle\mathcal{L}_{nl}(n\omega) - \angle (1+\mathcal{L}_{bl}(n\omega)) +\\
        &\indent \indent\indent\indent\indent\indent\indent\indent\indent\indent\angle \Delta_l(n\omega)),\text{ for } n>1.
        \end{aligned}
    \end{equation}
% \end{enumerate}   
Functions $\mathcal{C}_{bl}(\omega)$, $\mathcal{C}_{nl}(\omega)$, $\Delta_l(n\omega)$, $\mathcal{L}_{bl}(n\omega)$ and $\mathcal{L}_{nl}(n\omega)$ can be found in \eqref{eq: Cbl}, \eqref{Cnl_final}, \eqref{eq:Vnl(t)npi2} and \eqref{eq: L, LNL}, respectively.
\end{thm}
\begin{proof}
The proof is provided in \ref{appen: proof for thm3}. 
\end{proof}
\begin{rem}
\label{rem:Gamma - Vnl}
From \eqref{vl,vnl} and \eqref{eq:u_sum}, $\Gamma(\omega)$ in \eqref{eq: Gamma_thm} represents the ratio of $v_{nl}(t)=1+\sum_{n=3}^{\infty}v_{nl}^n(t)$ to $v_{nl}^1(t)$ at the input frequency $\omega$. It serves as an indicator of the relative magnitude of higher-order harmonics $\sum_{n=3}^{\infty} v_{nl}^n(t)$ compared to the first-order harmonic $v_{nl}^1(t)$. A larger value of $\Gamma(\omega)$ indicates a relatively larger magnitude of higher-order harmonics in closed-loop reset systems. 
\end{rem}
% \begin{rem}
% \label{rem: Gamma_En_E1}
% Based on \eqref{eq:G(w)1}, $|\Gamma(\omega) -1| \propto \frac{\sum\nolimits_{n=1}^{\infty}|E_n| }{|E_1|}$.
% \end{rem}

Building upon the analytical model introduced in Theorem \ref{thm: Pulse-based model for RCS in closed loop}, we propose a new Higher-Order Sinusoidal Input Describing Function (HOSIDF) for the frequency response analysis of closed-loop reset control systems. The detailed formulation is presented in Theorem \ref{thm: Method C, new HOSIDF}.
\begin{thm} 
\label{thm: Method C, new HOSIDF}
(The closed-loop HOSIDF for SISO reset control systems) Consider a closed-loop SISO reset control system in Fig. \ref{fig1:RC system} with a reset controller $\mathcal{C}$ \eqref{eq: State-Space eq of RC} (where $n_r=1$) and to a sinusoidal reference input signal $r(t) = |R|\sin(\omega t)$. This system complies with Assumptions \ref{assum: stable} and \ref{assum:2reset}. By employing the \enquote{Virtual Harmonics Generator} to the input signal $r(t)$, the signal $r_n(t) = |R|\sin (n\omega t),\ n\in\mathbb{N}$ is introduced, along with its Fourier transform $R_n(\omega)$. The transfer functions $\mathcal{S}_n(\omega)$, $\mathcal{T}_n(\omega)$, and $\mathcal{CS}_n(\omega)$ of the closed-loop reset system at steady states are defined as follows:
	 \begin{equation}
	 	\label{eq: sensitivity functions in CL}
    \resizebox{1\hsize}{!}{$
%	 		$
	 		\mathcal{S}_n(\omega)= \frac{E_n(\omega)}{R_n(\omega)}=
	 		\begin{cases}
	 			\frac{1}{1+\mathcal{L}_{o}(\omega)} , & \text{for } n=1\\
	 			-\frac{\Gamma(\omega)\mathcal{L}_{nl}(n\omega)|\mathcal{S}_{1}(\omega)|e^{jn\angle \mathcal{S}_{1}(\omega)}}{1+\mathcal{L}_{bl}(n\omega)} , & \text{for odd} \ n > 1\\
	 			0, & \text{for even} \ n \geqslant 2 \end{cases}
			$}
	 \end{equation}
	\begin{equation}
		\label{eq: Complementary sensitivity functions in CL}
    \resizebox{1\hsize}{!}{$
			\mathcal{T}_n(\omega)= \frac{Y_n(\omega)}{R_n(\omega)}=
			\begin{cases}
				\frac{\mathcal{L}_{o}(\omega)}{1+\mathcal{L}_{o}(\omega)} , & \text{for } n=1\\
				\frac{\Gamma(\omega)\mathcal{L}_{nl}(n\omega)|\mathcal{S}_{1}(\omega)|e^{jn\angle \mathcal{S}_{1}(\omega)}}{1+\mathcal{L}_{bl}(n\omega)} , & \text{for odd} \ n > 1\\
				0, & \text{for even} \ n \geqslant 2 \end{cases}
			$}
	\end{equation}
        \begin{equation}
	\label{eq: CS}        
        \mathcal{CS}_n(\omega)= \frac{U_n(\omega)}{R_n(\omega)}=\frac{\mathcal{T}_n(\omega)}{\mathcal{P}(n\omega)},
        \end{equation}
 %        \begin{equation}
	% \label{eq: Theta}
 %        \Theta_n(\omega) = \frac{V_n(\omega)}{R_n(\omega)}=\frac{\mathcal{T}_n(\omega)}{\mathcal{C}_\alpha(n\omega)\mathcal{P}(n\omega)},
 %        \end{equation}
%         \begin{equation}
% 		\label{eq: CS}
% %		\resizebox{\textwidth}{!}{%
% %			$
% 			\mathcal{CS}_n(\omega)= 
% 			\begin{cases}
% 				\frac{\mathcal{L}_{o}(\omega)}{\mathcal{P}(\omega)(1+\mathcal{L}_{o}(\omega))} , & \text{for } n=1\\
% 				\frac{\Gamma(\omega)\mathcal{C}_{nl}(n\omega)\mathcal{C}_{\alpha}(n\omega)|\mathcal{S}_{1}(\omega)|e^{jn\angle \mathcal{S}_{1}(\omega)}}{1+\mathcal{L}_{bl}(n\omega)} , & \text{for odd} \ n > 2\\
% 				0, & \text{for even} \ n \geqslant 2 \end{cases}
% %			$
% 			%
% %			}
% 	\end{equation}
where
\begin{equation}
\label{eq: linear and nonlinear part of t_mc}
\begin{aligned}
  R_n(\omega) &= |R|\mathscr{F}[\sin(n\omega t)],\\
\mathcal{L}_{o}(n\omega) &= \mathcal{L}_{bl}(n\omega) + \Gamma(\omega)\mathcal{L}_{nl}(n\omega).	
\end{aligned}
\end{equation}
Functions $\mathcal{L}_{bl}(n\omega)$, $\mathcal{L}_{nl}(n\omega)$, and $\Gamma(\omega)$ are given in \eqref{eq: L, LNL} and \eqref{eq: Gamma_thm}. 
\end{thm}
\begin{proof}
The proof is provided in \ref{appen: Proof for thm4}.
\end{proof}
% \begin{rem}
%     In a closed-loop RCS, $\Gamma(\omega) \propto |T_n(n\omega)|= |S_n(n\omega)|$. Thus, larger $\Gamma(\omega)$ 
% \end{rem}
\begin{rem}
\label{cor: e,y,u,v}
Consider a SISO reset control system under Assumptions \ref{assum: stable} and with a sinusoidal reference signal $r(t) = |R|\sin(\omega t)$, as depicted in Fig. \ref{fig1:RC system}. By utilizing the \enquote{Virtual Harmonic Generator}, the input signal $r(t)$ generates $n=2k+1 (k\in\mathbb{N})$ harmonics $r_n(t) = |R|\sin(n\omega t)$. The Fourier transform for $r_n(t)$ is denoted as $R_n(\omega)$. The steady-state error signal $e(t)$, output signal $y(t)$, and control input signal $u(t)$ are expressed as follows:
    \begin{equation}
        \begin{aligned}
            e(t) &= \sum\nolimits_{n=1}^{\infty}e_n(t)=\sum\nolimits_{n=1}^{\infty}\mathscr{F}^{-1}\left[\mathcal{S}_n(\omega)R_n(\omega)\right],\\
            y(t) &= \sum\nolimits_{n=1}^{\infty}y_n(t)= \sum\nolimits_{n=1}^{\infty}\mathscr{F}^{-1}\left[\mathcal{T}_n(\omega)R_n(\omega)\right],\\
            u(t) &= \sum\nolimits_{n=1}^{\infty}u_n(t)=\sum\nolimits_{n=1}^{\infty}\mathscr{F}^{-1}\left[\mathcal{CS}_n(\omega)R_n(\omega)\right].
            % v(t) &= \sum\nolimits_{n=1}^{\infty}\mathscr{F}^{-1}\left[\Theta_n(\omega)R_n(\omega)\right].
        \end{aligned}
    \end{equation}
\end{rem}
\begin{rem}
\label{rem: Gamma}
The function $\Gamma(\omega)$ in \eqref{eq: Gamma_thm}, represents the ratio of the nonlinear outputs $v_{nl}(t)$ to $v^1_{nl}(t)$ in \eqref{eq:u_sum} at input frequency $\omega$. Previous frequency response analysis methods, Method A and Method B, introduced in Section \ref{subsec: PS}, assume that the higher-order harmonics $v^n_{nl}(t)$ (for $n>1$) generated by $e_n(t)$ are zero, thereby implying $\Gamma(\omega) = 1$. However, this assumption does not hold across the entire frequency spectrum of the reset control system, leading to inaccuracies in the analysis. Theorem \ref{thm: Method C, new HOSIDF} addresses these inaccuracies analytically by introducing $\Gamma(\omega)$ in \eqref{eq: Gamma_thm}. While Theorem \ref{thm: Method C, new HOSIDF} exhibits comparable accuracy to Method B within the frequency range where $\Gamma(\omega) = 1$, it has superior accuracy compared to Method B when $\Gamma(\omega) \neq 1$, particularly in scenarios where the magnitudes of higher-order harmonics $v^n_{nl}(t)$ (for $n>1$) are significant. 
% Method A not only disregards the $v^n_{nl}(t)$ (for $n>1$) but also neglects the $v^n_{bl}(t)$ (for $n>1$) in \eqref{eq:u_sum}. 
\end{rem}
% \begin{rem}
% According to Theorem \ref{thm: Method C, new HOSIDF} and Remark \ref{cor: e,y,u,v}, the frequency response analysis for closed-loop SISO reset control systems to sinusoidal inputs reveals that the magnitudes of the higher-order harmonics (\(n>1\)) in both $y(t)$ and $e(t)$ are equal, while their phases exhibit opposite signs.
% \end{rem}

%% file: 4_Results.tex
% \section{Illustrative Examples}
% \label{sec: results}
% This section shows the experimental setup: a precision motion stage and the T-reset control systems designed for illustration. The effectiveness of Theorem \ref{thm: Method C, new HOSIDF} in predicting the closed-loop behavior of T-reset control systems is demonstrated. Additionally, the superior performance of T-reset control systems in enhancing steady-state precision compared to traditional RCS is illustrated.
\section{Case Study 1: Proportional-Clegg-Integrator Proportional-Integrator-Derivative (PCI-PID) Control System}
\label{example 1}
In this section, we design a PCI-PID controller on a precision motion stage for validating the Theorem \ref{thm: Method C, new HOSIDF}. 
\subsection{Precision Positioning Setup}
\label{sec:spider}
The plant utilized in this study is a planar motion system with three degrees of freedom as depicted in Fig. \ref{fig1:Spyder}(a), referred to as \enquote{Spyder} stage. This system employs dual leaf flexures, each associated with corresponding masses ($M_1$, $M_2$, $M_3$), for connection to the base ($M_c$). These masses are driven by three voice coil actuators labeled $A_1$, $A_2$, and $A_3$. Linear encoders (denoted as \enquote{Enc}), specifically Mercury M2000 with a resolution of 100 nm and sampled at 10 kHz, are utilized to monitor the positions of the masses. Additionally, with additional oversampling introduced on the FPGA, this resolution is increased to 3.125 nm. For the SISO investigation, only actuator $A_1$ is used to position mass $M_1$. The control systems are implemented on an NI compactRIO platform and incorporate a linear current source power amplifier.
\begin{figure}[h]
	%	\missingfigure{RCsystem}
	\centerline{\includegraphics[width=\columnwidth]{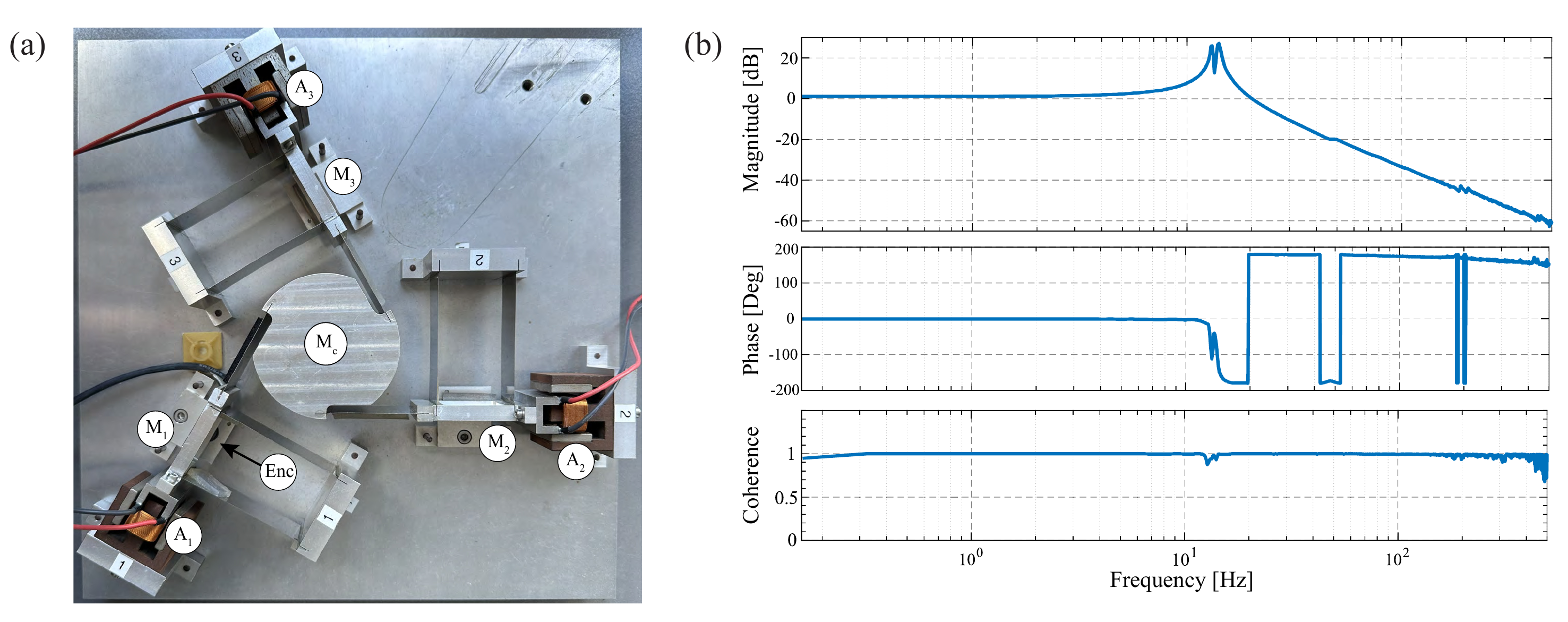}}
	\caption{(a) The planar precision positioning system. (b) The FRF data from actuator $A_1$ to attached mass $M_1$.}
	\label{fig1:Spyder}
\end{figure}

Figure \ref{fig1:Spyder}(b) depicts the measured Frequency Response Function (FRF) of the system, which exhibits a collocated double mass-spring-damper system with additional high-frequency parasitic dynamics. To improve control clarity, the system is approximated to a single eigenmode mass-spring-damper system using Matlab's identification tool. The transfer function of the system is expressed as:
\begin{equation}
\label{eq:P(s)}
    \mathcal{P}(s) = \frac{6.615e5}{83.57s^2+279.4s+5.837e5}.
\end{equation}
% \begin{figure}[h]
% 	%	\missingfigure{RCsystem}
% 	\centerline{\includegraphics[width=0.5\columnwidth]{figs/FRF_Spider.pdf}}
% 	\caption{FRF data from actuator $A_1$ to attached mass $M_1$.}
% 	\label{fig1: Bode2 of Spyder}
% \end{figure}
\subsection{The Validation and the Limitation of Theorem \ref{thm: Method C, new HOSIDF} on the Analysis of the PCI-PID Control System}
The PID controller is widely used in industries. Within the PID framework, we design a PCI-PID controller for the precision motion stage, as shown in Fig. \ref{Structure_PCIID}. The parameters of the PCI-PID system are as follows: the reset value $\gamma=0$, $K_p = 20.5$, the cut-off frequency $\omega_c = 2\pi\cdot 150$ [rad/s], $\omega_d = \omega_c/4.8$, $\omega_t=\omega_c\cdot4.8$, $\omega_f = 10\omega_c$, and $\omega_i=0.1\omega_c$. The design specifications for the PCI-PID system aim to achieve a bandwidth (BW) of 150 Hz and a phase margin (PM) of 50$\degree$ in the open-loop. Additionally, the stability and convergence of the system have been verified through testing.
\begin{figure}[h]
    \centering
    \centerline{\includegraphics[width=\columnwidth]{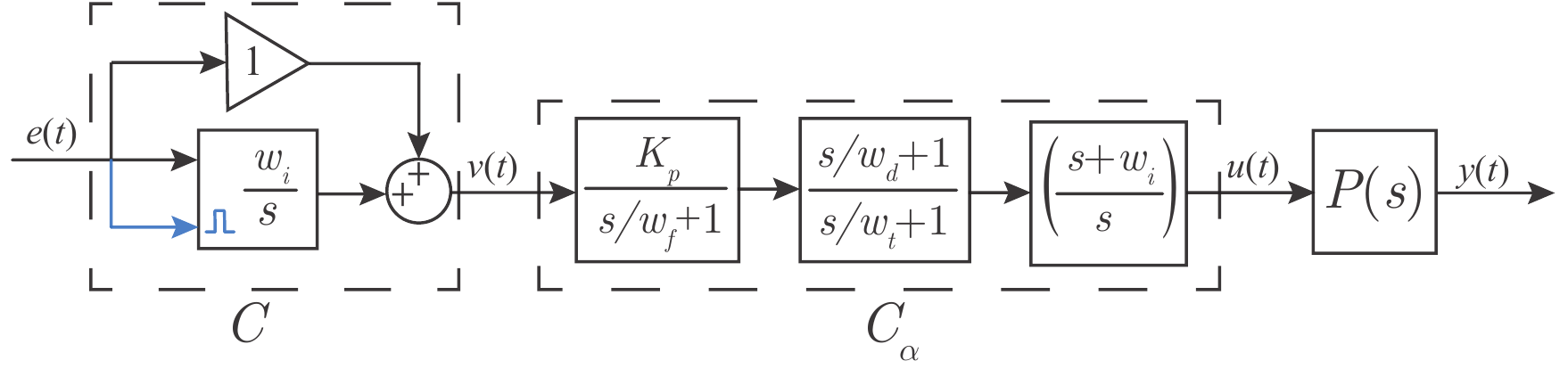}}
    \caption{The open-loop block diagram of the PCI-PID control system.}
	\label{Structure_PCIID}
\end{figure}
\begin{figure}[htpb]
	\centering
	\includegraphics[width=0.98\columnwidth]{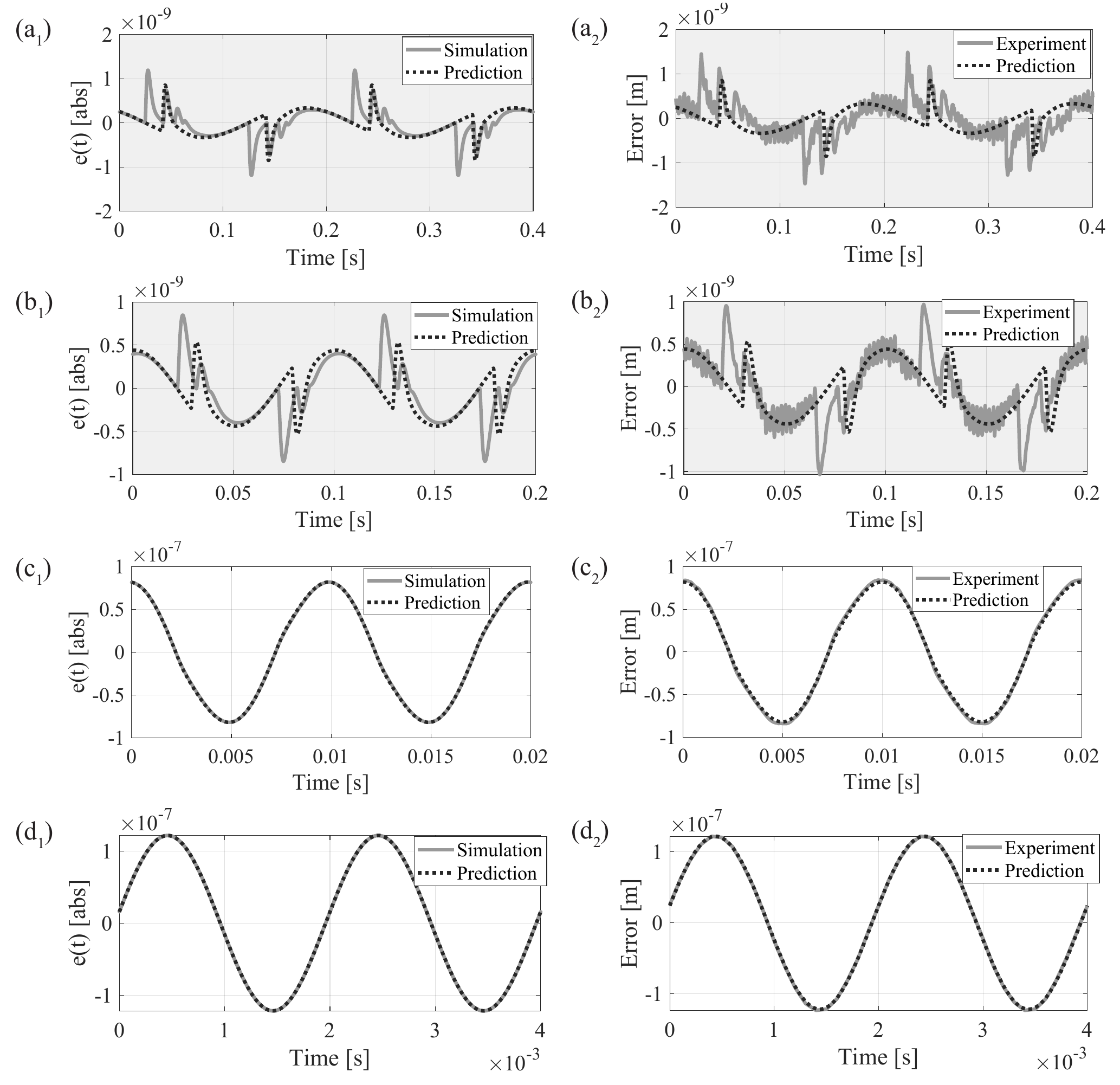}
	\caption{The steady-state errors $e(t)$ for the PCI-PID system under sinusoidal input $r(t) = \sin(2\pi f t)$ at input frequencies $f$ of 5 Hz in ($a_1$) and ($a_2$), 10 Hz in ($b_1$) and ($b_2$), 100 Hz in ($c_1$) and ($c_2$), and 500 Hz in ($d_1$) and ($d_2$) as predicted, simulated, and experimentally measured. The grey background indicates scenarios involving multiple-reset actions.}
	\label{CL PCIID_et_3freqs}
\end{figure} 

The steady-state errors indicate the tracking precision of the motion control system. We utilize the steady-state errors to validate the accuracy of Theorem \ref{thm: Method C, new HOSIDF}. Figures \ref{CL PCIID_et_3freqs}($c_1$), ($c_2$), ($d_1$), and ($d_2$) depict the steady-state errors \(e(t)\) of the PCI-PID system under input signal $r(t)=1E-7\sin(2\pi f)$ at input frequencies $f$ of 100 Hz and 500 Hz, as predicted, simulated, and measured in experiments. The alignment between the predictions, simulations, and experiments suggests that Theorem \ref{thm: Method C, new HOSIDF} accurately predicts the system behavior. Note that the reference input signal used in the experiments is $r(t) = 1E-7\sin(2\pi f t)$ [m].

However, as depicted in Fig. \ref{CL PCIID_et_3freqs}($a_1$), ($a_2$), ($b_1$), and ($b_2$), at input frequencies of 5 Hz and 10 Hz where multiple resets per cycle occur, Assumption \ref{assum:2reset} is not satisfied. This results in inaccuracies in the predictions provided by Theorem \ref{thm: Method C, new HOSIDF}. Note that the jagged noises observed in the measured signals during experiments arise from the sensor's resolution limitations. In practical scenarios, multiple resets can introduce excessive nonlinearity to the system, leading to undesired performance issues such as the limit cycle problem in the PI+CI system \cite{banos2011limit}. It is preferred to design reset systems without multiple resets across the entire frequency range. The examples of the PCI-PID system in Fig. \ref{CL PCIID_et_3freqs}($a_1$)-($b_2$) serve to illustrate the limitation of the accuracy of the new analysis method imposed by Assumption \ref{assum:2reset}.

In Section \ref{example 2}, we introduce a new reset control structure as a case study aimed at addressing the multiple-reset limitation observed in the PCI-PID control structure. This new structure is designed to achieve two reset instants per steady-state cycle, thereby satisfying Assumption \ref{assum:2reset} required for Theorem \ref{thm: Method C, new HOSIDF}.
\section{Case Study 2: Two-reset-PCI-PID (T-PCI-PID) Control System}
\label{example 2}
\subsection{The Two-Reset Control System}
In this section, we introduced a new reset control structure termed the Two-Reset Control System (T-RCS), illustrated in Fig. \ref{fig2:RC system}. This structure forces the system to reset twice per steady-state cycle under a sinusoidal reference input signal, serving as an illustrative example for validating Theorem \ref{thm: Method C, new HOSIDF}. 
% Experimental results suggest that this new structure achieves improved tracking performance compared to traditional reset systems. However, this aspect extends beyond the primary focus of this paper and will be explored further in future work.
% Since the error signal $e(t)$ includes many harmonics in a closed-loop reset system, the objective of the transfer function $\mathcal{C}_s$ is to filter the first-order harmonic in $e(t)$ and feed it into $e_s(t)$.
\begin{figure}[h]
	%	\missingfigure{RCsystem}
	\centerline{\includegraphics[width=0.92\columnwidth]{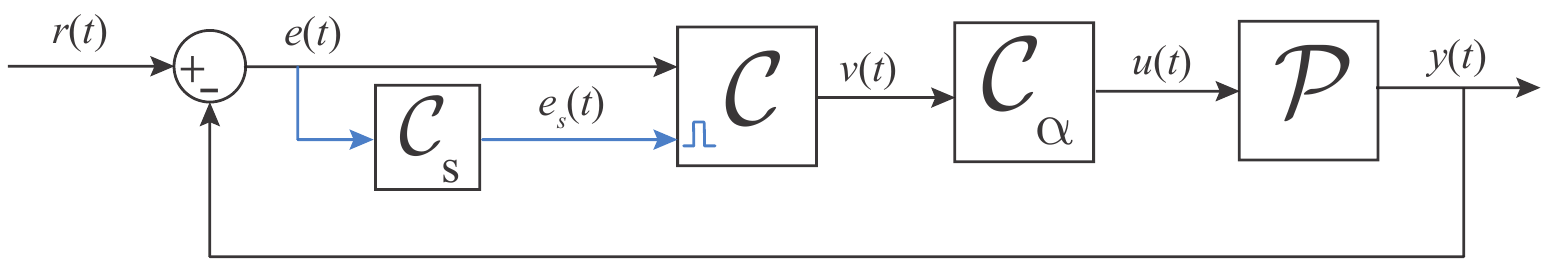}}
	\caption{Block diagram of the T-reset system, where $\mathcal{C}_s(s)$ is given in \eqref{Cs(s)} and $e_s(t)$ is the reset signal.}
	\label{fig2:RC system}
\end{figure}

In a traditional closed-loop reset system under sinusoidal input $r(t) = |R|\sin(\omega t)$, the error signal $e(t)$ is nonlinear and includes infinitely many harmonics $e_n(t)$ as defined in \eqref{eq: e,y,u}. This new T-RCS introduces a transfer function $\mathcal{C}_s$ between the error signal $e(t)$ and the reset-triggered signal $e_s(t)$. The transfer function of $\mathcal{C}_s(s)$ is given by
\begin{equation}
\label{Cs(s)}
    \mathcal{C}_{s}(s) = k_{cs}\left[\frac{(s/(\omega))^2+s/(\omega\cdot Q_1) +1}{(s/(\omega))^2+s/(\omega\cdot Q_2) +1}\right],
\end{equation}
where $Q_1<Q_2\in\mathbb{R}^+$. In this paper, we set $Q_1=1$, $Q_2=100$, and $k_{cs}=0.05$.
\begin{figure}[h]
	%	\missingfigure{RCsystem}
	\centerline{\includegraphics[width=0.9\columnwidth]{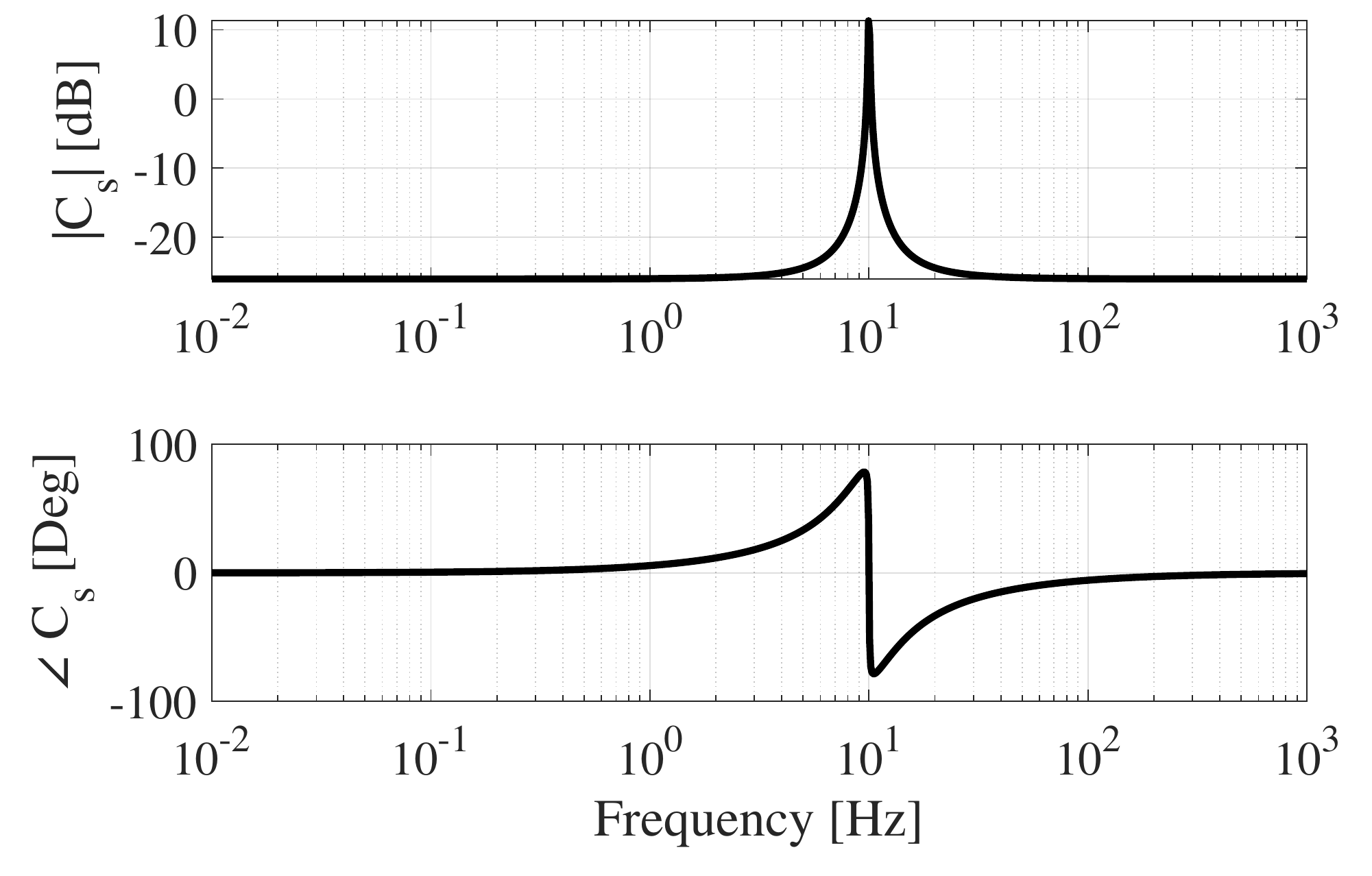}}
	\caption{Bode Plot of $\mathcal{C}_s(s)$ for $\omega=2\pi \cdot 10$ [rad/s].}
	\label{Bode Plot Cs}
\end{figure}
% The new structure is designed for SISO RCS working under sinusoidal input signal and step signal. The transfer function $\mathcal{C}_s$ is designed based on the two input signal:
% \begin{enumerate}
%     \item 

From \eqref{Cs(s)}, $\mathcal{C}_s(s)$ functions as an anti-notch filter. This filter allows frequencies within a specific range to pass through while attenuating all others. By adjusting the parameters $Q_1$ and $Q_2$, the frequency range of interest can be tailored. For example, selecting $Q_1=1, \  Q_2=100,\ k_{cs}=0.05$ for a frequency of $\omega= 2\pi \cdot 10$ [rad/s] yields the Bode plot depicted in Fig. \ref{Bode Plot Cs}. At the frequency $\omega = 2\pi\cdot10$ [rad/s], $\angle \mathcal{C}_s(\omega) = 0$, and the magnitude $|\mathcal{C}_s(\omega)|$ is relatively a peak to other frequencies. Therefore, $\mathcal{C}_s(s)$ approximately allows signals with a frequency of $\omega$ to pass through. By aligning $\omega$ with the first-order input frequency in $e(t)$ \eqref{eq: e,y,u}, the reset-triggered signal $e_s(t)$ is approximately expressed as:
\begin{equation}
\label{es(t)}
    e_s(t) = |E_1\mathcal{C}_s|\sin(\omega t+\angle E_1).
\end{equation}
From \eqref{es(t)}, the reset-triggered signal $e_s(t)$ shares the same frequency and phase as the first-order harmonic $e_1(t)$ in the error signal $e(t)$. Since the reset-triggered signal $e_s(t)$ is phase-dependent and amplitude-independent, equation \eqref{es(t)} indicates that the reset action of $\mathcal{C}$ is triggered based on $e_1(t)$. Thus, within this structure, the T-RCS ensures two reset instants per steady-state. This example serves to validate Theorem \ref{thm: Method C, new HOSIDF}, satisfying Assumption \ref{assum:2reset}.
% \item Suppose that a closed-loop RCS is subjected to a step input signal given by
% \begin{equation}
% \label{eq: h(t)}
%     \displaystyle h(t):={\begin{cases}1,&t>0,\\0,&t\leq 0.\end{cases}}
% \end{equation} 
% The step signal $h(t)$ is the normalised square wave signal $q_0(t)$ given in \eqref{eq: q0} when its fundamental frequency $\omega$ tends to 0, expressed as $h(t) = \lim_{\omega\to 0} q_0(t)$. Thus, to filter the fundamental frequency of $h(t)$, $\omega$ in $\mathcal{C}_{s}(s)$ is set to be equal to a relative small value 1E-16. 
% \end{enumerate}
\subsection{The Application of Theorem \ref{thm: Method C, new HOSIDF} to Analyze the T-PCI-PID System}
% The steady-state error signals $e(t)$ of the closed-loop CI without $\mathcal{C}_s$ and the new closed-loop Two-Reset CI (T-CI) at input frequencies of 0.1 Hz and 100 Hz are illustrated in Fig. \ref{CLSCI_et_0.1_100Hz}(a) and (b), respectively. In the traditional CI under a low-frequency sinusoidal reference input signal, it exhibits multiple-reset behavior (defined as more than two resets per steady-state). The new T-CI can overcome the multiple-reset scenarios and exhibits lower steady-state error. Conversely, under a high-frequency sinusoidal input, the T-CI and CI behave approximately the same because the traditional CI exhibits two resets per cycle, and the notch filter $\mathcal{C}_s$ will keep the system behavior.
% \begin{figure}[h]
% 	%	\missingfigure{RCsystem}
% 	\centerline{\includegraphics[width=0.6\columnwidth]{figs/CLSCI_et_0.1_100Hz.pdf}}
% 	\caption{The simulated error signal $e(t)$ of closed-loop CI and T-CI at input frequency of (a) 0.1 Hz and (b) 100 Hz.}
% 	\label{CLSCI_et_0.1_100Hz}
% \end{figure}
\subsubsection{The Accuracy of Theorem \ref{thm: Method C, new HOSIDF}}
\label{subsub: case 2}
Figure \ref{Open-loop block diagrams of PCID RCS} shows the structure of the Two-reset-PCI PID (T-PCI-PID) Control System. The transfer function $\mathcal{C}_s$ is defined in \eqref{Cs(s)}. Except for the setting for shaping filter $\mathcal{C}_s$ with $Q_1=1,\ Q_2=100,\ K_{cs}=0.05$, the parameters of the T-PCI-PID system are specified the same as those of the PCI-PID control system in Case Study 1. It has been verified that the T-PCI-PID system is stable and convergent.
\begin{figure}[h]
    \centering
    \centerline{\includegraphics[width=\columnwidth]{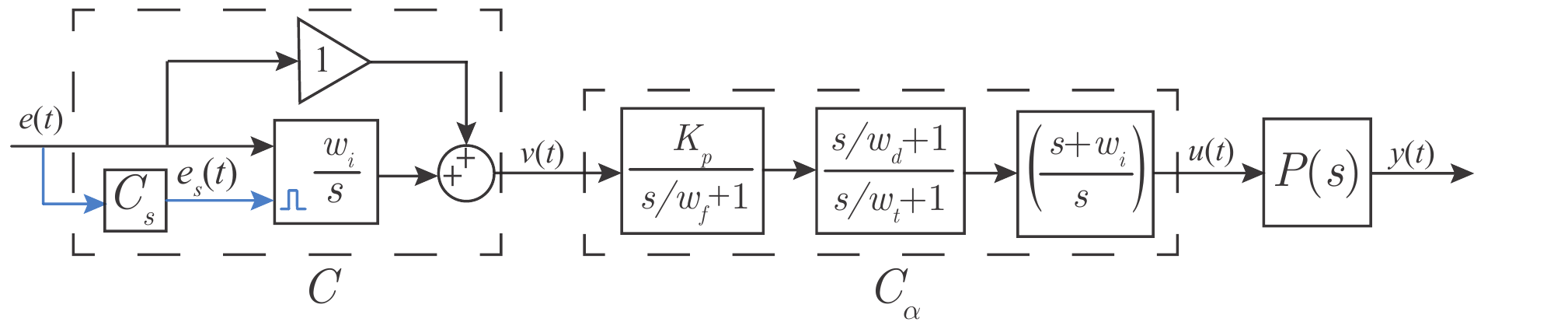}}
    \caption{The open-loop block diagram of the T-PCI-PID control system.}
	\label{Open-loop block diagrams of PCID RCS}
\end{figure}

Figure \ref{CL T-PCIID_et_4freqs} ($a_1$)-($d_1$) illustrates the simulated and Theorem \ref{thm: Method C, new HOSIDF}-predicted steady-state errors \(e(t)\) for the T-PCI-PID system. These simulations consider a reference input signal \(r(t) = 1E-7\sin(2\pi f t)\) at frequencies \(f\) of 5 Hz, 10 Hz, 100 Hz, and 500 Hz. Given that the reset action is amplitude-independent, we display the scaled reset triggered signals for these frequencies. The reset triggered signal \(e_s(t)\) forces two reset instants per steady-state period in multiple-reset systems at input frequencies 5 Hz and 10 Hz, while it maintains the reset instants unchanged in two-reset systems at input frequencies 100 Hz and 500 Hz. 

The analytical predictions closely align with the simulations in T-PCI-PID systems, validating the accuracy of Theorem \ref{thm: Method C, new HOSIDF}. This conclusion is further supported by the experimental results in Fig. \ref{CL T-PCIID_et_4freqs} ($a_2$)-($d_2$). Note that the jagged signal observed at input frequencies 5 Hz and 10 Hz in the measured results arises from sensor limitations.
\begin{figure}[h]
	\centering
	\includegraphics[width= \columnwidth]{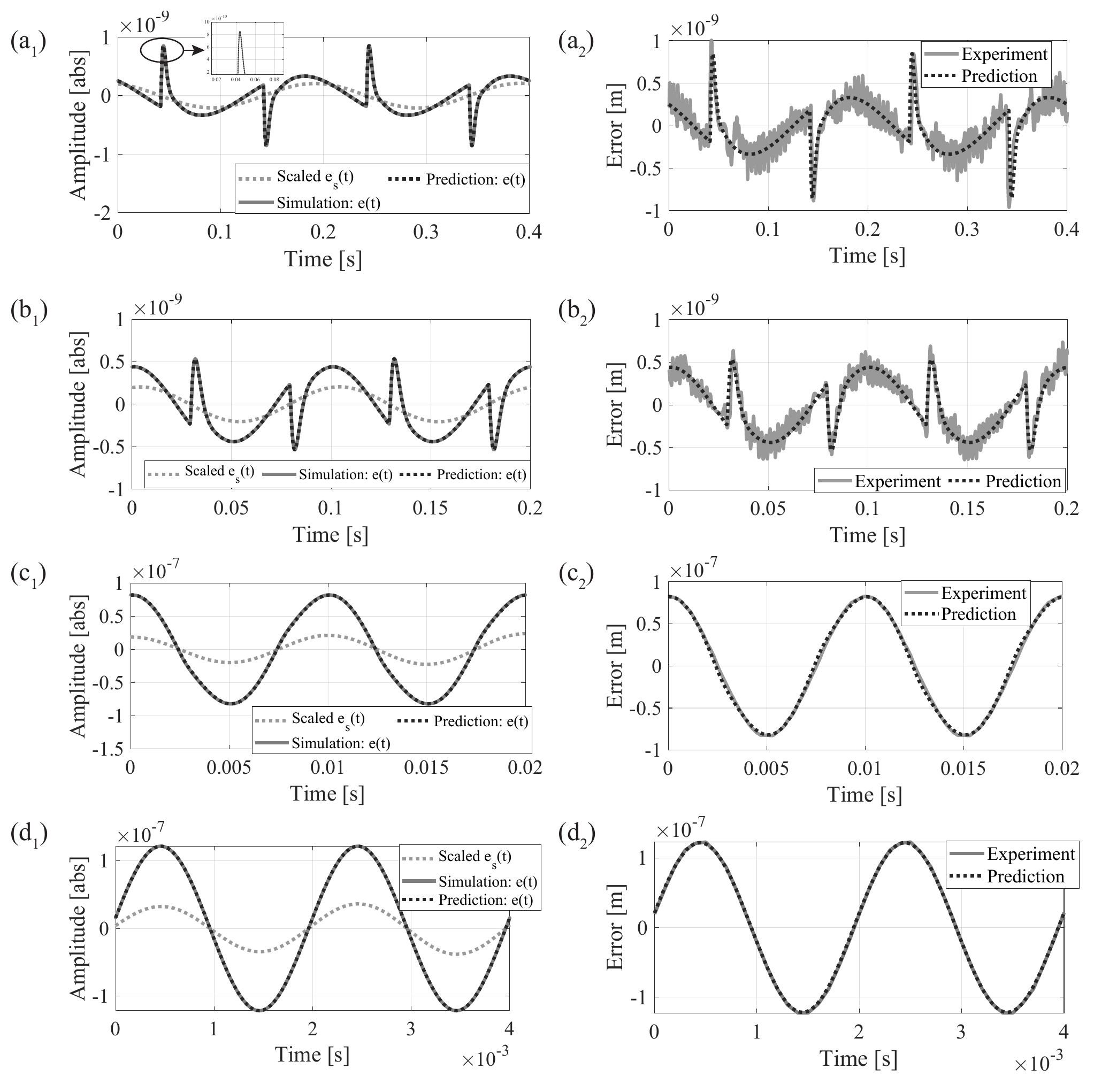}
	\caption{The steady-state error $e(t)$ for the T-PCI-PID system under sinusoidal input $r(t) = \sin(2\pi f t)$ at input frequencies $f$ of 5 Hz in ($a_1$) and ($a_2$), 10 Hz in ($b_1$) and ($b_2$), 100 Hz in ($c_1$) and ($c_2$), and 500 Hz in ($d_1$) and ($d_2$) as predicted, simulated, and experimentally measured. $e_s(t)$ represents the reset triggered signal.}
	\label{CL T-PCIID_et_4freqs}
\end{figure} 

Experimental results presented in Fig. \ref{CL T-PCIID_et_4freqs} not only validate the accuracy of Theorem \ref{thm: Method C, new HOSIDF} but also demonstrate the superior performance achieved by the T-PCI-PID system compared to the PCI-PID under sinusoidal inputs. The T-PCI-PID system effectively reduces low-frequency steady-state error while retaining high-frequency performance. For instance, under the same input signal at the input frequency of 10 Hz, the T-PCI-PID system in Fig. \ref{CL T-PCIID_et_4freqs}($b_2$) has a maximum steady-state error magnitude of 0.5E-9 [m], while the PCI-PID system in Fig. \ref{CL PCIID_et_3freqs}($b_2$) has a maximum steady-state error magnitude of 1E-9 [m]. However, since this paper primarily focuses on developing frequency response analysis for closed-loop reset systems, a comprehensive exploration of the T-RCS structure's capabilities across diverse situations will be addressed in future research. 
% The enhanced performance observed in the T-PCI-PID system is based on a single sinusoidal reference input in this case. Generalizing this new control structure to accommodate different reset elements and considering the impact of other input signals, including disturbances, noise, and step inputs, warrants further investigation. While the primary focus of this paper lies in introducing a frequency response analysis tool for closed-loop reset systems, future research endeavors into these aspects promise a more thorough comprehension of the T-RCS structure's capabilities across diverse situations.
% The experimental and simulated output signal $y(t)$ of T-PCI-PID control systems under a reference input signal $r(t) = 1E-8\sin (2\pi f t)$ [m] at four frequencies ($f$) are shown in Fig. \ref{CL T-PCIID_et_4freqs}($a_2$)-($d_2$). Results showcase an agreement between the experimental and simulated data, affirming the reliability of the simulation. Note that the comparison focuses on the output signal $y(t)$ rather than $e(t)$ due to the fact that, at lower frequencies such as 1 Hz, the amplitude of the error signal $e(t)$ becomes low and falls beyond the sensor's operational range, rendering the results less dependable.
% Subfigures ($a_2$)-($d_2$) compare the experimental (Exp) and simulated (Sim) \(y(t)\) for T-PCI-PID systems at these four frequencies.
\subsubsection{The Relation of the Accuracy of Theorem \ref{thm: Method C, new HOSIDF} and the Number of Harmonics}
Let \(e_{\text{sim}}(t)\) and \(e_{\text{pre}}(t)\) denote the simulated and predicted steady-state errors in the T-PCI-PID system, respectively. The Prediction Error (PE) between them is defined as PE = \(|e_{\text{sim}}(t) - e_{\text{pre}}(t)|\). 
\begin{figure}[h]
	\centering
	\includegraphics[width=0.92\columnwidth]{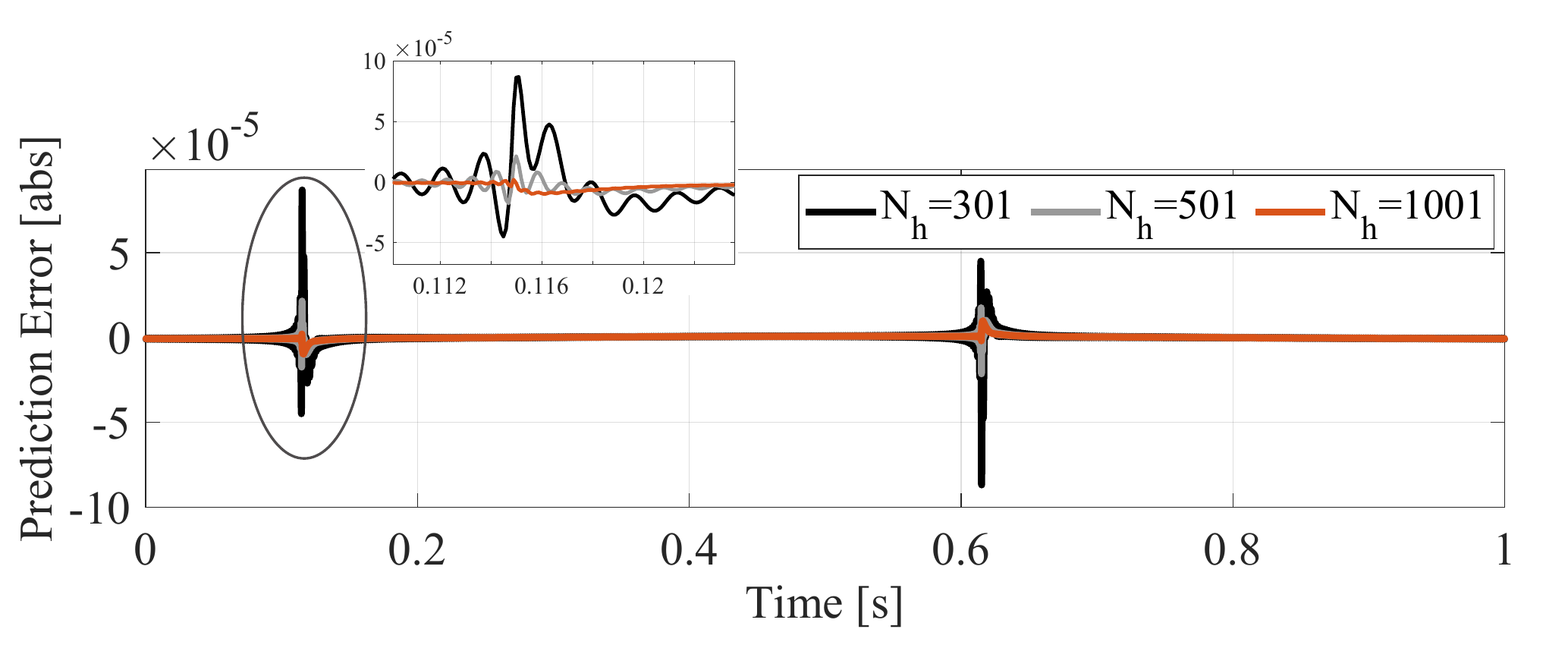}
	\caption{The PE of the T-PCI-PID control system under the input $r(t) = \sin(2\pi t)$ for different numbers of harmonics (\(N_h\)) included in the prediction calculation process, specifically 301, 501, and 1001.}
	\label{CL T-PCIID_predictionerror_Nh}
\end{figure} 

Figure \ref{CL T-PCIID_predictionerror_Nh} shows the relationship between the PE and the number of harmonics (denoted as \(N_h\)) included in the prediction calculation process, in the T-PCI-PID control system under a sinusoidal input signal $r(t) = \sin(2\pi t)$. As the number of harmonics \(N_h\) increases, the accuracy of the method increases. Ideally, PE approaches 0 as \(N_h\) tends to infinity. The zoom-in plot emphasizes the differences in predictions between the simulation and Theorem \ref{thm: Method C, new HOSIDF}, resembling vibration signals. This discrepancy arises because Theorem \ref{thm: Method C, new HOSIDF} calculates outputs by summing a finite number of sinusoidal harmonics. The unaccounted infinite harmonics contribute to the prediction error, resulting in a pattern reminiscent of vibrations. There is a trade-off between the PE and $N_h$ (indicating the calculation time in practice). Readers can strive to minimize PE while considering the computational time trade-off.
\subsection{Discussion: The Usability of Theorem \ref{thm: Method C, new HOSIDF} in the Reset Systems Design}
% The Bode plots of the open-loop transfer function \(\mathcal{L}_o(\omega)\) given in \eqref{eq: Lo} and 
In linear systems, the analytical connection between open-loop and closed-loop frequency-domain analysis serves as an effective tool for designing and predicting the performance of the systems. Let $\mathcal{L}_n(\omega)$ and $\mathcal{S}_n(\omega)$ represent the open-loop transfer function and the closed-loop sensitivity function, respectively. However, in reset systems, the relationship $\mathcal{S}_n(\omega) = 1/(\mathcal{L}_n(\omega) +1)$ does not hold, neither for the first-order harmonic (for $n=1$) nor higher-order harmonics (for $n>1$). The developed closed-loop sensitivity functions in \eqref{eq: sensitivity functions in CL} illustrate that the first and higher-order harmonics in the open-loop have a cross effect on the first and higher-order harmonics in the closed-loop, mediated by the parameter $\Gamma(\omega)$.

As highlighted in Remark \ref{rem: Gamma}, Method A and Method B assume $\Gamma(\omega) = 1$ for all $\omega$, which implies that higher-order harmonics $e_n(t)$ for $n > 1$ undergo no reset actions, leading to inaccuracies. Theorem \ref{thm: Method C, new HOSIDF} addresses this issue by introducing the parameter $\Gamma(\omega)$. The parameter $\Gamma(\omega)$ affects computations for both first-order and higher-order harmonics, as demonstrated in equations \eqref{eq: sensitivity functions in CL} to \eqref{eq: linear and nonlinear part of t_mc}. Its introduction provides a more accurate representation of the system's behavior, particularly in scenarios involving significant higher-order harmonics. Moreover, $\Gamma(\omega)$ can be used to design and tune reset systems to be less affected by higher-order harmonics, for improving system performance.

Figure \ref{pciid_Gamma_final} illustrates the $\Gamma(\omega)$ values in the T-PCI-PID system when subjected to sinusoidal inputs across frequencies ranging from 1 Hz to 1000 Hz. It is evident from the figure that $\Gamma(\omega) = 1$ does not hold across the entire frequency range. Given this context, comparing the new method with previous Methods A and B would be redundant, as these methods have already been shown to make inaccurate assumptions. To validate the accuracy of the new theorem, we directly compare its predictions with simulations and experiments in Section \ref{subsub: case 2}. 

Note that the depiction of $\Gamma(\omega)$ in Fig. \ref{pciid_Gamma_final} specifically pertains to the system described in case study 2. For different reset control systems, $\Gamma(\omega)$ will vary as determined by \eqref{eq: Gamma_thm}. 
\begin{figure}[h]
	\centering
	\includegraphics[width=0.92\columnwidth]{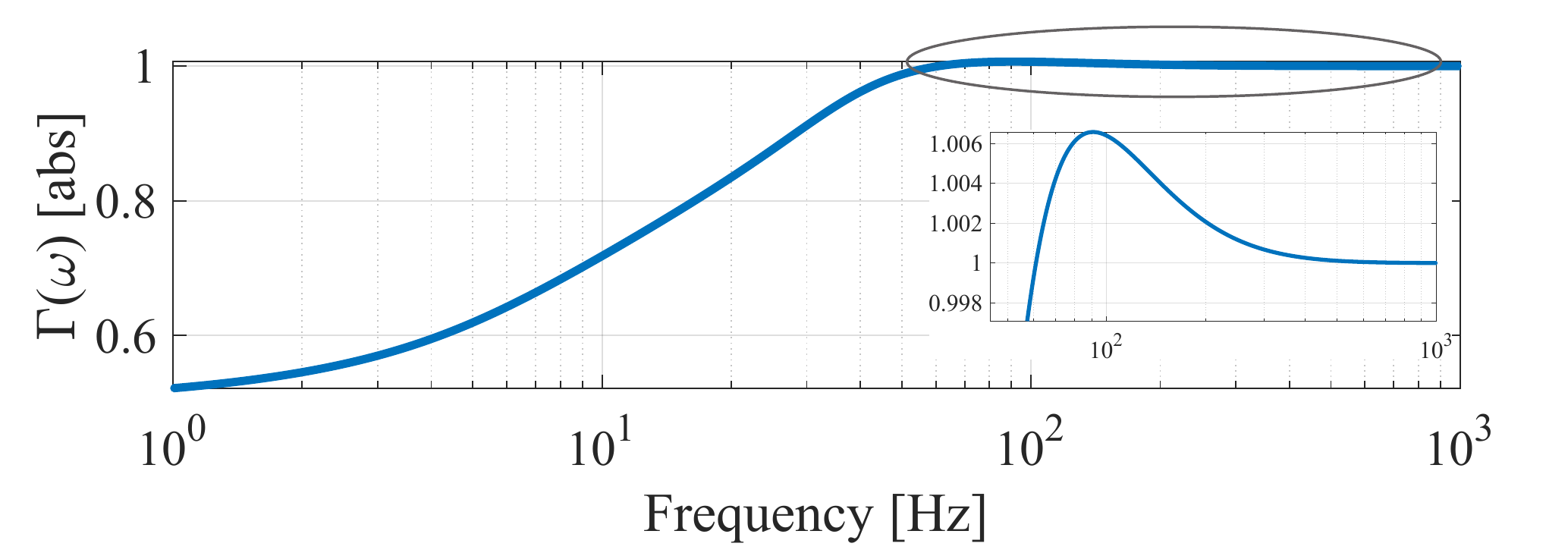}
	\caption{The $\Gamma(\omega)$ of the T-PCI-PID control system.}
	\label{pciid_Gamma_final}
\end{figure} 

Theorem \ref{thm: Method C, new HOSIDF} serves two primary purposes in the design of reset systems:\\
(1) Assessing the time-domain performance of the closed-loop system, for example evaluating its tracking precision at various frequencies, as demonstrated in Fig. \ref{CL PCIID_et_3freqs} and \ref{CL T-PCIID_et_4freqs}.\\
(2) Analyzing the frequency-domain characteristics of closed-loop reset systems by generating plots of the closed-loop HOSIDF using \eqref{eq: sensitivity functions in CL} to \eqref{eq: CS}, as exemplified in Fig. \ref{CL T-PCIID_open_closed_loop_BP}. 
\begin{figure}[h]
    \centering
    \centerline{\includegraphics[width=0.92\columnwidth]{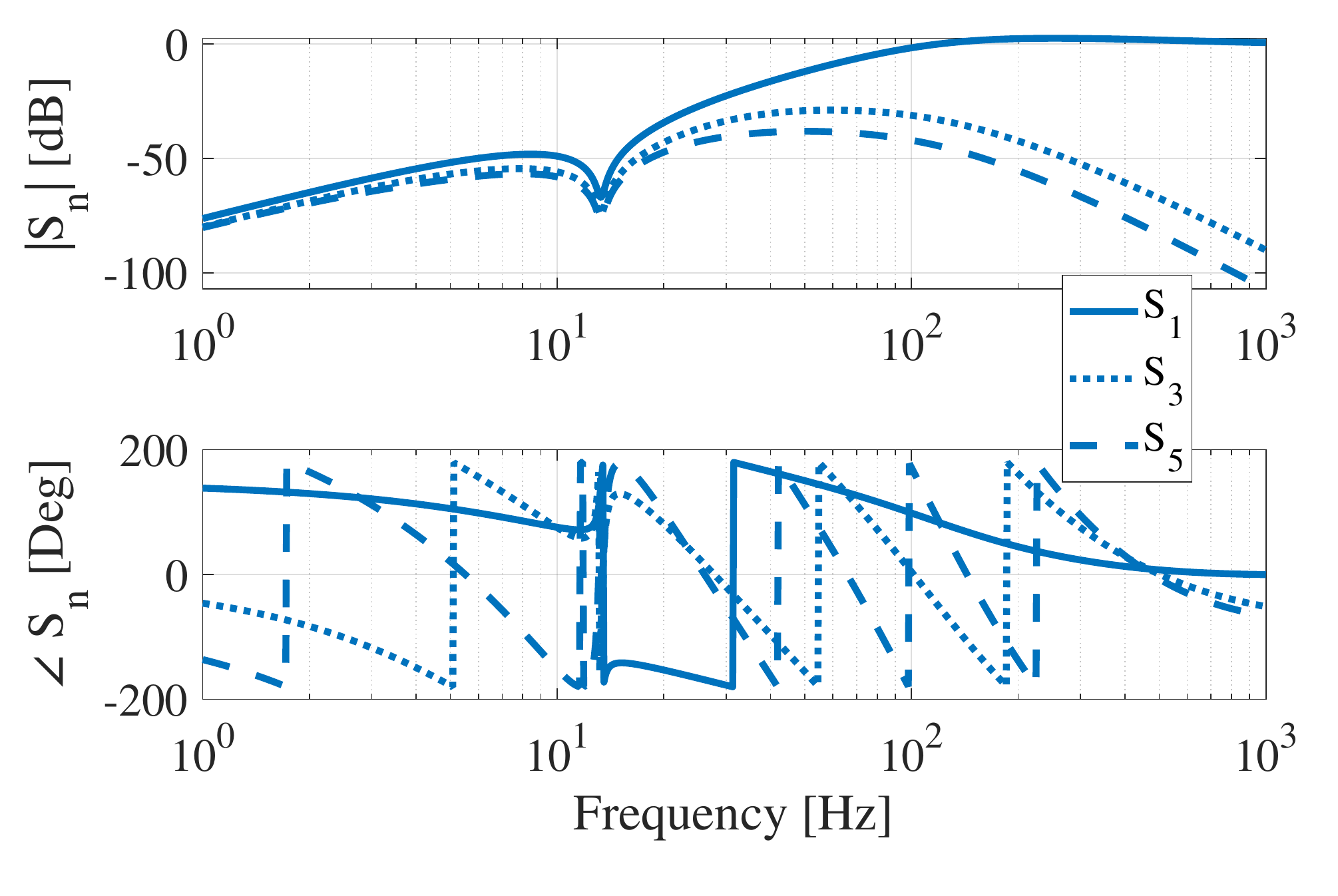}}
    \caption{The Bode plot of the closed-loop Higher Order Sinusoidal Input Sensitivity Function (HOSISF) $\mathcal{S}_n(\omega)$ for the closed-loop T-PCI-PID control system, with the first three harmonics \(\mathcal{S}_1\), \(\mathcal{S}_3\), and \(\mathcal{S}_5\).}
	\label{CL T-PCIID_open_closed_loop_BP}
\end{figure}\\
Theorems \ref{thm: vnl_npi}, \ref{thm: Open loop model for RC}, \ref{thm: Pulse-based model for RCS in closed loop}, and \ref{thm: Method C, new HOSIDF} collectively offer tools to analyze the behavior of reset controllers in both open-loop and closed-loop, facilitating the application of loop-shaping techniques in the reset systems design in the future. 

%% file: 5_Conclusion.tex
\section{Conclusion}
\label{sec:conclusion}
Reset control systems are effective in improving the performance of motion systems. To facilitate the practical design of reset systems, this study develops frequency response analysis methods for open-loop and closed-loop reset control systems, by assessing their steady-state responses to sinusoidal inputs. Results show the efficacy of the methods in predicting the precision of two-reset systems on precision motion stages. Moreover, the methods establish connections between open-loop and closed-loop analysis of reset systems. However, the paper primarily develops these analysis tools. Future research can explore practical applications of these frequency response analysis methods in designing reset control systems.

The frequency response analysis for closed-loop reset systems under sinusoidal disturbance and noise follows a similar derivation process as the theories presented in this paper. However, to emphasize and clarify the contribution of this paper, we have chosen not to include analysis for systems with disturbance or noise inputs here. Instead, we plan to address these aspects in our future research discussing disturbances.

The frequency response analysis is currently limited to two-reset systems. In our future research, we aim to develop techniques to identify two-reset systems and analyze multiple-reset systems, thereby expanding the scope of our analysis methods. Furthermore, the newly introduced Two-Reset Control System (T-RCS) in this research is designed to enforce two reset instants per steady-state cycle when the system is subjected to sinusoidal inputs. The T-RCS has shown improved steady-state tracking precision at low frequencies compared to traditional reset systems under sinusoidal reference inputs. The enhanced performance is attributed to the elimination of multiple-reset occurrences achieved by the T-RCS. The practical application of the T-RCS under various types of inputs are worth exploring in future research endeavors.

%% file: 6_Appendix.tex
\appendix
\section{Proof for Lemma \ref{lem: CI MODEL}}
\label{appen: proof for lemma1}
\vspace*{12pt}
\begin{proof} 
Consider a GCI subjected to a sinusoidal input signal $e(t) = |E_1|\sin (\omega t)$, under Assumption \ref{open-loop stability} and the Zeno-free condition in \eqref{eq: zeno-free}, at steady states.

Based on \eqref{eq: State-Space eq of RC}, the output signal $u_{ci}(t)$ of the GCI (with $A_R =0, B_R = 1, C_R = 1, D_R = 0$) is given by
\begin{equation}
\label{uci, t+}
\begin{cases}
\dot u_{ci}(t) = e(t), & e(t) \neq 0, \\
u_{ci}(t^+) = \gamma u_{ci}(t), & e(t) = 0.
\end{cases}
\end{equation}
For its BLS (which is an integrator), we have 
\begin{equation}
\label{eq: dot ui}
    \dot u_{i}(t) = e(t).
\end{equation}
Define $q_i(t) = u_{ci}(t) - u_{i}(t)$. From \eqref{uci, t+} and \eqref{eq: dot ui}, we have
	\begin{equation}
 \label{eq: qi(t)0}
		\begin{cases}
			\dot q_i(t) = \dot u_{ci}(t) - \dot u_{i}(t) = 0, & e(t) \neq 0,\\
			q_i(t^+) = \gamma q_i(t) + (\gamma - 1)u_{i}(t), & e(t) = 0.
		\end{cases}
	\end{equation}	
% Thus, $q_i(t)$ is a square wave with pulse locating at reset instants. \cite{buitenhuis2020frequency}
The reset instant of the GCI with a sinusoidal input signal $e(t) = |E_1|\sin(\omega t)$ is determined by $t_i = i{\pi}/{\omega}$, where $i \in \mathbb{Z}^{+}$ and $e(t_i) = 0$. Utilizing \eqref{eq: qi(t)0}, the signal $q_i(t)$ between two consecutive reset instants $[t_i^+, t_{i+1}]$ (where $t_i^+$ denotes the after-reset instant) is expressed as:
	\begin{equation}
		\label{eq: qi(t)1}
		q_i(t) = q_i(t_i^+) + \int_{t_i}^{t}\dot q_i(\tau)d\tau = q_i(t_i^+),\ t\in [t_i^+, t_{i+1}].
	\end{equation}
Given the input signal $e(t) = |E_1|\sin(\omega t)$ and based on the state-space equation in \eqref{eq: State-Space eq of RC}, at the reset instant $t_i=i{\pi}/{\omega}$, the base-linear output $u_i(t_i)$ is given by:
\begin{equation}
\label{uti}
   u_{i}(t_i) = \begin{cases}
       0, & \text{for even}\ i, \\
       {2|E_1|}/{\omega}, & \text{for odd}\ i.
   \end{cases} 
\end{equation}
Combining \eqref{eq: qi(t)0}, \eqref{eq: qi(t)1}, and \eqref{uti}, at the reset instant $t_i$, $q_i(t_i^+)$ is given by
\begin{equation}
\label{eq:qi00}
	\begin{aligned}
		q_i(t_i^+) &= \gamma q_i(t_{i-1}^+) + (\gamma - 1)u_{i}(t_i)\\
		&= \begin{cases}
			\gamma q_i(t_{i-1}^+), & \text{for even}\ i, \\
			\gamma q_i(t_{i-1}^+) + {2|E_1|(\gamma - 1)}/{\omega}, & \text{for odd}\ i.
		   \end{cases}
	\end{aligned}
\end{equation}
Based on \eqref{eq:qi00}, for an odd $i$, $q_i(t_i^+)$ is given by 
\begin{equation}
\label{eq:qi01}
	\begin{aligned}
		q_i(t_i^+) &= q_i(t_{i+2}^+)\\
		&= 	\gamma q_i(t_{i+1}^+) + {2|E_1|(\gamma - 1)}/{\omega}\\
		&= 	\gamma^2 q_i(t_{i}^+) + {2|E_1|(\gamma - 1)}/{\omega}.
	\end{aligned}
\end{equation}
Equations \eqref{eq:qi00} and \eqref{eq:qi01} can be concluded that
\begin{equation}
	\label{eq: t_i+}
	q_i(t_i^+) = 
	\begin{cases}
		-{2|E_1|\gamma (\gamma + 1)^{-1}}/{\omega}, & \text{for even}\ i, \\
		-{2|E_1|(\gamma + 1)^{-1}}/{\omega}, & \text{for odd}\ i.
	\end{cases}
\end{equation}
Combining \eqref{eq: qi(t)1} and \eqref{eq: t_i+}, $q_i(t)$ in the time domain can be derived as follows:
	\begin{equation}
		\label{eq: q(t)00}
		q_i(t) = 
		\begin{cases}
			-{2|E_1|\gamma (\gamma + 1)^{-1}}/{\omega}, & \text{for}\ t \in [2k, 2k+1)\cdot\pi/\omega,\\
			-{2|E_1|(\gamma + 1)^{-1}}/{\omega}, &  \text{for}\ t \in [2k+1, 2k+2)\cdot\pi/\omega.
		\end{cases}
	\end{equation}	
Here we conclude the proof.
\end{proof}
\section{Proof for Theorem \ref{thm: vnl_npi}}
\label{appen: proof for thm1}
\vspace*{12pt}
\begin{proof}
Consider a reset controller $\mathcal{C}$ \eqref{eq: State-Space eq of RC} where $n_r=1$ with an input signal of $e(t) = |E_n|\sin(n\omega t + \angle E_n) \ ( \angle E_s\in(-\pi,\pi),\ n=2k+1,\ k\in\mathbb{N})$ and a reset triggered signal of $e_s(t) = |E_s|\sin(\omega t+\angle E_s) \ (\angle E_s\in(-\pi,\pi))$, at steady states. Let $x_c(t)$ and $x_{bl}(t)$ denote the state and the base-linear state of the $\mathcal{C}$ and the BLC $\mathcal{C}_{bl}$, respectively. Define
\begin{equation}
\label{A: eq: xc,xnl,xbl}
    x_{nl}(t) = x_c(t) - x_{bl}(t)\ (\in \mathbb{R}^{n_c\times 1}),
\end{equation}
where $n_c$ is the number of states of the reset controller $\mathcal{C}$ and $x_{bl}(t)$ is given by
\begin{equation}
\label{A: X_bl(t)npi}
    x_{bl}(t) = |E_n\Delta_l(n\omega)|\sin(n\omega t + \angle E_n + \angle\Delta_l(n\omega)),
\end{equation}
where 
\begin{equation}
\label{eq: Delta_l(nw)}
 \Delta_l(n\omega) = (jn\omega I-A_R)^{-1}B_R\ (\in\mathbb{R}^{n_c\times 1}). 
\end{equation}
Based on \eqref{eq: State-Space eq of RC} and \eqref{A: eq: xc,xnl,xbl}, we have:
\begin{equation}
\label{A: eq: q(t)11}
    \begin{cases}
        \dot x_{nl}(t) = A_Rx_{nl}(t), &  e_s(t)\neq0,\\
        x_{nl}(t^+) = A_\rho x_{nl}(t) + (A_\rho-I)x_{bl}(t), & e_s(t) = 0.
    \end{cases}
\end{equation}
For the reset controller $\mathcal{C}$ with a reset triggered signal of $e_s(t) = |E_s|\sin(\omega t+\angle E_s)$, the set of reset instants is denoted by $J_o := \{t_i \mid t_i=(i\pi-\angle E_s)/\omega,\ i \in \mathbb{Z}^{+}\}$. The reset interval is given by $\sigma_i = t_{i+1}-t_i=\pi/\omega$. According to \eqref{A: eq: q(t)11}, between two consecutive reset instants $[t_{i}^+, t_{i+1}]$ where $e_s(t) \neq 0$, the expression for $x_{nl}(t)$ is given by
\begin{equation}
\label{A: xnl(t)}
    x_{nl}(t) = e^{A_R (t-t_i)}\Delta_n(\omega) , \ \text{ for } t\in[t_{i}^+, t_{i+1}],
\end{equation}
where $\Delta_n(\omega)\in \mathbb{R}^{n_c\times 1}$ is a constant matrix independent of time $t$.

Based on \eqref{A: xnl(t)}, at the reset instant $t_{i+1} \in J_o$, $x_{nl}(t_{i+1})$ is given by
\begin{equation}
\label{A: xnl(t_i+1)}
    x_{nl}(t_{i+1}) = e^{A_R (t_{i+1}-t_i)}\Delta_n(\omega)  =e^{A_R\pi/\omega} \Delta_n(\omega).
\end{equation}
From \eqref{A: eq: q(t)11}, $x_{nl}(t^+)$ at the reset instant $t_{i+1}$ is given by
\begin{equation}
\label{A: xnl(t_i+1+)}
    x_{nl}(t_{i+1}^+) = A_\rho x_{nl}(t_{i+1}) + (A_\rho-I)x_{bl}(t_{i+1}).
\end{equation}
Substituting $x_{nl}(t_{i+1})$ from \eqref{A: xnl(t_i+1)} into \eqref{A: xnl(t_i+1+)}, $x_{nl}(t_{i+1}^+)$ is obtained as
\begin{equation}
\label{A: xnl(t_i+1+)2}
    x_{nl}(t_{i+1}^+)  = A_\rho e^{A_R\pi/\omega} \Delta_n(\omega) + (A_\rho-I)x_{bl}(t_{i+1}).
\end{equation}
From \eqref{A: X_bl(t)npi}, at the reset instant $t_i= (i\pi-\angle E_s)/\omega$, the base-linear state is given by
\begin{equation}
\label{A: eq:x_bl(t_i)_npi}
\begin{aligned}
    x_{bl}(t_i) &= |E_n\Delta_l(n\omega)|\sin(ni\pi+ \angle E_n + \angle\Delta_l(n\omega)-n\angle E_s)\\
    &=
    \resizebox{0.85\hsize}{!}{$
        \begin{cases}
        |E_n\Delta_l(n\omega)|\sin(\angle\Delta_l(n\omega)+ \angle E_n-n\angle E_s), & \text{ for even }i,\\
        -|E_n\Delta_l(n\omega)|\sin(\angle\Delta_l(n\omega)+ \angle E_n-n\angle E_s), & \text{ for odd }i,
    \end{cases}    
    $}
\end{aligned}
\end{equation}
where $\Delta_l(n\omega)$ is given in \eqref{eq: Delta_l(nw)}.

Define 
\begin{equation}
\label{A: eq: Delta_cn}
\Delta_c^n(\omega) = |\Delta_l(n\omega)|\sin(\angle\Delta_l(n\omega)+ \angle E_n-n\angle E_s)\ (\in \mathbb{R}^{n_c\times 1}),   
\end{equation}
and substitute $\Delta_c^n(\omega)$ from \eqref{A: eq: Delta_cn} into \eqref{A: eq:x_bl(t_i)_npi}, we have
\begin{equation}
\label{A: eq:x_bl(t_i)_npi2}
    x_{bl}(t_i) = 
    \begin{cases}
        |E_n|\Delta_c^n(\omega), & \text{ for even }i,\\
        -|E_n|\Delta_c^n(\omega), & \text{ for odd }i.
    \end{cases}    
\end{equation}
From \eqref{A: eq:x_bl(t_i)_npi2}, at the reset instant $t_{i+1} = ((i+1)\pi-\angle E_s)/\omega$, the expression for $x_{bl}(t_{i+1})$ is written as:
\begin{equation}
\label{A: eq:xbl(ta+)1}
    x_{bl}(t_{i+1}) =
    \begin{cases}
    -|E_n|\Delta_c^n(\omega),\ &\text{for even }i,\\
    |E_n|\Delta_c^n(\omega),\ &\text{for odd }i,  
    \end{cases}    
\end{equation}
Substituting $x_{bl}(t_{i+1})$ from \eqref{A: eq:xbl(ta+)1} into \eqref{A: xnl(t_i+1+)2}, $x_{nl}(t_{i+1}^+)$ is derived as
\begin{equation}
\label{A: eq:xnl(ta+)0}
    \resizebox{1\hsize}{!}{$
    x_{nl}(t_{i+1}^+) =
    \begin{cases}
    A_\rho e^{A_R\pi/\omega}\Delta_n(\omega)  - (A_\rho-I)|E_n|\Delta_c^n(\omega),\ &\text{for even }i,\\
    A_\rho e^{A_R\pi/\omega}\Delta_n(\omega)  + (A_\rho-I)|E_n|\Delta_c^n(\omega),\ &\text{for odd }i.    
    \end{cases}
    $}
\end{equation}
From \eqref{A: eq:xnl(ta+)0}, we have
\begin{equation}
\label{A: eq:xnl(ta+)}
    \resizebox{1\hsize}{!}{$
    x_{nl}(t_{i}^+) =
    \begin{cases}
    A_\rho e^{A_R\pi/\omega}\Delta_n(\omega)  + (A_\rho-I)|E_n|\Delta_c^n(\omega),\ &\text{for even }i,\\
    A_\rho e^{A_R\pi/\omega}\Delta_n(\omega)  - (A_\rho-I)|E_n|\Delta_c^n(\omega),\ &\text{for odd }i.  
    \end{cases}
    $}
\end{equation}
From \eqref{A: xnl(t)}, at the reset instant $t_i^+$, we have
\begin{equation}
\label{A: x_nl_ti+}
    x_{nl}(t_i^+) = \Delta_n(\omega).
\end{equation}
Let equations \eqref{A: eq:xnl(ta+)} and \eqref{A: x_nl_ti+} be set equal to each other. Then, $\Delta_n(\omega)$ is derived as follows:
\begin{equation}
\label{A: eq:xnl(ta+)2}
    \Delta_n(\omega) = 
    \begin{cases}
     |E_n|\Delta_v^n(\omega),\ &\text{for even }i,\\
     -|E_n|\Delta_v^n(\omega),\ &\text{for odd }i,   
    \end{cases}
\end{equation}
where 
\begin{equation}
\label{A: eq: Delta_vn}
    \Delta_v^n(\omega) = (A_\rho e^{A_R\pi/\omega}-I)^{-1}(I-A_\rho)\Delta_c^n(\omega)\ (\in \mathbb{R}^{n_c\times 1}).
\end{equation}
Since the reset interval $\sigma_i=\pi/\omega$, equations \eqref{A: xnl(t)} and \eqref{A: eq:xnl(ta+)2} together illustrate that $x_{nl}(t)$ is bounded and has the same period $2\pi/\omega$ as the reset triggered signal $e_s(t)$. The absolute integrability of $x_{nl}(t)$ implies the existence of its Fourier and Laplace transforms.
% Define $t^* = t-t_i$. Equation \eqref{A: xnl(t)} is written as:
% \begin{equation}
% \label{A: xnl2}
%   x_{nl}(t) = e^{A_Rt^*}\Delta_n(\omega), \text{ for }t^*\in [0^+, \ t_{i+1}-t_{i}].
% \end{equation}
Let $X_{nl}(s)$ denote the Laplace transform of $x_{nl}(t)$. From \eqref{A: xnl(t)}, for $t\in[t_i^+, \ t_{i+1}]$, the Laplace transform of $x_{nl}(t)$ is expressed as:
\begin{equation}
\label{A: Xnl(s)}
X_{nl}(s) = \mathscr{L}[x_{nl}(t)] =  (sI-A_R)^{-1}e^{-A_Rt_i}\Delta_n(\omega).
\end{equation}
Define a parameter $Q(s) = \mathscr{L}[q(t)]$ as:
\begin{equation}
\label{A: eq: Q(s), Xnl(s)}
    Q(s) = (sI)^{-1}(sI-{A_R})X_{nl}(s).
\end{equation}
Then, based on \eqref{A: Xnl(s)} and \eqref{A: eq: Q(s), Xnl(s)}, for $t\in[t_i^+, \ t_{i+1}]$, $Q(s)$ is given by
\begin{equation}
\label{A: Q(s)}
Q(s) = (sI)^{-1}e^{-A_Rt_i}\Delta_n(\omega).
\end{equation}
From \eqref{A: eq:xnl(ta+)2} and \eqref{A: Q(s)}, during the time interval $[t_i^+, \ t_{i+1}]$, the inverse Laplace transform of $ Q(s)$ is given by:
\begin{equation}
\label{A: q(t)0}
    \resizebox{1\hsize}{!}{$
q(t) =  \mathscr{L}^{-1}[Q(s)]= 
    \begin{cases}
     |E_n|\Delta_v^n(\omega)u(t-t_i) + C_{\beta1},\ &\text{for even }i,\\
     -|E_n|\Delta_v^n(\omega)u(t-t_i) + C_{\beta2},\ &\text{for odd }i,    
    \end{cases}
    $}
\end{equation}
where $u(t)$ is an unit step signal. The parameters $C_{\beta1}\in \mathbb{R}^{n_c\times1}$ and $C_{\beta2}\in \mathbb{R}^{n_c\times1}$ are the values of $q(t)$ at the reset instant $t_i$. From \eqref{A: q(t)0}, we obtain that $q(t)$ remains a constant matrix during the time interval $[t_i^+, \ t_{i+1}]$.

From \eqref{A: xnl(t)}, at the time instant $t_{i+1}$, $x_{nl}(t_{i+1}) = e^{A_R\pi/\omega}\Delta_n(\omega)$. From \eqref{A: eq:xnl(ta+)2}, $x_{nl}(t_{i})$ is given by
\begin{equation}
\label{A: eq:xnl(ti) with deltan}
   x_{nl}(t_i) = 
    \begin{cases}
     -|E_n|e^{A_R\pi/\omega}\Delta_v^n(\omega),\ &\text{for even }i,\\
     |E_n|e^{A_R\pi/\omega}\Delta_v^n(\omega),\ &\text{for odd }i,    
    \end{cases}
\end{equation}
From \eqref{A: xnl(t)} and \eqref{A: eq:xnl(ta+)2}, at the time instant $t_{i}^+$, $x_{nl}(t_{i}^+)$ is given by
\begin{equation}
\label{A: eq:xnl(ti+) with deltan}
     x_{nl}(t_{i}^+)= 
     \begin{cases}
     |E_n|\Delta_v^n(\omega),\ &\text{for even }i,\\
     -|E_n|\Delta_v^n(\omega),\ &\text{for odd }i.    
    \end{cases}
\end{equation}
From the time instant $t_i$ to $t_i^+$, $x_{nl}(t_{i})$ jumps to $x_{nl}(t_{i}^+)$. From \eqref{A: eq:xnl(ti) with deltan} and \eqref{A: eq:xnl(ti+) with deltan}, this jump is given by
\begin{equation}
\label{A: xnl jump at ti}
    \resizebox{1\hsize}{!}{$  
\begin{aligned}
x_{nl}(t_{i}^+) - x_{nl}(t_{i}) &=  
    \begin{cases}    
    |E_n|(I+e^{A_R\pi/\omega})\Delta_v^n(\omega),\ &\text{for even }i,\\
    -|E_n|(I+e^{A_R\pi/\omega})\Delta_v^n(\omega),\ &\text{for odd }i.    
    \end{cases}
\end{aligned}
    $}
\end{equation}
Substituting $\Delta_v^n(\omega)$ from \eqref{A: eq: Delta_vn} into \eqref{A: xnl jump at ti}, we have
\begin{equation}
\label{A: xnl jump at ti2}
\begin{aligned}
&x_{nl}(t_{i}^+) - x_{nl}(t_{i}) = \\
    &\resizebox{1\hsize}{!}{$   
\begin{cases}
     |E_n|(I+e^{A_R\pi/\omega})(A_\rho e^{A_R\pi/\omega}-I)^{-1}(I-A_\rho)\Delta_c^n(\omega),\ &\text{for even }i,\\
     -|E_n|(I+e^{A_R\pi/\omega})(A_\rho e^{A_R\pi/\omega}-I)^{-1}(I-A_\rho)\Delta_c^n(\omega),\ &\text{for odd }i.  
    \end{cases}
    $}
\end{aligned}
\end{equation}
Define
\begin{equation}
\label{A: eq: Delta_q}
\resizebox{\columnwidth}{!}{$   
\begin{aligned}
\Delta_q^n(\omega) &= (I+e^{A_R \pi/\omega})(A_\rho e^{A_R\pi/\omega}-I)^{-1}(I-A_\rho)\Delta_c^n(\omega) \ (\in \mathbb{R}^{n_c\times 1})
\end{aligned}   
$}
\end{equation}
and substitute $\Delta_q^n(\omega)$ from \eqref{A: eq: Delta_q} into \eqref{A: xnl jump at ti2}, we have
\begin{equation}
\label{A: xnl jump at ti3}
x_{nl}(t_{i}^+) - x_{nl}(t_{i}) =  
\begin{cases}
     |E_n|\Delta_q^n(\omega),\ &\text{for even }i,\\
     -|E_n|\Delta_q^n(\omega),\ &\text{for odd }i.  
    \end{cases}
\end{equation}
This jump indicates that from the time instant $t_i$ to $t_i^+$, $x_{nl}(t)$ is an impulse signal denoted by $\zeta_{nl}(t)$. From \eqref{A: xnl jump at ti3}, $\zeta_{nl}(t)$ is expressed as:
\begin{equation}
\label{A: x{nl}(t{i}+) - x{nl}(t{i})2}
\begin{aligned}
    \zeta_{nl}(t) &= [x_{nl}(t_{i}^+) - x_{nl}(t_{i})]\delta(t-t_i) \\
    &=\begin{cases}
    |E_n|\Delta_q^n(\omega)\delta(t-t_i),\ &\text{for even }i,\\
    -|E_n|\Delta_q^n(\omega)\delta(t-t_i),\ &\text{for odd }i,   
    \end{cases}    
\end{aligned}
\end{equation}
where $\delta(t)$ represents the Dirac delta function.

Since the Laplace transform of the delta function is 1, the Laplace transform of the impulse signal $\zeta_{nl}(t)$ in \eqref{A: x{nl}(t{i}+) - x{nl}(t{i})2} is given by:
\begin{equation}
\label{A: Delta_nl(s)}
    \zeta_{nl}(s) =    
    \begin{cases}
    |E_n|\Delta_q^n(\omega)e^{-t_is},\ &\text{for even }i,\\
    -|E_n|\Delta_q^n(\omega)e^{-t_is},\ &\text{for odd }i.   
    \end{cases}
\end{equation}
The impulse response is defined as the response of a system to a Dirac delta input \cite{pesaran1998generalized}. Define a signal $\zeta_{q}(t)$ as the impulse response of the impulse signal $\zeta_{nl}(t)$ filtered through the linear time invariant (LTI) transfer function $(sI)^{-1}(sI-{A_R})$. From \eqref{A: Delta_nl(s)}, $\zeta_{q}(t)$ is given by:
\begin{equation}
\label{A: zeta_q1}
\begin{aligned}
    \zeta_q(t) &= \mathscr{L}^{-1}[(sI)^{-1}(sI-{A_R})\zeta_{nl}(s)]\\
    &=\begin{cases}
    |E_1|\Delta_q^n(\omega)\delta(t-t_i)-|E_n|A_R\Delta_q^n(\omega),\ &\text{for even }i,\\
    -|E_1|\Delta_q^n(\omega)\delta(t-t_i)+|E_n|A_R\Delta_q^n(\omega),\ &\text{for odd }i.   
    \end{cases}    
\end{aligned}
\end{equation}
Equation \eqref{A: zeta_q1} demonstrates that at the time instant $t_i$, $\zeta_q(t_i)$ has an initial value of $\pm |E_n|A_R\Delta_q^n(\omega)$ and undergoes a jump of:
\begin{equation}
\label{A: zeta q jump}
    \zeta_q(t_i^+)-\zeta_q(t_i) = 
    \begin{cases}
    |E_n|\Delta_q^n(\omega),\ &\text{for even }i,\\
    -|E_n|\Delta_q^n(\omega),\ &\text{for odd }i.   
    \end{cases}
\end{equation}
From \eqref{A: eq: Q(s), Xnl(s)}, the signal $x_{nl}(t)$ filtered by $(sI)^{-1}(sI-{A_R})$ generates the signal $q(t)$. Equations \eqref{A: zeta_q1} and \eqref{A: zeta q jump} demonstrate that at the time instant $t_i$ to $t_i^+$, the signal $x_{nl}(t_i)$ jumps to $x_{nl}(t_i^+)$. The jump $\zeta_{nl}(t)$, filtered by the transfer function $(sI)^{-1}(sI-{A_R})$, generates a jump $\zeta_q(t_i^+)-\zeta_q(t_i)$ in \eqref{A: zeta q jump}. Thus, from \eqref{A: eq: Q(s), Xnl(s)}, \eqref{A: zeta_q1}, and \eqref{A: zeta q jump}, the jump in $q(t) = \mathscr{L}^{-1}[Q(s)]$ at the time instant $t_i$ is given by:
\begin{equation}
\label{A: zeta_q2}
    q(t_i^+) - q(t_i) = \zeta_{q}(t_i^+) - \zeta_{q}(t_i).
\end{equation}
Based on \eqref{A: zeta q jump} and \eqref{A: zeta_q2}, at the time instant $t_i$, we have
\begin{equation}
\label{A: q jump}
    q(t_i^+) = 
    \begin{cases}
    q(t_i) + |E_n|\Delta_q^n(\omega),\ &\text{for even }i,\\
    q(t_i)  -|E_n|\Delta_q^n(\omega),\ &\text{for odd }i.   
    \end{cases}
\end{equation}
% From \eqref{A: xnl(t)}, \eqref{A: eq:xnl(ta+)2} and \eqref{A: x{nl}(t{i}+) - x{nl}(t{i})2}, during the interval $[t_i^+,\ t_{i+1}^+]$, $x_{nl}(t)$ is given by
% \begin{equation}
% \label{A: eq:xnl2}
% x_{nl}(t) = 
% \begin{cases}
%     \Delta_v^n(\omega)e^{t-t_i} -\Delta_q^n(\omega)\delta(t_{i+1}) ,\ &\text{for even }i,\\
%     -\Delta_v^n(\omega)e^{t-t_i} +\Delta_q^n(\omega)\delta(t_{i+1}),\ &\text{for odd }i, 
%     \end{cases}
% \end{equation}
% where $\delta(t)$ is a Dirac delta function.
% The nonlinear signal $x_{nl}(t)$ does not directly result from the signal $r(t)$ but from the reset action. To generate the nonlinear signal $x_{nl}(t)$ defined in \eqref{A: eq: xc,xnl,xbl}, we introduce a signal $q(t)$, which can be considered a disturbance introduced to the system. This process is illustrated in the block diagram for $\mathcal{C}$ in Fig. \ref{fig: State Space of Reset Controller}(b). In this new block diagram, the state $x_c(t)$ has two elements: $x_{bl}(t)$ resulting from the reference $r(t)$, and $x_{nl}(t)$ resulting from the input $q(t)$.
% \begin{equation}
% \label{A: eq:xnl3}
% x_{nl}(t) = 
% \begin{cases}
%     \Delta_v^n(\omega)e^{t^*} -\Delta_q^n(\omega)\delta(t_{i+1}) ,\ &\text{for even }i,\\
%     -\Delta_v^n(\omega)e^{t^*} +\Delta_q^n(\omega)\delta(t_{i+1}),\ &\text{for odd }i. 
%     \end{cases}
% \end{equation}
From \eqref{A: q(t)0}, $q(t_i^+)$ for even $i$ is given by
\begin{equation}
\label{A: qt_i+1}
q(t_i^+) =  |E_n|\Delta_v^n(\omega) + C_{\beta1}.
\end{equation}
From \eqref{A: q jump}, at the time instant $t_i$ for even $i$, $q(t_i^+) = q(t_i) + |E_n|\Delta_q^n(\omega)$. Then, based on \eqref{A: q(t)0} and \eqref{A: q jump}, $q(t_i^+)$ for even $i$ can be written as
\begin{equation}
\label{A: qt_i+2}
q(t_i^+) = -|E_n|\Delta_v^n(\omega) + C_{\beta2} + |E_n|\Delta_q^n(\omega). 
\end{equation}
By setting equations \eqref{A: qt_i+1} and \eqref{A: qt_i+2} equal to each other, we get:
\begin{equation}
\label{A: eq:Cq1}
C_{\beta1} =  C_{\beta2} -2|E_n|\Delta_v^n(\omega) + |E_n|\Delta_q^n(\omega).
\end{equation}
Substituting $C_{\beta1}$ from \eqref{A: eq:Cq1} to \eqref{A: q(t)0}, $q(t)$ is given by
\begin{equation}
\label{A: q(t)1}
\begin{aligned}
&q(t) =  \\
&\resizebox{1\hsize}{!}{$  
\begin{cases}
     -|E_n|\Delta_v^n(\omega)u(t-t_{2k}) + |E_n|\Delta_q^n(\omega) + C_{\beta2}, &\text{ for } t\in[t_{2k}^+, \ t_{2k+1}],\\
     -|E_n|\Delta_v^n(\omega)u(t-t_{2k+1})  + C_{\beta2}, &\text{ for } t\in[t_{2k+1}^+, \ t_{2k+2}],
    \end{cases}  
    $}  
\end{aligned}
\end{equation}
where $k\in\mathbb{N}$.

Let $f(t)_k\ (k\in\mathbb{Z}^+)$ represent the $k$-th state of a function $f(t)$. Considering the reset controller $\mathcal{C}$ \eqref{eq: State-Space eq of RC} with resetting the first state where $n_r=1$, the first state $x_c(t_i)_1$ in the state $x_c(t_i)$ is reset to $\gamma x_c(t_i)_1$, and there are no reset actions in state $x_c(t)_k$ for $k>1$. Therefore, we have
\begin{equation}
\label{A: k>1,x_c}
    x_c(t)_k = x_{bl}(t)_k, \text{ for } k>1.
\end{equation}
From \eqref{A: eq: xc,xnl,xbl}, \eqref{A: x{nl}(t{i}+) - x{nl}(t{i})2}, and \eqref{A: k>1,x_c}, we have
\begin{equation}
\label{A: k>1,x_nl}
\begin{aligned}
   x_{nl}(t)_k &= 0, \text{ for } k>1,\\
   \Delta_q^n(\omega)_k &= 0, \text{ for } k>1.   
\end{aligned}
\end{equation}
From \eqref{A: eq: Q(s), Xnl(s)} and \eqref{A: k>1,x_nl}, we have
\begin{equation}
\label{A: k>1,q}
  q(t)_k = 0, \text{ for } k>1. 
\end{equation}
Equation \eqref{A: q(t)1} indicates that during the time interval $[t_{i}^+, \ t_{i+1}]$, $q(t)$ is a constant matrix. Equation \eqref{A: q jump} demonstrates that at the time instant $t_i \in J_o$, $q(t)$ has a jump of $\pm|E_n|\Delta_q^n(\omega)$. Therefore, the signal $q(t)$ is absolutely integrable, ensuring the existence of its Fourier and Laplace transforms. Since $\sigma_i=t_{i+1}-t_i=\pi/\omega$, from \eqref{A: q jump}, \eqref{A: q(t)1}, and \eqref{A: k>1,q}, the first state of $q(t)$ denoted by $q(t)_1$ is a square wave with a period of ${2\pi}/{\omega}$. It shares the same phase as the reset triggered signal $e_s(t)$ and has a magnitude of $|E_n|\Delta_q^n(\omega)_1$. 
% Thus, $q(t)$ is written as $q(t)=[q(t)_1, 0,...,0]^T \in \mathbb{R}^{n_c\times1}$, where $n_c$ is the number of states in $x_c(t)$. 
Define a normalized square wave signal $q_s(t)$ with the same phase as the reset-triggered signal $e_s(t)$, expressed as follows:
\begin{equation}
	\label{A: eq: q0}
	q_s(t) = \frac{4}{\pi}\sum\nolimits_{\mu=1}^{\infty} \frac{\sin(\mu\omega t+\mu\angle E_s)}{\mu},\ \mu = 2k+1 (k\in\mathbb{N}),
\end{equation}
whose Fourier transform is given by
\begin{equation}
\label{A: eq: Q_s(w)}
Q_s(\omega)=\mathscr{F}[q_s(t)]  =  4\sum\nolimits_{\mu=1}^{\infty} \mathscr{F}[\sin(\mu\omega t+\mu\angle E_s)]/{(\mu\pi)}. 
\end{equation}
% where $T_R = [1, 0,...,0]^T \in \mathbb{N}^{n_c\times 1}$.
According to \eqref{A: q(t)1}, \eqref{A: k>1,q} and \eqref{A: eq: q0}, $q(t)$ can be expressed as:
\begin{equation}
\label{A: q(t)}
    q(t) = |E_n|\Delta_q^n(\omega)q_s(t)/2 -|E_n|\Delta_v^n(\omega) + C_{\beta2} + |E_n|\Delta_q^n(\omega)/2.
\end{equation}
For a time-domain signal $q(t)$, the frequency $\omega$ is a constant. A constant function corresponds to a delta function in the frequency domain. Defining $Q(\omega)=\mathscr{F}[q(t)]$, we derive $Q(\omega)$ from \eqref{A: q(t)} as follows:
\begin{equation}
\label{A: Q(w)0}
\begin{aligned}
Q(\omega) = |E_n|\Delta_q^n(\omega) Q_s(\omega)&/2 +  (|E_n|\Delta_q^n(\omega)/2\\
&-\Delta_v^n(\omega) + C_{\beta2}) \delta(\omega).        
\end{aligned}
\end{equation}
Since $\delta(\omega)$ is a Dirac delta function, which is zero for $\omega \neq 0$, equation \eqref{A: Q(w)0} is simplified as: 
\begin{equation}
\label{A: Q(w)}
 Q(\omega)= |E_n|\Delta_q^n(\omega) Q_s(\omega)/2.    
\end{equation}
From \eqref{A: eq: Q_s(w)} and \eqref{A: Q(w)}, $ Q(\omega)$ is given by
\begin{equation}
\label{A: eq: Q(W)2}
Q(\omega)= 2|E_n|\Delta_q^n(\omega)\sum\nolimits_{\mu=1}^{\infty} \mathscr{F}[\sin(\mu\omega t+\mu\angle E_s)]/{(\mu\pi)}.     
\end{equation}
Both $Q(\omega)$ \eqref{A: eq: Q(W)2} and $Q_s(\omega)$ \eqref{A: eq: Q_s(w)} contain $\mu$ harmonics. Let $Q^\mu(\omega)$ and $Q_s^\mu(\omega)$ represent the $\mu$-th ($\mu\in\mathbb{Z}^+$) harmonics of $Q(\omega)$ and $Q_s(\omega)$, respectively. They are expressed as:
\begin{equation}
\label{A: Qn,Q,Q0n,Q0}
\begin{aligned} 
Q(\omega) &= \sum\nolimits_{\mu=1}^{\infty}Q^\mu(\omega),\\    
Q_s(\omega) &= \sum\nolimits_{\mu=1}^{\infty}Q_s^\mu(\omega).    
\end{aligned}
\end{equation}
From \eqref{A: eq: Q_s(w)}, \eqref{A: eq: Q(W)2}, and \eqref{A: Qn,Q,Q0n,Q0}, $Q_s^\mu(\omega)$ and $Q^\mu(\omega)$ are given by
\begin{equation}
\label{A: Q_0n}
\begin{aligned}
Q_s^\mu(\omega) &= {4\mathscr{F}[\sin(\mu\omega t+\mu\angle E_s)]}/{(\mu\pi)},\\
Q^\mu(\omega) &= {2|E_n|\Delta_q^n(\omega)\mathscr{F}[\sin(\mu\omega t+\mu\angle E_s)]}/{(\mu\pi)}.
\end{aligned}    
\end{equation}
% Given that the reset controller in Fig. \ref{fig: State Space of Reset Controller}(a) shares the same reset-triggered signal $e(t) = |E_n|\sin (\omega t)$ as the open-loop CI discussed in Lemma \ref{lem: CI MODEL}, it can be inferred that $q(t)$ is a square wave with identical phases but distinct magnitudes compared to $q_i(t)$ in \eqref{A: eq: qi(t)}. Let $Q(\omega)=\mathscr{F}[q(t)]$ and suppose that $M_q(\omega) = Q(\omega)/Q_i(\omega)$. Then, substituting $Q_i(\omega)$ in \eqref{A: eq: Q_i(w)} into $Q(\omega) = M_q(\omega)Q_i(\omega)$, $Q(\omega)$ is obtained as
% \begin{equation}
% \label{A: eq: Q(w)}
%     Q(\omega) = 2M_q(\omega)M_i(\omega)/\pi \cdot \sum\nolimits_{\mu=1}^{\infty}Q_{0}^{n}(n\omega).
% \end{equation}
Let $X_{nl}(\omega) = \mathscr{F}[x_{nl}(t)]$, and $X_{nl}^\mu(\omega)$ represents the $\mu$-th harmonic of $X_{nl}^\mu(\omega)$, defined as:
\begin{equation}
\label{A: Xnl, Xnl_n}
 X_{nl}(\omega) = \sum\nolimits_{\mu=1}^{\infty} X_{nl}^\mu(\omega).   
\end{equation}
From \eqref{A: eq: Q(s), Xnl(s)}, \eqref{A: Q_0n} and \eqref{A: Xnl, Xnl_n}, $X_{nl}^\mu(\omega)$ is derived as
\begin{equation}
\label{A: X_{nl}(w)}
\begin{aligned}   
    X_{nl}^\mu(\omega) &= (j\mu\omega I-{A_R})^{-1}j\mu\omega IQ^\mu(\omega).
\end{aligned}
\end{equation} 
According to \eqref{A: eq: xc,xnl,xbl}, \eqref{A: Xnl, Xnl_n}, and \eqref{A: X_{nl}(w)}, $X_c(\omega) = \mathscr{F}[x_c(t)]$ is obtained as
\begin{equation}
\label{A: Xc(w)}
    X_c(\omega) = X_{bl}(\omega) + \sum\nolimits_{\mu=1}^{\infty} (j\mu\omega I-{A_R})^{-1}j\mu\omega IQ^\mu(\omega),
\end{equation}
where $X_{bl}(\omega) = \mathscr{F}[x_{bl}(t)] $.

Let $V(\omega) = \mathscr{F}[v(t)]$. According to \eqref{eq: State-Space eq of RC}, $V(\omega)$ is given by
\begin{equation}
\label{A: V(w)00}
    V(\omega) = C_RX_c(\omega) + D_RE(\omega).
\end{equation}
Substituting $X_c(\omega)$ from \eqref{A: Xc(w)} into \eqref{A: V(w)00}, $V(\omega)$ is given by
\begin{equation}
\label{A: V(w)01}
\begin{aligned}
V(\omega) = &C_RX_{bl}(\omega) + D_RE(\omega)+ \\
& \sum\nolimits_{\mu=1}^{\infty} C_R(j\mu\omega I-{A_R})^{-1}j\mu\omega IQ^\mu(\omega).
\end{aligned}
\end{equation}
% Note that given $C_R\in\mathbb{R}^{1\times n_c}$, $A_R\in\mathbb{R}^{n_c\times n_c}$, and $M\in\mathbb{R}^{n_c\times 1}$ in \eqref{A: V(w)01}, the term $M$ followed by $E_{1\mu}(\omega)$ can be canceled.
From \eqref{eq: State-Space eq of RC}, we have 
\begin{equation}
\label{A: Vbl(w)}
 V_{bl}(\omega) = C_{bl}(\omega)E(\omega) = C_RX_{bl}(\omega) + D_RE(\omega).  
\end{equation}
Then, substituting $V_{bl}(\omega)$ from \eqref{A: Vbl(w)} into \eqref{A: V(w)01}, $V(\omega)$ is given by
\begin{equation}
\label{A: V(w)02}
    V(\omega) = V_{bl}(\omega) + \sum\nolimits_{\mu=1}^{\infty} C_R(j\mu\omega I-{A_R})^{-1}j\mu\omega IQ^\mu(\omega).
\end{equation}
Define 
\begin{equation}
 \label{A: eq:R_delta}
\begin{aligned}
    V_{nl}(\omega) &= \sum\nolimits_{\mu=1}^{\infty} V_{q}^\mu(\omega),\\
    V_{q}^\mu(\omega) &= \Delta_x(\mu\omega)Q^\mu(\omega),\\
    \Delta_x (\mu\omega) &= C_R(j\mu\omega I-A_R)^{-1}j\mu\omega I\  (\in\mathbb{R}^{1\times n_c}).
 \end{aligned}
\end{equation}
and substitute $V_{q}^\mu(\omega)$ from \eqref{A: eq:R_delta} into \eqref{A: V(w)02}, $V(\omega)$ is given by
\begin{equation}
\label{A: V(w)03}
    V(\omega) = V_{bl}(\omega) + V_{nl}(\omega) =V_{bl}(\omega) + \sum\nolimits_{\mu=1}^{\infty}V_{q}^\mu(\omega).
\end{equation}
% from \eqref{A: eq: Delta_q} and \eqref{A: eq:R_delta}
In the time domain, equation \eqref{A: V(w)03} is given by
\begin{equation}
\label{A: eq: v(t)}
    v(t) = v_{bl}(t) + v_{nl}(t).
\end{equation}
From \eqref{A: eq:R_delta}, $v_{nl}(t)$ is given by
\begin{equation}
\begin{aligned}
\label{A: eq: vnl(t)}
v_{nl}(t) &= \mathscr{F}^{-1}\bigg[\sum\nolimits_{\mu=1}^{\infty}\Delta_x(\mu\omega)Q^\mu(\omega)\bigg],\ \mu\in\mathbb{Z}^+,
\end{aligned}
\end{equation}
where $Q^\mu(\omega)$ and $\Delta_x(\mu\omega)$ are defined in \eqref{A: eq:R_delta} and \eqref{A: Q_0n}, respectively. Equations \eqref{A: eq: v(t)} and \eqref{A: eq: vnl(t)} conclude the proof.
\end{proof}

\section{Proof for Theorem \ref{thm: Open loop model for RC}}
\label{appen: proof for thm2}
\vspace*{12pt}
\begin{proof}
The reset controller $\mathcal{C}$, operating with the input signal and reset-triggered signal $e(t) = |E_1|\sin(\omega t+\angle E_1)$ at steady states is defined as the reset controller discussed in Theorem \ref{thm: vnl_npi} when $e_s(t)=e(t)$. 

Following a similar proof process as in Appendix \ref{appen: proof for thm1}, consider that there are $n\in\mathbb{N}$ harmonics in the reset output signal $v(t)$. Referring to \eqref{A: V(w)03}, let $V(\omega) = \sum\nolimits_{n=1}^{\infty} V_n(\omega)$, we obtain $V_n(\omega)$ as
\begin{equation}
\label{eq:Vn}
    V_n(\omega) =
    \begin{cases}
    V_{bl}(\omega) + V_{q}^1(\omega),\ &\text{for }n=1,\\
    V_{q}^n(\omega),\ &\text{for odd } n>1,\\    
    0,\ &\text{for even } n\geq2.
    \end{cases}
\end{equation}
By applying the \enquote{Virtual Harmonic Generator} \cite{saikumar2021loop}, the input signal $e(t)$ generates $n$ harmonics $e_{1n}(t) = |E_1|\sin(n\omega t +n\angle E_1)$, whose Fourier transform is $E_{1n}(\omega) = |E_1|\mathscr{F}[\sin(n\omega t +n\angle E_1)]$. From \eqref{A: eq:R_delta}, $\mathcal{C}_{nl}(n\omega)$ is defined as
\begin{equation}
\label{Cnl(w)}
     \mathcal{C}_{nl}(n\omega)= \frac{V_{q}^n(\omega)}{E_{1n}(\omega)} =  {2}\Delta_x (n\omega)\Delta_q(\omega) /{(n\pi)}.
\end{equation}
From \eqref{A: Vbl(w)}, \eqref{eq:Vn}, and \eqref{Cnl(w)}, the $n$-th HOSIDF for $\mathcal{C}$, denoted as $\mathcal{C}_n(\omega)$, is defined as
\begin{equation}
\label{C_n}
    \mathcal{C}_n(\omega) = \frac{V_n(\omega)}{E_{1n}(\omega)}=
    \begin{cases}
    \mathcal{C}_{bl}(\omega) + \mathcal{C}_{nl}(\omega),\ &\text{for }n=1,\\
    \mathcal{C}_{nl}(n\omega),\ &\text{for odd } n>1,\\    
    0,\ &\text{for even } n\geq2.
    \end{cases}
\end{equation}
This concludes the proof. 
\end{proof}
\section{Proof for Corollary \ref{cor: v,vnl}}
\label{appen: proof for cor1}
\vspace*{12pt}
\begin{proof}
Define
\begin{equation}   
\label{Vnl(w)}
    V_{nl}(\omega) = \sum\nolimits_{n=1}^{\infty} V_{q}^n(\omega).
\end{equation}
From \eqref{Cnl(w)}, we have
\begin{equation}
\label{vqn2}
 V_{q}^n(\omega) = E_{1n}(\omega)\mathcal{C}_{nl}(n\omega).
\end{equation}
From \eqref{Vnl(w)} and \eqref{vqn2}, the inverse Fourier transform of $V_{nl}(\omega)$ is given by
\begin{equation}
\label{eq: vnl-11}
    v_{nl}(t) = \sum\nolimits_{n=1}^{\infty}\mathscr{F}^{-1}[E_{1n}(\omega)\mathcal{C}_{nl}(n\omega)].
\end{equation}
From \eqref{A: V(w)03} and \eqref{Vnl(w)}, $V(\omega)$ given by
\begin{equation}
 V(\omega) =  V_{bl}(\omega) + V_{nl}(\omega).
\end{equation}
The inverse Fourier transform of $V(\omega)$ is given in \eqref{eq:vt=vbl+vnl}. This concludes the proof.
\end{proof}
% \bmsubsection{Subsection title of first appendix\label{A: app1.1a}}
\section{Proof for Theorem \ref{thm: Pulse-based model for RCS in closed loop}}
\label{appen: proof for thm3}
\vspace*{12pt}
\begin{proof}

Consider a closed-loop reset control system (with reset controller $\mathcal{C}$ \eqref{eq: State-Space eq of RC} where $n_r = 1$), as depicted in Fig. \ref{fig1:RC system} with a sinusoidal reference input signal $r(t) = |R|\sin(\omega t)$ and under Assumptions \ref{assum: stable} and \ref{assum:2reset}, at steady states. 

The Fourier transform of \eqref{vl,vnl} is given by
\begin{equation}
\label{Vl(w),Vnl(w)}
\begin{aligned}
V_{l}(\omega) &=  \sum\nolimits_{n=1}^{\infty}V_{bl}^{n}(\omega),\\
V_{nl}(\omega) &= \sum\nolimits_{n=1}^{\infty}V_{nl}^n(\omega).    
\end{aligned}
\end{equation}
% From \eqref{eq: vl} and \eqref{vt,vl,vnl}, \(V(\omega)\) is given by
% \begin{equation}
% \label{eq:V(omega)}
%      V(\omega) = \sum\nolimits_{n=1}^{\infty} [E_n(\omega)\mathcal{C}_{bl}(n\omega) +  V_{nl}^{n}(\omega)],
% \end{equation}
% where $V_{nl}^{n}(\omega) = \mathscr{F}[v_{nl}^{n}(t)]$.
Based on \eqref{eq: vl} and \eqref{Vl(w),Vnl(w)}, the signal $v_l(t)$ in \eqref{vt,vl,vnl} is given by
\begin{equation}
\label{vlt}
    v_{l}(t) = \sum\nolimits_{n=1}^{\infty} \mathscr{F}^{-1}[E_n(\omega)\mathcal{C}_{bl}(n\omega)].
\end{equation}
The following section outlines the derivation of the signal $v_{nl}(t)$ in \eqref{vt,vl,vnl}. This process commences with the derivation of its first-order harmonic $v^1_{nl}(t)$. According to \eqref{eq:vt=vbl+vnl}, the first $\mathcal{C}$ under an input signal and the reset-triggered signal of $e_1(t) = |E_1|\sin(\omega t+ \angle E_1)$, yields the base-linear output denoted by $v_{bl}^{1}(t)$ and the nonlinear output signal denoted by $v_{nl}^{1}(t)$. To generate the signal $v^1_{nl}(t)$, the \enquote{Virtual Harmonic Generator} is employed to generate harmonics $e_{1n}(t)$ from $e_1(t)$, as illustrated in Fig. \ref{fig: Model for Closed-loop RCS}(a), as given by:
\begin{equation}
\label{eq:e_1n}
    \begin{aligned}
        e_{1n}(t) &= |E_1|\sin(n\omega t + n\angle E_1),\ n =2k+1 (k \in \mathbb{N}),
    \end{aligned}
\end{equation}
whose Fourier transform is $E_{1n}(\omega) = \mathscr{F}[e_{1n}(t)]$.

The signal $v_{bl}^{1}(t)$ is derived by \eqref{eq: vl}. From \eqref{Cnl(w)}, the signal $v_{nl}^{1}(t)$ is given by:
\begin{equation}
\label{eq: unln}
    \begin{aligned}
       v_{nl}^{1}(t) &= \sum\nolimits_{n=1}^{\infty}\mathscr{F}^{-1}[E_{1n}(\omega)\mathcal{C}_{nl}(n\omega)].
    \end{aligned}
\end{equation}
In Fig. \ref{fig: Model for Closed-loop RCS}(a), consider a reset controller $\mathcal{C}$ with a sinusoidal input signal $e_n(t) = |E_n|\sin(n\omega t + \angle E_n)$ and a reset triggered signal $e_1(t) = |E_1|\sin(\omega t+\angle E_1)$. From \eqref{eq: uw2}, $V^n(\omega)$ is given by
\begin{equation}
V^n(\omega) = V_{bl}^n(\omega) + V_{nl}^n(\omega).    
\end{equation}
From \eqref{eq:Vnl(t)npi1}, $V_{nl}^n(\omega)$ is given by
\begin{equation}
\label{eq: closed Vnl1}
\begin{aligned}
V_{nl}^n(\omega) &=\sum\nolimits_{\mu=1}^{\infty}\Delta_x(\mu\omega)Q^\mu(\omega),
\end{aligned}
\end{equation}
where $\Delta_x(\mu\omega)$ and $Q^\mu(\omega)$ are given in \eqref{eq:Vnl(t)npi2}.
% \begin{equation}
% \label{eq: closed Vnl2}
% \begin{aligned}
% \Delta_l(n\omega) &=  (jn\omega I-A_R)^{-1}B_R,\\ 
% \Delta_x(\mu\omega) &= C_R(j\mu\omega I-A_R)^{-1}j\mu\omega I,\\
% \Delta_c^n(\omega) &= |\Delta_l(n\omega)| \sin(\angle \Delta_l(n\omega)+ \angle E_n -n\angle E_1 ),\\   
% Q^\mu(\omega) &= 2|E_n|\Delta_q^n(\omega) \mathscr{F}[\sin(\mu\omega t + \mu\angle E_1)]/(\mu\pi),\\
% \Delta_q^n(\omega) &= (I+e^{A_R\pi/\omega})(A_\rho e^{A_R\pi/\omega}-I)^{-1}(I-A_\rho)\Delta_c^n(\omega).
% \end{aligned}
% \end{equation}

From \eqref{eq: closed Vnl1}, $V_{nl}^n(\omega)$, for a constant $\omega$ and varying $n$, exhibit the same phase as the reset-triggered signal $e_1(t)$. Considering that $V_{nl}^n(\omega)$ exhibits the same phase for different \(n\), as indicated in \eqref{Vl(w),Vnl(w)}, we introduce a function \(\Gamma(\omega)\) to express the ratio of \(V_{nl}(\omega)\) to \(V^{1}_{nl}(\omega)\):
\begin{equation}
		\label{eq: G0}
		\Gamma(\omega)=\frac{V_{nl}(\omega)}{V^{1}_{nl}(\omega)} =\frac{\sum\nolimits_{n=1}^{\infty}V^{n}_{nl}(\omega)}{V^{1}_{nl}(\omega)}.
\end{equation}
% Theorem \ref{thm: Open loop model for RC} illustrates that $v_{nl}^{n}(t) $is characterized as a filtered pulse signal, with pulses occurring at reset instants $t_m$. Thus, $\Gamma(\omega)$ can be formulated as
% \begin{equation}
% 		\label{eq: G}
%      \Gamma(\omega)= \frac{\sum\nolimits_{n=1}^{\infty}v_{nl}^{n}(t_m^+)}{v^{1}_{nl}(t_m^+)}.
% \end{equation}
From \eqref{eq: unln} and \eqref{eq: G0}, $V_{nl}(\omega)$ is given by
\begin{equation}
\label{eq: vnl}
    \begin{aligned}
V_{nl}(\omega) &= \sum\nolimits_{n=1}^{\infty}\Gamma(\omega)E_{1n}(\omega)\mathcal{C}_{nl}(n\omega). 
    \end{aligned}
\end{equation}
From \eqref{Vl(w),Vnl(w)} and \eqref{eq: vnl}, we have
\begin{equation}
\label{eq: vnl_n}
V_{nl}^n(\omega) = \Gamma(\omega)E_{1n}(\omega)\mathcal{C}_{nl}(n\omega).    
\end{equation}
Equations \eqref{vt,vl,vnl}, \eqref{vlt}, and \eqref{eq: vnl} describe the new block diagram for the closed-loop RCS, presented in Fig. \ref{fig: Model for Closed-loop RCS}(b). In the new block diagram, the unknown parameter \(\Gamma(\omega)\) will be elucidated in the subsequent derivation.

By substituting $V_{nl}^n(\omega)$ from \eqref{eq:Vnl(t)npi2} and \eqref{eq: closed Vnl1} into \eqref{eq: G0}, $\Gamma(\omega)$ is simplified as
% Referring to \eqref{Q(w)}, for an open-loop reset controller $\mathcal{C}$ under the input signal of $e_{n}(t)$, the Fourier transform of $q_n(t)$ is given by:
% \begin{equation}
% \label{Q_n}
%     Q_n(\omega) = |E_n|\Delta_q^n(\omega)Q_0(\omega)/2,
% \end{equation}
% where  $\Delta_q^n(\omega)$ is derived from \eqref{eq:xnl(ta+)2} and \eqref{x{nl}(t{i}+) - x{nl}(t{i})2}, given by
% \begin{equation}
% \label{Cqn}
%     \Delta_q^n(\omega) = (A_\rho e^{A_R\pi/\omega}-I)^{-1}(I-A_\rho) \Delta_c^n(\omega). 
% \end{equation}
% Since each reset controller $\mathcal{C}$ shares the same reset-triggered signal $e_1(t)$, $v_{nl}^n(t)$ shares the same phase. Based on \eqref{x{nl}(t{i}+) - x{nl}(t{i})2}, let the magnitude of the square wave added to the state $x_n(t)$ as $\Delta_q^n(\omega)$. 
\begin{equation}
		\label{eq: G2}
\Gamma(\omega)=\frac{\sum\nolimits_{n=1}^{\infty}|E_n|\Delta_c^n(\omega)}{|E_1|\Delta_c^1(\omega)}.
\end{equation}
As illustrated in Fig. \ref{fig: Model for Closed-loop RCS}(b), in the closed-loop configuration, we have
\begin{equation}
\label{eq:E_n}
    E_n(\omega) = - V_n(\omega)\mathcal{C}_\alpha(n\omega)\mathcal{P}(n\omega).
\end{equation}
Substituting \eqref{eq: vl} and \eqref{eq: vnl_n} into \eqref{eq: uw2}, $V_{n}(\omega)$ is given by
\begin{equation}
\label{eq: Un}
    V_n(\omega) = E_n(\omega)\mathcal{C}_{bl}(n\omega) + \Gamma(\omega)E_{1n}(\omega)\mathcal{C}_{nl}(n\omega).
\end{equation}
Substituting \eqref{eq: Un} into \eqref{eq:E_n}, we have
\begin{equation}
\label{eq: En2}
    E_n(\omega) = - E_n(\omega)\mathcal{L}_{bl}(n\omega) - \Gamma(\omega)E_{1n}(\omega)\mathcal{L}_{nl}(n\omega),
\end{equation}
where
\begin{equation}
\label{eq:Lbl,Lnl}
    \begin{aligned}
    \mathcal{L}_{bl}(n\omega) &= \mathcal{C}_{bl}(n\omega)\mathcal{C}_{\alpha}(n\omega)\mathcal{P}(n\omega),\\
    \mathcal{L}_{nl}(n\omega) &= \mathcal{C}_{nl}(n\omega)\mathcal{C}_{\alpha}(n\omega)\mathcal{P}(n\omega).
    \end{aligned}
\end{equation}
Based on the definitions of $e_n(t)$ and $e_{1n}(t)$ provided in \eqref{eq: e,y,u} and \eqref{eq:e_1n}, we express \eqref{eq: En2} in the time domain as follows:
\begin{equation}
\label{eq:phase and Mag}
\begin{aligned}
&|E_n||1+\mathcal{L}_{bl}(n\omega)| \sin(n\omega t + \angle E_n + \angle (1+\mathcal{L}_{bl}(n\omega))) =\\
&-\Gamma(\omega)|E_1||\mathcal{L}_{nl}(n\omega)| \sin(n\omega t + n\angle E_1 + \angle\mathcal{L}_{nl}(n\omega)).
\end{aligned}
\end{equation}
From \eqref{eq:phase and Mag} and the given condition $|E_n|>0$, we can deduce the following equations:
\begin{equation}
\label{eq: phase and Mag2}
\begin{aligned}
|E_n| &=  \frac{\Gamma(\omega)|\mathcal{L}_{nl}(n\omega)|}{|1+\mathcal{L}_{bl}(n\omega)|}|E_1|,\\
\angle E_n &= n\pi+ n\angle E_1 + \angle\mathcal{L}_{nl}(n\omega) - \angle (1+\mathcal{L}_{bl}(n\omega)), \\
&\indent\indent\indent\indent\indent\indent\indent\indent\indent\text{ for } n=2k+1>1.    
\end{aligned}
\end{equation}
From \eqref{eq:Vnl(t)npi2}, $\Delta_c^n(\omega)$ is given by
\begin{equation}
\label{Delta_cn2}
\Delta_c^n(\omega) =|\Delta_l(n\omega)| \sin(\angle \Delta_l(n\omega)+ \angle E_n -n\angle E_1 ).
\end{equation}
% \begin{equation}
% \label{Delta_cn2}
% \Delta_c^n(\omega) =
% \begin{cases}
% |\Delta_l(\omega)| \sin(\angle \Delta_l(\omega)), & \text{ for } n=1,\\
% |\Delta_l(n\omega)| \sin(\angle \Delta_l(n\omega) +  \angle E_n - n\angle E_1 ), & \text{ for } n > 1.     
% \end{cases}
% \end{equation}
Substituting the relation between $\angle E_1$ and $\angle E_n (n>1)$ from \eqref{eq: phase and Mag2} into \eqref{Delta_cn2}, we obtain:
\begin{enumerate}
    \item For $n=1$,\\
    \begin{equation}
    \label{Delta_cn31}
     \Delta_c^1(\omega) =   |\Delta_l(\omega)| \sin(\angle \Delta_l(\omega)).
     \end{equation}
     \item  For $n>1$,\\
     \begin{equation}
    \label{Delta_cn32}
     \begin{aligned}
     \Delta_c^n(\omega) = - &|\Delta_l(n\omega)| \sin(\angle \Delta_l(n\omega) +  \\&\angle\mathcal{L}_{nl}(n\omega)- \angle (1+\mathcal{L}_{bl}(n\omega))). 
     \end{aligned}
     \end{equation}
\end{enumerate}
% \begin{equation}
% \label{Delta_cn3} 
% \resizebox{1\hsize}{!}{$
% \Delta_c^n(\omega) = 
% \begin{cases}
% |\Delta_l(\omega)| \sin(\angle \Delta_l(\omega)), & \text{ for } n=1,\\ 
% - |\Delta_l(n\omega)| \sin(\angle \Delta_l(n\omega) + \angle\mathcal{L}_{nl}(n\omega) - \angle (1+\mathcal{L}_{bl}(n\omega))),& \text{ for }n>1.    
% \end{cases}  
% $}
% \end{equation}
Substituting $\Delta_c^n(\omega)$ from \eqref{Delta_cn31} and \eqref{Delta_cn32} into \eqref{eq: G2}, we have
\begin{equation}
\label{eq: G4}
\Gamma(\omega)=1+\frac{\sum\nolimits_{n=3}^{\infty}|E_n|\Delta_c^n(\omega)}{|E_1|\Delta_c^1(\omega)}.
\end{equation}
% \begin{equation}
% \label{eq: G4}
% \Gamma(\omega)=1-\frac{\sum\nolimits_{n=3}^{\infty}|E_n\Delta_l(n\omega)| \sin(\angle \Delta_l(n\omega) + \angle\mathcal{L}_{nl}(n\omega) - \angle (1+\mathcal{L}_{bl}(n\omega)))}{|E_1\Delta_l(\omega)| \sin(\angle \Delta_l(\omega))}.
% \end{equation}
% Since the harmonic $n = 2k+1, k\in \mathbb{N}$ is an odd number, the selection of $m\in\mathbb{Z}^+$ does not impact the value of $\sin((nm+1)\pi)/\sin(m\pi)$. Equation \eqref{eq: G3} can be simplified as
% \begin{equation}
% \label{eq: G4}
% \Gamma(\omega)=1-\frac{\sum\nolimits_{n=3}^{\infty}|E_n\Delta_l(n\omega)| \sin(\angle\mathcal{L}_{nl}(n\omega) - \angle (1+\mathcal{L}_{bl}(n\omega)) + \angle \Delta_l(n\omega))}{|E_1\Delta_l(\omega)| \sin(\angle \Delta_l(\omega))}.
% \end{equation}
Define
\begin{equation}
\label{eta}
    \begin{aligned}        
	\Psi_n(\omega) &= {|\mathcal{L}_{nl}(n\omega)|}/{|1+\mathcal{L}_{bl}(n\omega)|}
        % \Omega_1(\omega) &= -|\Delta_l(\omega)| \sin(\angle \Delta_l(\omega)),\\
        % \Omega_n(\omega) &= -|\Delta_l(n\omega)| \sin(\angle\mathcal{L}_{nl}(n\omega) - \angle (1+\mathcal{L}_{bl}(n\omega)) + \angle \Delta_l(n\omega)),
    \end{aligned}
\end{equation}
and substitute $\Psi_n(\omega)$ from \eqref{eta} into \eqref{eq: phase and Mag2}, we then obtain:
\begin{equation}
	\label{eq: Rn2}
	|E_n| = \Gamma(\omega)\Psi_n(\omega)|E_1|, \text{ for }n>1.
\end{equation}
Substituting \eqref{eta} and \eqref{eq: Rn2} into \eqref{eq: G4}, $\Gamma(\omega)$ is given by
\begin{equation}
\label{eq: G5}
\Gamma(\omega)=1+\Gamma(\omega)\frac{\sum\nolimits_{n=3}^{\infty}\Psi_n(\omega)\Delta_c^n(\omega)}{\Delta_c^1(\omega)}.
\end{equation}
% where $\Delta_c(\omega)=|\Delta_l(\omega)| \sin(\angle \Delta_l(\omega))$ is defined in \eqref{Cnl_final}.
 
Derived from \eqref{eq: G5}, $\Gamma(\omega)$ is obtained as below:
\begin{equation}
		\label{eq: gamma_final}		
		\begin{aligned}
		  \Gamma(\omega) &= 1/\left(1-{\sum\nolimits_{n=3}^{\infty}\Psi_n(\omega)\Delta_c^n(\omega)}/{\Delta_c^1(\omega)}\right).
		\end{aligned}
	\end{equation}
 Here, $\Gamma(\omega)$ is derived and the proof of Theorem \ref{thm: Pulse-based model for RCS in closed loop} is concluded.
\end{proof}

\section{Proof for Theorem \ref{thm: Method C, new HOSIDF}}
\label{appen: Proof for thm4}
\vspace*{12pt}
\begin{proof}
Consider a closed-loop SISO reset control system in Fig. \ref{fig1:RC system} with a reset controller $\mathcal{C}$ \eqref{eq: State-Space eq of RC} (where $n_r=1$) and to a sinusoidal reference input signal $r(t) = |R|\sin(\omega t)$, complying with Assumptions \ref{assum: stable} and \ref{assum:2reset}, at steady states.

From the block diagram for the closed-loop reset system in Fig. \ref{fig: Model for Closed-loop RCS}(b), we can express the first harmonic of the output $Y(\omega)$ as $Y_1(\omega)$, given by
\begin{equation}
\label{eq:Y1(w)}
\begin{aligned}
    Y_1(\omega) = E_1(\omega)[\mathcal{L}_{bl}(\omega) + \Gamma(\omega)\mathcal{L}_{nl}(\omega)],
\end{aligned}
\end{equation}
where $\mathcal{L}_{bl}(n\omega)$ and $\mathcal{L}_{nl}(n\omega)$ are given in \eqref{eq:Lbl,Lnl}.

Let $R_n(\omega) = |R|\mathscr{F}[\sin(n\omega t)]$, in the closed loop, we have
\begin{equation}
\label{eq:E1(w)}
    Y_{1}(\omega) = R_1(\omega)-E_{1}(\omega).
\end{equation}
Combining \eqref{eq:Y1(w)} and \eqref{eq:E1(w)}, the first order sensitivity function for the closed-loop reset system, denoted as $\mathcal{S}_1(\omega)$ is given by 
\begin{equation}
\label{eq: E_1}
		\mathcal{S}_1(\omega)= \frac{E_1(\omega)}{R(\omega)} = \frac{1}{1+\mathcal{L}_{bl}(\omega)+\Gamma(\omega)\mathcal{L}_{nl}(\omega)}.
\end{equation}
Define 
\begin{equation}
    \label{eq: Lo}
	\mathcal{L}_{o}(n\omega) = \mathcal{L}_{bl}(n\omega)+\Gamma(\omega)\mathcal{L}_{nl}(n\omega).
\end{equation}
Substituting \eqref{eq: Lo} into \eqref{eq: E_1}, $\mathcal{S}_1(\omega)$ is given by
\begin{equation}
    \label{eq: S1}
\mathcal{S}_1(\omega)= \frac{1}{1+\mathcal{L}_{o}(\omega)}.
\end{equation}
From \eqref{eq:e_1n} and \eqref{eq: E_1}, we obtain
\begin{equation}
\label{eq:E_1n}
    E_{1n}(\omega) = |\mathcal{S}_{1}(\omega)|e^{jn\angle \mathcal{S}_{1}(\omega)}R_n(\omega),
\end{equation}
where $R_n(\omega) = |R|\mathscr{F}[\sin(n\omega t)]$.

Combining \eqref{eq: En2} and \eqref{eq:E_1n}, the $n$-th ($(n = 2k+1,\ k\in \mathbb{N})$) sensitivity function for the closed-loop reset system, denoted by $\mathcal{S}_n(\omega)$ is given by
\begin{equation}
\label{Sn}
	\begin{aligned}
		\mathcal{S}_n(\omega) &= \frac{E_n(\omega)}{R_n(\omega)}= \frac{-\Gamma(\omega)\mathcal{L}_{nl}(n\omega)}{1+\mathcal{L}_{bl}(n\omega)}\frac{E_{1n}(\omega)}{R_n(\omega)} \\
		&= -\frac{\Gamma(\omega)\mathcal{L}_{nl}(n\omega)|\mathcal{S}_{1}(\omega)|e^{jn\angle \mathcal{S}_{1}(\omega)}}{1+\mathcal{L}_{bl}(n\omega)}.
	\end{aligned}
\end{equation}
According to Theorems \ref{thm: Open loop model for RC} and \ref{thm: Pulse-based model for RCS in closed loop}, the even harmonics of $\mathcal{S}_n(\omega)$ are zeros. The n$^{\text{th}}$ complementary sensitivity function $\mathcal{T}_n(\omega)$ in \eqref{eq: Complementary sensitivity functions in CL} and the control sensitivity function $\mathcal{CS}_n(\omega)$ in \eqref{eq: CS} can be derived using the same method from \eqref{eq:Y1(w)} to \eqref{Sn}. This concludes the proof. 
\end{proof}